\pdfoutput=1
\def\AVCOMBINED{1}
\documentclass[ejsv2,noshowframe,preprint]{imsart}

\newif\ifavmain
\newif\ifavsupp
\newif\ifavbibdone \avbibdonefalse
\ifdefined\AVCOMBINED
  \avmaintrue
  \avsupptrue
\else\ifdefined\AVSUPPLEMENTONLY
  \avmainfalse
  \avsupptrue
\else
  \avmaintrue
  \avsuppfalse
\fi\fi

\RequirePackage[authoryear,round]{natbib}
\RequirePackage[colorlinks,citecolor=blue,linkcolor=blue,urlcolor=blue]{hyperref}
\RequirePackage{graphicx}
\RequirePackage{mathtools}
\RequirePackage{booktabs}
\RequirePackage{array}
\RequirePackage{algorithm}
\RequirePackage{algpseudocode}
\RequirePackage{placeins}
\RequirePackage{caption}
\RequirePackage{microtype}
\RequirePackage{xcolor}
\RequirePackage{url}

\IfFileExists{main.tex}{
  
}{
  
}

\startlocaldefs

\theoremstyle{plain}
\newtheorem{theorem}{Theorem}[section]

\newtheorem{corollary}[theorem]{Corollary}
\newtheorem{lemma}[theorem]{Lemma}
\newtheorem{proposition}[theorem]{Proposition}
\theoremstyle{definition}
\newtheorem{definition}[theorem]{Definition}
\newtheorem{assumption}{Assumption}[section]
\theoremstyle{remark}
\newtheorem{remark}[theorem]{Remark}

\allowdisplaybreaks
\DeclareMathOperator{\KL}{KL}
\DeclareMathOperator*{\argmin}{arg\,min}
\DeclareMathOperator*{\argmax}{arg\,max}
\newcommand{\E}{\mathbb{E}}
\renewcommand{\Pr}{\mathbb{P}}
\newcommand{\F}{\mathcal{F}}
\newcommand{\hashid}[1]{{\ttfamily\footnotesize #1}}
\newcommand{\avmainnumset}[2]{\expandafter\def\csname avmainnum@#1\endcsname{#2}}
\avmainnumset{sec:model}{2.1}
\avmainnumset{sec:setup}{2.2}
\avmainnumset{sec:sota}{2.3}
\avmainnumset{sec:method}{3}
\avmainnumset{sec:procedure}{3.1}
\avmainnumset{sec:object}{3.2}
\avmainnumset{sec:calibration}{3.3}
\avmainnumset{sec:closure}{3.4}
\avmainnumset{sec:theory}{4}
\avmainnumset{sec:validity}{4.1}
\avmainnumset{sec:sufficiency}{4.2}
\avmainnumset{sec:opt}{4.3}
\avmainnumset{sec:regret}{4.4}
\avmainnumset{sec:ceiling}{4.5}
\avmainnumset{sec:boundaryshape}{4.6}
\avmainnumset{sec:phase}{5}
\avmainnumset{sec:corners}{5.1}
\avmainnumset{sec:ebonf}{5.2}
\avmainnumset{sec:dpvalidation}{5.3}
\avmainnumset{sec:sims}{6}
\avmainnumset{sec:apps}{7}
\avmainnumset{sec:disc}{8}
\avmainnumset{ass:exch}{2.1}
\avmainnumset{ass:indep}{2.2}
\avmainnumset{rem:degenerate}{2.1}
\avmainnumset{def:bstar}{3.1}
\avmainnumset{prop:bstarvalid}{4.1}
\avmainnumset{prop:fwer}{4.2}
\avmainnumset{lem:sufficient}{4.3}
\avmainnumset{prop:espcoords}{4.4}
\avmainnumset{thm:bayes}{4.5}
\avmainnumset{cor:esp}{4.6}
\avmainnumset{cor:corners}{4.7}
\avmainnumset{thm:grow}{4.8}
\avmainnumset{prop:growmean}{4.9}
\avmainnumset{thm:regret}{4.10}
\avmainnumset{prop:twosided}{4.12}
\avmainnumset{cor:rate}{4.13}
\avmainnumset{lem:ceiling}{4.14}
\avmainnumset{cor:penalty}{4.15}
\avmainnumset{prop:sandwich}{4.16}
\avmainnumset{cor:np}{4.17}
\avmainnumset{thm:levelset}{4.18}
\avmainnumset{eq:fwerN}{1}
\avmainnumset{eq:costobj}{2}
\avmainnumset{eq:procstat}{3}
\avmainnumset{eq:mixmart}{4}
\avmainnumset{eq:lr_identity}{6}
\avmainnumset{eq:growinghorizon}{7}
\avmainnumset{eq:waldreduction}{8}
\avmainnumset{eq:bfeasible}{5}
\avmainnumset{eq:dpbracket}{9}
\newcommand{\mainnum}[1]{\ifavmain\ref{#1}\else\csname avmainnum@#1\endcsname\fi}
\newcommand{\mainref}[2]{\ifavmain #1~\ref{#2}\else #1~\mainnum{#2}\fi}
\newcommand{\suppref}[3]{\ifavsupp #1~\ref{#2}\else #3\fi}
\ifdefined\AVBLIND
  \newcommand{\auditdetail}[2]{#2}
\else
  \newcommand{\auditdetail}[2]{#1}
\fi
\newcolumntype{L}[1]{>{\raggedright\arraybackslash}p{#1}}
\newcommand{\spacingset}[1]{}

\newcommand{\papertitle}{Pooling Sequential Evidence Across Hypotheses:
Rate-Optimal Multiple Testing at a Fixed Horizon}
\newcommand{\runningtitle}{Pooling Sequential Evidence Across Hypotheses}

\endlocaldefs

\begin{document}

\ifavmain
\begin{frontmatter}
\title{\papertitle}
\runtitle{\runningtitle}

\ifdefined\AVBLIND
\begin{aug}
\author{\fnms{}~\snm{Anonymous}}
\end{aug}
\else
\begin{aug}
\author[A]{\fnms{Prasanjit}~\snm{Dubey}\ead[label=e1]{pdubey31@gatech.edu}}
\and
\author[A]{\fnms{Xiaoming}~\snm{Huo}\ead[label=e2]{huo@gatech.edu}}
\address[A]{H.~Milton Stewart School of Industrial and Systems Engineering,
Georgia Institute of Technology,
\printead{e1,e2}}
\end{aug}
\fi

\begin{abstract}
We study sequential testing of a fixed family of hypotheses when observations are costly and a
sampling horizon is specified in advance. The challenge is to pool evidence for earlier decisions
when the number and identities of false hypotheses are unknown, while controlling the probability
of any false rejection at level $\alpha$. Existing merges attain the pooled growth rate at a
single number of false hypotheses: averaging when one is false, multiplying when all are. Under an
independent-stream model with common simple null and alternative distributions, we test each
intersection---the claim that all its hypotheses are true---with a prior-weighted mixture of
products of marginal likelihood ratios. Closed testing combines these elementary-symmetric-polynomial
mixtures to identify individual false hypotheses. Design-specific boundary calibration gives
finite-horizon family-wise error control, with exact finite-state guarantees or a confidence
qualification for Monte Carlo calibration.
The prior-matched mixture uniquely maximizes expected log evidence at each horizon. Mixtures
assigning positive weight to every nonempty subset of streams attain log-growth rate $lD$
when $l$ streams follow the alternative. Here $D$ is the mean log likelihood ratio per alternative
observation, and a round supplies one observation per stream. This rate attains the first-order
intersection-delay lower bound as $\alpha\downarrow0$, with dimension, configuration and weights
fixed and the horizon growing sufficiently fast. Power for an individual hypothesis cannot exceed
the best single-stream power at a given deadline. For these full-support mixtures, uncapped full
closure at Ville's threshold $1/\alpha$ places each false hypothesis's rejection time between single-stream
crossing times, almost surely for all sufficiently small $\alpha$; the bounds coincide when all
hypotheses are false.
Gaussian and basket-trial simulations and language-model and randomized-advertising replays
illustrate these distinct intersection and individual benefits. The primary basket boundaries are
$40$--$52\%$ below $1/\alpha$. Across $41$ simulated configurations with identifiable
first-correct-rejection costs, the median reduction in capped mean patient outcomes relative to
prespecified interim-look Bonferroni tests is $31\%$.
\end{abstract}


\end{frontmatter}

\section{Introduction}

Sequential studies often monitor a fixed family of hypotheses within a prespecified observation
budget. Basket trials cap enrollment across disease cohorts~\citep{flaherty2020,odwyer2023};
language-model evaluations incur substantial scoring costs on finite benchmark
slices~\citep{evaleval2026,hsushekhar2026}; and stratified randomized experiments may reach decisions
before processing their full holdouts. These settings motivate procedures that control family-wise
error while making useful claims early.

Existing merges are tuned in advance for a particular number of false hypotheses. Averaging the
streams' evidence~\citep{vovk2021,hartoglei2025} attains the pooled growth rate only when one
hypothesis is false, and multiplying it~\citep{fischerramdas2024} only when all are, while testing
each stream separately forgoes pooling and pays a multiplicity penalty in its threshold. A
threshold valid over an unlimited horizon also leaves error budget unspent when the horizon is
fixed in advance. The signal-seeking basket design followed in \S\ref{sec:basket} makes these
costs concrete: ten disease cohorts, each capped at 31 analyzable patients and each testing a null
response rate of $0.05$ against a working alternative of $0.25$ at level $0.05$. Which merge
serves best there depends on how many cohorts respond, and that is what the trial is run to
discover. The threshold matters as much: because one response multiplies a cohort's likelihood
ratio by five, evidence jumps past an unlimited-horizon threshold rather than approaching it, so
monitoring at that threshold used only between a third and two fifths of the nominal error
budget in the simulations reported there.

We test $K$ null hypotheses $H_1,\dots,H_K$ using independent streams with common null density
$f_0$ and alternative density $f_A$. The unknown configuration $c\subseteq\{1,\dots,K\}$
identifies the false hypotheses, whose streams we call \emph{active}. Each monitoring round supplies one observation from every stream under observation,
up to a horizon $N$ fixed before data collection. A procedure assigns irrevocable rejection times
$\tau_k$ and controls the probability of rejecting any true null through $N$ by $\alpha$ under
every configuration. This is strong family-wise error rate (FWER) control. Rejecting $H_k$
permits discontinuing its stream, saving $(N-\tau_k)^+$ observations, where $x^+=\max(x,0)$;
the cost of a joint claim is instead the number of observations it takes to make.
Deadline power alone does not reward early rejection, since postponing a rejection until $N$
preserves that event (Remark~\ref{rem:degenerate}). Two obstacles stand between this specification
and a usable procedure. The merge must attain the pooled growth rate at whatever number of streams
is active, since that number is unknown when the merge is chosen; and the boundary must hold at
the horizon $N$ and under the sampling schedules the design permits, not only in the
unlimited-horizon limit. Our goal is one procedure that meets both requirements, rather
than a separate construction for each.

Our procedure pools evidence for each group of hypotheses before using closed testing to identify
individual false hypotheses. We call it \emph{PEM closure}, for its prior-weighted
elementary-symmetric-polynomial (ESP) mixture. For every nonempty intersection
$H_S=\bigcap_{k\in S}H_k$, the null is that all hypotheses in $S$ are true.
For $m=|S|$, fix a prior $\beta_l^{(m)}$, $l=1,\dots,m$, on the number of false
hypotheses in $S$ and distribute its mass uniformly over subsets of that size. If $X_{k,s}$ is
observation $s$ in stream $k$, the marginal likelihood ratios are $L_{k,t}=\prod_{s\le t}f_A(X_{k,s})/f_0(X_{k,s})$, and the resulting mixture is
\[
E_{S,t}=\sum_{l=1}^{m}w_l^{(m)}
e_l\bigl(\{L_{k,t}\}_{k\in S}\bigr),
\qquad w_l^{(m)}=\beta_l^{(m)}/\tbinom ml.
\]
Here $e_l$ sums products over the $l$-element subsets of $S$; for three inputs $x=(x_1,x_2,x_3)$,
$e_1(x)=x_1+x_2+x_3$, $e_2(x)=x_1x_2+x_1x_3+x_2x_3$, and $e_3(x)=x_1x_2x_3$.
We call the mixture the PEM statistic. The geometric default
$w_l^{(m)}=(2^{1/m}-1)^l$ has \emph{full support}: it assigns positive mass to every nonempty configuration.
Before observing data, the operator selects a boundary $b_m$ certified for the intersection's
null crossing law, horizon and sampling schedule. Intersections are permanently marked when
their statistics cross their boundaries; $H_k$ is rejected once every intersection containing $k$
is marked. Ville's inequality~\citep{ville1939} supplies the generic threshold $1/\alpha$, which controls null
crossings over an unlimited monitoring horizon. Exact finite-state calibration can lower that
threshold for a specified finite horizon; Monte Carlo calibration carries an explicit confidence
qualification. Section~\ref{sec:method} gives both
routes, the Ville fallback, literal closure (evaluating every intersection), and a conservative
shortcut for larger families.

Two levels of claim must be distinguished. Rejecting $H_S$ establishes that at least one hypothesis
in $S$ is false. Rejecting an \emph{elementary} hypothesis $H_k$, an individual one rather than an
intersection, identifies which hypothesis is false and requires all its containing intersections to
have been rejected. Write $D=\KL(f_A\|f_0)$ for the
Kullback--Leibler divergence, the mean log likelihood ratio per observation from an active stream.
Pooling $l$ active streams can increase the log-growth rate
for an intersection from $D$ to $lD$. It cannot give an elementary hypothesis more information
about its own truth than its stream supplies. The distinction explains both the gains and the
limits of closure.

The main results make these statements precise.
\begin{enumerate}
\item \textbf{Validity and the state of the experiment.}
Feasible local boundaries give strong FWER control through $N$; Ville boundaries give anytime
control, with no correction over intersections or monitoring rounds
(Propositions~\ref{prop:bstarvalid} and~\ref{prop:fwer}). The labeled marginal likelihood ratios
are minimal sufficient at each fixed time, and their ESPs generate the permutation-invariant
information (Lemma~\ref{lem:sufficient}, Proposition~\ref{prop:espcoords}). Sequential closure
additionally needs permanent intersection marks.

\item \textbf{Choice of mixture.}
The prior-matched mixture uniquely maximizes expected log evidence at each horizon, up to null
sets (Theorem~\ref{thm:bayes}); its limiting choices, the arithmetic mean and product, correspond to singleton and
all-alternatives priors (Corollary~\ref{cor:corners}). Without a specified prior, the arithmetic
mean maximizes the worst-case expected log evidence over configurations
(Theorem~\ref{thm:grow}, Proposition~\ref{prop:growmean}). These expected-log
criteria do not order rejection times universally.

\item \textbf{Local rejection rates.}
For fixed family size, configuration and full-support weights, the PEM statistic has log-growth
rate $lD$ and attains the lower bound $\log(1/\alpha)/(lD)$ on first-order local delay as
$\alpha\downarrow0$, under uncapped monitoring or a horizon growing fast enough to contain the
crossing (Proposition~\ref{prop:twosided}, Corollary~\ref{cor:rate}). Unlike the corners
(Table~\ref{tab:streams}), a full-support mixture pays only a constant pathwise log-evidence loss
against every configuration likelihood ratio (Theorem~\ref{thm:regret}).

\item \textbf{Elementary rejection times.}
Any FWER-controlling procedure faces a single-stream level-$\alpha$ power envelope, defined
separately at each deadline (Lemma~\ref{lem:ceiling}); one sequential rule need not attain every
envelope. For literal full-support PEM closure at Ville's boundary, as $\alpha\downarrow0$ an
eventual pathwise sandwich bounds elementary rejection times between singleton crossings at
thresholds a fixed log distance apart. The sandwich collapses under dense alternatives,
removing the multiplicity cost retained by $e$-Bonferroni, which tests each stream at threshold
$K/\alpha$ (Corollary~\ref{cor:penalty}). This does
not assert finite-level pathwise dominance or cover horizon-calibrated boundaries.

\item \textbf{Use of the horizon.}
Terminal mixture thresholding, randomized when necessary, is Neyman--Pearson optimal against the prior mixture
(Corollary~\ref{cor:np}). For early rejection, Theorem~\ref{thm:levelset} gives a time-dependent
boundary in the one-stream Lagrangian problem. Numerical primal and dual bounds assess the
deployed flat rule on tractable local problems (\S\ref{sec:dpvalidation}), without establishing
optimality of the full closed-testing policy.
\end{enumerate}

At the primary geometric prior, the calibrated basket boundaries are $40$--$52\%$ below
$1/\alpha$. Across $41$ simulated basket configurations with identifiable first-correct-rejection
costs, the median reduction in capped mean patient outcomes relative to the prespecified
interim-look Bonferroni rule is $31\%$. These are design-specific results under the certified
horizons and schedules. Language-model and advertising replays illustrate complementary evidence
patterns (\S\ref{sec:apps}, Table~\ref{tab:appsummary}); their true null configurations are unknown.

The construction combines established ingredients: closed testing, ESP merging, and
likelihood-ratio log optimality. Its contribution is to connect the configuration prior, local
evidence growth, elementary rejection limits and design-specific horizon calibration within one
procedure. Section~\ref{sec:sota} and Table~\ref{tab:streams} position it against the sequential
comparators.

Sections~\ref{sec:problem}--\ref{sec:phase} develop the model, procedure, theory and comparisons;
\S\S\ref{sec:sims}--\ref{sec:disc} give the empirical studies and discussion. The supplement
contains the exact algorithm, the complete proofs and supporting derivations, additional
simulations, the finite-population construction, and the three application audits. Below, an
appendix letter or an S-prefixed number refers to that supplement.

\section{Problem formulation}\label{sec:problem}

\subsection{Data, hypotheses, and marginal evidence}\label{sec:model}

The data are $K$ independent sequential streams, and the unknown is which of them follow the
alternative.  We test $K$ hypotheses $H_1,\dots,H_K$; write
$[K]=\{1,\dots,K\}$. For each $k$, observations
$X_{k,1},X_{k,2},\dots$ arrive sequentially. Let
$\F_t=\sigma\{X_{k,s}:k\in[K],\,s\le t\}$ be the full filtration, the information available
from all streams through time $t$, and, for every
$S\subseteq[K]$, let $\F_t^S=\sigma\{X_{k,s}:k\in S,\,s\le t\}$ be the local filtration. The unknown configuration is
$\vec h=(h_1,\dots,h_K)\in\{0,1\}^K$, with $h_k=1$ iff $H_k$ is false; we identify it with the set
$c=\{k\in[K]:h_k=1\}$ of alternatives and write $l_c=|c|$ for its size, abbreviated $l$ when
unambiguous. Under $h_k=0$ the stream is
i.i.d.\ $f_0$ (null); under $h_k=1$ it is i.i.d.\ $f_A$ (alternative).

The baseline theory uses two assumptions: relabeling the streams relabels their joint law, and
the streams are independent conditional on their null or alternative status.

\begin{assumption}[Exchangeability]\label{ass:exch}
The tests share a common null $f_0$ and common alternative $f_A$. Write $X=(X_{k,s})$ for the
full data array and $\mathcal L_{\vec h}(X)$ for its joint law under configuration $\vec h$.  For a
permutation $\rho$ of $[K]$, define $(\rho\vec h)_k=h_{\rho^{-1}(k)}$ and
$(\rho X)_{k,s}=X_{\rho^{-1}(k),s}$. Then
$\mathcal L_{\vec h}(\rho X)=\mathcal L_{\rho\vec h}(X)$.
\end{assumption}
\begin{assumption}[Conditional independence]\label{ass:indep}
Given $\vec h$, the $K$ streams are independent.
\end{assumption}

Assumption~\ref{ass:exch} is the $\vec h$-exchangeability condition
of~\citet[Assumption~2.1]{rosset2022}, in the form used by~\citet{dubeyhuo2026};
Assumption~\ref{ass:indep} is the baseline we relax in \S\ref{sec:disc}.

The marginal evidence for each stream is its accumulated likelihood ratio.  The per-step likelihood
ratio and accumulated likelihood ratio are
\[
r_{k,s}=\frac{f_A(X_{k,s})}{f_0(X_{k,s})},\qquad
L_{k,t}=\prod_{s\le t} r_{k,s}.
\]
The empty product gives $L_{k,0}=1$; values above one favor the alternative over the null.
Each $L_{k,t}$ is a
nonnegative $\F_t^{\{k\}}$-martingale with $\E_{f_0}[L_{k,t}]=1$ when $H_k$ is true. Under
Assumption~\ref{ass:indep}, it is also an $\F_t$-martingale under every fixed configuration in which
$H_k$ is null. The per-step variables $r_{k,s}$ are i.i.d.\ mean-one under $f_0$, not themselves a
time-indexed martingale. Throughout, both divergences $D=\KL(f_A\|f_0)$ and
$\bar D=\KL(f_0\|f_A)$ are finite and strictly positive.  Under $f_A$ the log likelihood ratio
$\log L_{k,t}$ is a random walk with drift $D$ per observation; under $f_0$ its drift is $-\bar D$.
These two drifts govern every rate statement in the paper.

The validity concept throughout is the e-process, whose expectation at any stopping time is at
most one under the null.  A composite null hypothesis is identified with the set of probability
laws of the data that it allows, and we call that set a \emph{null family}; for the intersection
hypothesis $H_S$ the null family consists of the law of the data under every configuration
$\vec h$ with $h_k=0$ for all $k\in S$, whatever the streams outside $S$ do.  An \emph{e-process} for a null family $H$ is a
nonnegative adapted process $(E_t)$ satisfying $\E_P[E_\tau]\le1$ for every law $P\in H$ and every
bounded stopping time $\tau$.
By truncation and Fatou's lemma, the same inequality then holds for every almost surely finite
stopping time; no $E_\infty$ need be defined. Its value at any fixed time is an e-value. Every
nonnegative supermartingale starting at a value at most one is an e-process. Under the baseline
model, all intersection e-processes used by PEM closure (\S\ref{sec:method}) are test
martingales: nonnegative martingales starting at one.

\subsection{Error criterion and objective}\label{sec:setup}

The procedure must control the family-wise error rate through the prespecified horizon.  Fix a level $0<\alpha<1$ and a prespecified horizon $N$, the largest number of observations any
stream contributes, so that every decision is made at some time $t\le N$.  The streams are
monitored on the common clock of \S\ref{sec:model}, and at each time $t$ the procedure decides from
$\F_t$, the data of all $K$ streams, which hypotheses to reject.  A sequential policy is a vector
$\boldsymbol\tau=(\tau_1,\dots,\tau_K)$ of $(\F_t)$-stopping times at which the hypotheses are
rejected, irrevocably; different hypotheses may be rejected at different times, and once $H_k$ is
rejected, PEM closure no longer needs its stream for later decisions and permits discontinuation.  Write
$\mathcal R_t=\{k:\tau_k\le t\}$ for the set of hypotheses rejected by time $t$ and
$V_t=|\{k\in\mathcal R_t:h_k=0\}|$ for the number of true nulls among them.  Write $P_{\vec h}$ for
the law of the data under configuration $\vec h$, and $\Pr_{\vec h}$ and $\E_{\vec h}$ for
probability and expectation under it.  We require, for every configuration,
\begin{equation}\label{eq:fwerN}
\mathrm{FWER}_N(\boldsymbol\tau)=\Pr_{\vec h}\!\Big(\exists\,t\le N: V_t>0\Big)\le\alpha .
\end{equation}
Every application in this paper fixes $N$ in advance: a benchmark is a finite population, a basket
protocol caps accrual, a holdout is exhausted.  Control over an unbounded horizon is a strictly
stronger requirement, and \eqref{eq:fwerN} follows from it, so nothing below is lost by adopting the
bounded form; \S\ref{sec:calibration} shows what is gained.  Two sources of multiplicity are in
play, the $K$ hypotheses and the $N$ analysis times, and PEM closure handles them by two
separate devices: closed testing for the hypotheses and a maximal inequality for the analysis times.

Among policies satisfying~\eqref{eq:fwerN}, we seek to reject as early as the evidence allows.
For one claim, write $\sigma$ for its rejection time: $\tau_k$ for the elementary claim $H_k$,
or the rejection time of an intersection $H_S$. Let $\Sigma$ be the stopping times based on that
claim's own observations together with an independent randomization; those observations are
stream $k$ for $H_k$, and the streams in $S$ for $H_S$. Include $\sigma=\infty$, meaning never reject, with reward
$(N-\sigma)^+=0$. Write $\Pr_0$ for the claim's null law: $f_0$ on stream $k$ or the all-null
law on $S$.

A design prior specifies the alternatives against which observation savings are evaluated.
For an intersection, fix prior probabilities $\omega_c$ over nonempty active subsets $c\subseteq S$.
The mixture law $\bar Q_\omega$ first selects $c$ with probability $\omega_c$, then generates the
streams in $S$ under that configuration; its likelihood ratio is given in~\eqref{eq:lr_identity}.
For an elementary claim, this reduces to $f_A$ on its single stream. We use the local benchmark
\begin{equation}\label{eq:costobj}
\max_{\sigma\in\Sigma}\ \E_{\bar Q_\omega}\big[(N-\sigma)^+\big]
\qquad\text{subject to}\qquad \Pr_0(\sigma\le N)\le\alpha .
\end{equation}
The constraint is~\eqref{eq:fwerN} restricted to one claim. For an elementary rejection at
$\sigma$, the reward is the number of observations saved by discontinuing that stream. For an
intersection, it measures rounds saved in a local, prior-mixture problem.

This local benchmark does not optimize the coupled elementary decisions under closure.
An intersection rejection licenses no stream discontinuation by itself, so the objective is not
the total saving $\sum_k(N-\tau_k)^+$ from a vector policy $(\tau_1,\dots,\tau_K)$.
We do not solve~\eqref{eq:costobj} in closed form or optimize it directly. Section~\ref{sec:boundaryshape}
characterizes the single-stream boundary shape, and \S\ref{sec:dpvalidation} uses dynamic
programming to bracket the difference between the deployed flat local rule and the optimum on
tractable local problems.

\begin{remark}[Power alone is degenerate]\label{rem:degenerate}
One might instead judge a rule by the probability that it rejects a false claim by the deadline,
that is, by its power at $N$.  This criterion never rewards stopping early.  Whether a rule has
rejected by $N$ is determined by the data up to $N$, so any rule that stops early can be replaced
by one that waits until $N$ and then rejects exactly on the paths where the first rule would have
rejected.  The replacement has the same power and the same level, and it uses all $N$
observations.  The largest power at $N$ is therefore attained by the fixed-sample Neyman--Pearson
test at $N$, which never stops early.  Early stopping must be justified by a cost of observation,
which is why the criterion is~\eqref{eq:costobj} rather than power.
\end{remark}

PEM closure makes claims at
two levels, and the $K$ streams help differently at each.  An \emph{elementary} claim rejects a
single $H_k$.  An \emph{intersection} claim rejects $H_S=\bigcap_{k\in S}H_k$, the hypothesis that
every stream in $S$ is null; its rejection asserts that at least one stream in $S$ is active, and it
is the claim that closed testing (\S\ref{sec:sota}) tests directly.  For an intersection claim the
streams in $S$ are pooled, and pooling $l$ active streams raises the exponential rate of the evidence
from $D$ to $lD$, which no valid procedure exceeds (Proposition~\ref{prop:twosided}).  For an
elementary claim the deadline-specific power envelope is attainable using the stream alone, and pooling can only
remove the multiplicity penalty (Lemma~\ref{lem:ceiling}, Corollary~\ref{cor:penalty}).

\subsection{Existing approaches and open questions}\label{sec:sota}

Existing work supplies three building blocks: closed testing for multiple hypotheses, rules for
combining evidence, and a threshold valid over time.
Closed testing~\citep{marcus1976} rejects $H_k$ once every intersection $H_S$ containing it has
been rejected. Valid local e-processes make this reduction anytime valid, that is, valid at every
stopping time; the framework is generic,
with efficient shortcuts for particular merges~\citep{fischerramdas2024,hartoglei2025}. At a fixed
time, arithmetic-mean closure is called $e$-Holm. It contains classical Holm applied to reciprocal
e-values because each arithmetic-mean local test dominates the corresponding Bonferroni test;
the procedures need not agree. We study a fixed family observed repeatedly, unlike the online
arrival of hypotheses also treated by~\citet{fischerramdas2024}. Order-invariant formulations
require running suprema of e-processes~\citep{tavyrikov2025}.

The local merges combine evidence using the normalized ESPs
$\tilde e_l=e_l/\binom{m}{l}$ of~\citet[Section~4, equation~(8)]{vovk2021} and their
convex mixtures. The arithmetic mean ($l=1$) and product ($l=m$) are the two \emph{corners}
used as comparators below. Sequential merging
and product-type extensions are studied by~\citet{vovkwang2024,ming2026}. Full-support mixtures
already belong to this established class; our contribution is their prior-optimal selection,
minimax characterization, rate and closure analysis, and horizon calibration. Ville's
$1/\alpha$ supplies the generic threshold without a horizon restriction.

The marginal comparator, $e$-Bonferroni, rejects $H_k$ when $L_{k,t}$ first reaches $K/\alpha$;
it does not pool evidence. Table~\ref{tab:streams} distinguishes how the comparators use the streams.
The corners do not
attain rate $lD$ at every configuration. The balanced global-null merge attains it at the sparse
and dense endpoints but generally not between them; the arm-sampling comparator uses a different
clock~\citep{bharti2026,sandoval2026}. Neither identifies individual active streams.
Asymptotic sequential multiple-testing theory optimizes terminal policies under known bounds on
the number of signals as error levels vanish~\citep{songfellouris2017,songfellouris2019,xingfellouris2025,xing2026,liusong2026}.
Here expected-log optimality holds at each finite horizon and feasible local boundaries give
nonasymptotic error control; the rejection-delay result remains asymptotic.

Bounded-horizon anytime-valid inference, which tightens thresholds to a known deadline, is
closest to the calibration
step~\citep{mathiswaudbysmith2026,baas2026curtailment,sokolova2026,clerico2026,taga2026}.
It is not uniformly single-claim: \citet[Section~6.4]{sokolova2026} discusses multiplicity across
platform-trial arms. Its scalar calibration must here serve intersections of different sizes and
cover their accrual and cap profiles. Fixed-sample optimal-FWER policies in related exchangeable
models also use ESPs, as coefficients of power and error functionals; here the polynomial is the
monitored e-process~\citep{rosset2022,dubeyhuo2025fast,dubeyhuo2026,dubeyhuo2026boost,dubeyhuo2026prior}.

\begin{table}[!tbp]
\centering
\footnotesize
\setlength{\tabcolsep}{3pt}
\caption{How existing methods use the $K$ streams, and the rate at which the evidence for an
intersection $H_S$ with $m$ streams, $l$ of them active, grows per synchronous round, with one observation per stream.  $D=\KL(f_A\|f_0)$
and $\bar D=\KL(f_0\|f_A)$.  The last row, in bold, is PEM closure.}
\label{tab:streams}
\begin{tabular}{@{}L{5.6cm}L{3.0cm}L{2.8cm}L{2.0cm}@{}}
\toprule
Method & Other streams enter through & Effect of a null stream & Rate for $H_S$\\
\midrule
$e$-Bonferroni \citep{vovk2021} & not at all; threshold $K/\alpha$ & penalty $\log K$ always & $D$ (no pooling)\\
Arithmetic-mean closure \citep{vovk2021,hartoglei2025} & additive pooling in each $S$ & dilution by $1/m$ & $D$\\
Product closure \citep{fischerramdas2024} & multiplicative pooling in each $S$ & subtracts $\bar D$; can stall & $lD-(m-l)\bar D$\\
Balanced global-null merge \citep{bharti2026} & pooled for one global claim; the two-atom member $\tfrac12\tilde e_1+\tfrac12\tilde e_m$ of the PEM family & dilution or subtraction, whichever atom leads & $\max\{D,\,lD-(m-l)\bar D\}$; equals $lD$ at $l=1$ and $l=m$\\
Arm-sampling global merge \citep{sandoval2026} & one arm sampled per step, for one global claim & through the sampling rule & not comparable on this clock: a step is one observation, not one per stream\\
Asymptotic sequential multiple testing \citep{songfellouris2017,xingfellouris2025} & a known bound on the number of signals & through the bound & first order as $\alpha\to0$\\
\textbf{PEM closure} (this paper, \S\ref{sec:method}) & \textbf{configuration mixture; posterior-weighted updates} & \textbf{downweighted adaptively} & $\boldsymbol{lD}$ \textbf{for every} $\boldsymbol{l}$\\
\bottomrule
\end{tabular}
\end{table}

Two choices remain: a mixture attaining rate $lD$ at unknown cardinality, and size-specific
boundaries valid for every admitted sampling profile. Section~\ref{sec:method} specifies them;
\S\ref{sec:theory} establishes the guarantees.

\section{The proposed procedure: PEM closure}\label{sec:method}

This section gives the implementable procedure, its statistic and weights, horizon calibration,
and the closed test with its shortcut.

\subsection{The procedure}\label{sec:procedure}

Prespecify the inputs and boundaries before observing data, then repeat the three steps below at
every $t\le N$.  The inputs are the null and alternative densities $f_0$ and $f_A$, the number
of hypotheses $K$, the level $\alpha$ and the horizon $N$, and, for each intersection size
$m=1,\dots,K$, a design prior $\beta^{(m)}=(\beta^{(m)}_1,\dots,\beta^{(m)}_m)$ on the number of
false hypotheses among $m$.  Absent a design prior, use the geometric default
$w_l^{(m)}=(2^{1/m}-1)^l$, which has full support and equalizes the drift-normalized log-regret
bound in Theorem~\ref{thm:regret} across cardinalities.  From these inputs, compute $K$ deployed boundaries $b_m$, $m=1,\dots,K$, one for each
intersection size.  Each must satisfy
$\Pr_0(\max_{0\le t\le N}E_{S,t}\ge b_m)\le\alpha$ under its calibration schedule, where
$E_{S,t}$ is the statistic in~\eqref{eq:procstat} below.  When the smallest feasible member of the
specified population candidate set exists, denote it by $b^\star(\alpha,N,m)$.
Use exact enumeration for a tractable finite-state model, a Monte Carlo boundary $\widehat b$
with a stated calibration-confidence guarantee otherwise, and $1/\alpha$ when neither route is
available.  Schedule guards check feasibility under the permitted sampling schedules and can raise a
calibrated boundary, always retaining the Ville fallback;
\S\ref{sec:calibration} describes both routes and \suppref{Algorithm}{alg:bstar}{Algorithm~S2} states them step by step.

\begin{enumerate}
\item \textbf{Marginal evidence.} For each unrejected hypothesis $H_k$, update the likelihood ratio
$L_{k,t}=\prod_{s\le t}f_A(X_{k,s})/f_0(X_{k,s})$ of its first $t$ observations.
\item \textbf{Pooling.} For every unmarked nonempty $S\subseteq[K]$ with $m=|S|$, compute the
prior-weighted mixture of the elementary symmetric polynomials of the likelihood ratios in $S$,
\begin{equation}\label{eq:procstat}
E_{S,t}=\sum_{l=1}^{m} w_l^{(m)}\,e_l\bigl(\{L_{k,t}\}_{k\in S}\bigr),
\qquad w_l^{(m)}=\frac{\beta^{(m)}_l}{\binom ml},
\end{equation}
where $e_l$ is the $l$-th elementary symmetric polynomial.  Mark $H_S$ rejected if
$E_{S,t}\ge b_m$; a mark, once made, stands.
\item \textbf{Closure.} Reject $H_k$ at the first $t$ at which every $H_S$ with $S\ni k$ carries
a mark, and discontinue stream $k$; the remaining streams continue to $N$.  For large $K$, a
conservative shortcut evaluates $K$ least-favorable intersections per hypothesis instead of
$2^{K-1}$.
\end{enumerate}

Figure~\ref{fig:closure} shows one such decision at $K=3$: marks accumulate on the lattice of
intersections, and an elementary hypothesis is rejected exactly when every intersection above it
carries one.

\begin{figure}[!tbp]
\centering
\includegraphics[width=\linewidth]{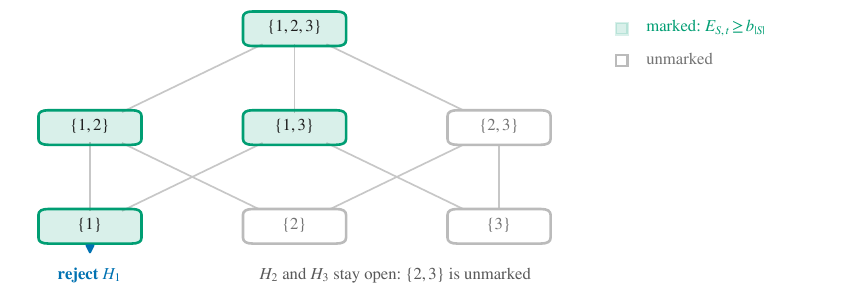}
\caption{An elementary rejection needs every intersection containing its stream to be marked, so a
single unmarked intersection holds two hypotheses open.  One PEM closure decision at $K=3$, drawn
over the seven nonempty intersections ordered by inclusion.  Stream~1 has cleared its boundary in
every intersection that contains it, so $H_1$ is rejected; $\{2,3\}$ has not crossed, so $H_2$ and
$H_3$ remain open even though $\{1,2\}$, $\{1,3\}$ and $\{1,2,3\}$ are marked.  Marks are
permanent, so the figure records the history of crossings and not only the current statistics.  The
diagram illustrates the rule and reports no simulation output.}
\label{fig:closure}
\end{figure}

The output is the rejection times $\tau_k=\max_{S\ni k}\inf\{t\le N:E_{S,t}\ge
b_{|S|}\}$ and the set of hypotheses rejected by $N$.  The intersection marks are also
reportable, with one distinction.  A single prespecified $H_S$, fixed before the data, is tested at
level $\alpha$ by its own mark under the applicable boundary-feasibility guarantee, so that mark
supports the claim that at least one stream in $S$ is
active.  The collection of all $2^K-1$ marks is not thereby simultaneously protected: to report a
set of intersection claims with family-wise protection, report $H_S$ only when every $H_{S'}$ with
$S'\supseteq S$ is also marked, which is the closed test applied at the intersection level.  The
primary full-intersection claim $H_{[K]}$ has no proper superset and so needs no such qualification,
and the elementary rejections are protected by construction.

\subsection{The intersection statistic and its weights}\label{sec:object}

The PEM statistic is a prior-weighted mixture of the configuration likelihood ratios,
which under an exchangeable prior is a mixture of elementary symmetric polynomials.  Fix a nonempty
$S\subseteq[K]$, write $m=|S|$, and let $\mathcal C$ be the nonempty configurations $c\subseteq S$.
Under configuration $c$ the likelihood ratio of the data in $S$ against the all-null law is
$L_{c,t}=\prod_{k\in c}L_{k,t}$ by Assumption~\ref{ass:indep}.  For a configuration prior
$w=(w_c)_{c\in\mathcal C}$ with $\sum_cw_c=1$, put $E^w_{S,t}=\sum_{c\in\mathcal C}w_cL_{c,t}$.  If
$w$ is exchangeable, $w_c=w_{|c|}$, this accumulated mixture has the ESP form
\begin{equation}\label{eq:mixmart}
E^w_{S,t}=\sum_{c\in\mathcal C}w_{|c|}L_{c,t}
=\sum_{l=1}^{m} w_l\,e_l\bigl(\{L_{k,t}\}_{k\in S}\bigr),
\qquad \sum_{l=1}^m\binom ml w_l=1,
\end{equation}
which is~\eqref{eq:procstat} with $\beta_l=\binom mlw_l$ the prior probability that exactly $l$ of
the $m$ hypotheses in $S$ are false.  Equivalently, with normalized ESPs $\tilde e_l=\binom
ml^{-1}e_l$, the right-hand side is $\sum_l\beta_l\tilde e_l(\{L_{k,t}\}_{k\in S})$.  Equation~\eqref{eq:mixmart} is the central object of the paper; when the prior is fixed we suppress
the superscript $w$.

The statistic is a test martingale in the full filtration under every
configuration in which all of $S$ is null, whatever the streams outside $S$ do.  Each $L_{c,t}$ is a
mean-one martingale under the all-null law of $S$, so $E^w_{S,t}$ is a mean-one martingale in the
local filtration $(\F^S_t)$.  More importantly for closed testing, under every full-data
configuration $\vec h$ satisfying $H_S$, conditional independence gives
$\E_{\vec h}[L_{c,t}\mid\F_{t-1}]=L_{c,t-1}$ for every $c\subseteq S$, so $E^w_{S,t}$ is an
$\F_t$ test martingale uniformly over the nuisance configurations outside $S$ and is anytime valid by
Ville's inequality.

A single fixed prior adapts to the observed sparsity, because the
per-step multiplier of the PEM statistic is a posterior-weighted average of the configuration
multipliers.  On $\{E^w_{S,t-1}>0\}$,
\[
\frac{E^w_{S,t}}{E^w_{S,t-1}}
=\sum_{c\in\mathcal C}\pi^w_{c,t-1}\prod_{k\in c}r_{k,t},
\qquad
\pi^w_{c,t-1}=\frac{w_cL_{c,t-1}}{E^w_{S,t-1}},
\]
the Bayes posterior over configurations given the data in $S$.  A stream whose likelihood ratio has
decayed drags down the posterior mass of every configuration that contains it, so null streams are
downweighted rather than averaged in or multiplied in.  This is the mechanism behind the last row of
Table~\ref{tab:streams}.

Once the model and sampling design are fixed, the mixture weights encode the design prior: if $\beta_l$
is the prior probability that exactly $l$ of the $m$ hypotheses in $S$ are false, the
expected-log-optimal weights are $w_l^{(m)}=\beta_l/\binom ml$ (Corollary~\ref{cor:esp}).  When
no prior is available, two full-support defaults are used.  The \emph{geometric} weights
$w_l^{(m)}=\gamma_m^l$, $\gamma_m=2^{1/m}-1$, equalize the drift-normalized log-regret bound in
Theorem~\ref{thm:regret} across cardinalities and place most level mass at $l=1$; they are the
primary choice in every application of \S\ref{sec:apps}. The \emph{level-uniform} weights $\beta_l=1/m$ put equal prior mass on each
number of false hypotheses and are used in the Gaussian simulations of \S\ref{sec:sims}, which
also report the sensitivity of PEM closure to the choice.  Full support is what keeps the loss
from a wrong prior bounded (Theorem~\ref{thm:regret}) and what attains the rate $lD$ at every $l$
(Proposition~\ref{prop:twosided}).

\subsection{The horizon-\texorpdfstring{$N$}{N} boundary}\label{sec:calibration}

Calibrating to a prespecified horizon can lower the rejection threshold while preserving its
error guarantee. Ville's threshold can leave error budget unused because evidence jumps past
the boundary (\emph{overshoot}) or never crosses it; supermartingale evidence may also lose
conditional mean along the way~\citep{delapenaklass2026}. A fixed horizon excludes further paths
whose first crossing would occur only after $N$. Calibration accounts for these effects through
the null crossing probability by $N$. In the computation, a \emph{gate value} is the value of
this statistic at a tabulated state, to be compared with a candidate threshold.

\begin{definition}[Feasible and exact horizon-$N$ boundaries]\label{def:bstar}
Fix an intersection size $m$ and the \emph{observation schedule} on which the streams in $S$ are
read, including any outcome-independent random accrual law and the per-stream caps.
Write $P_0$ for the joint all-null law of observations and accrual under that design.  A threshold
$b$ is \emph{feasible} for $(\alpha,N,m)$ and that schedule if
\begin{equation}\label{eq:bfeasible}
\Pr_{0}\bigl(\max_{0\le t\le N}E_{S,t}\ge b\bigr)\le\alpha .
\end{equation}
Prespecify a population candidate set $\mathcal V_{N,m}$ containing $1/\alpha$ and restricted to
$(0,1/\alpha]$.  For a continuous model use the support of the running maximum, intersected with
that interval and augmented by $1/\alpha$.  For finite-state enumeration use the tabulated gate
values in that interval, augmented by $1/\alpha$.
When its feasible subset has a least element, write
\[
b^\star(\alpha,N,m)=\min\bigl\{b\in\mathcal V_{N,m}:\ \eqref{eq:bfeasible}\ \text{holds}\bigr\}
\]
and call it the exact horizon-$N$ boundary for this candidate set.  Otherwise deploy $1/\alpha$.
The notation suppresses the fixed weights, schedule and candidate set.  Write $b_m$ for the
boundary actually deployed, including any increase required by a schedule guard or the
simulation-calibrated value $\widehat b$.
\end{definition}

Feasibility is sufficient for error control: Proposition~\ref{prop:bstarvalid} uses
only~\eqref{eq:bfeasible}. Since $1/\alpha$ is feasible, the stated fallback gives
$b_m\le1/\alpha$. A minimum exists for a finite candidate set but need not exist for a continuous
one; even Gaussian running maxima can have an atom at their initial value $1$.
\suppref{Appendix}{app:bstar}{Appendix~C} gives sufficient conditions
and counterexamples. The finite Monte Carlo candidate set remains well-defined without a
population minimum.

Calibration at $N$ does not justify a longer horizon: at a fixed threshold, the crossing
probability is nondecreasing in $N$. The selected value $b^\star$ itself need not increase with
$N$, because the prescribed candidate grid can change. Under synchronous sampling,
exchangeability gives one boundary per intersection size; unequal caps and thinned accrual
require the schedule checks below. Boundaries are minimized over the stated candidate grid,
with inclusive crossings $E_{S,t}\ge b$ and conservative numerical tie handling.
\suppref{Appendix}{app:bstar}{Appendix~C} collects the grid and horizon
counterexamples.

\suppref{Algorithm}{alg:bstar}{Algorithm~S2} details both calibration routes.
\begin{itemize}
\item \emph{Exact enumeration.} For a tractable finite-state model, exact recursion computes the
null crossing probability and the smallest feasible grid candidate, without a simulation-confidence
qualification. For Bernoulli streams it tracks the joint response counts, absorbing paths once
the boundary is crossed. Under an exchangeable null and symmetric gate, aggregation by count
multisets and, when $w_1^{(m)}>0$, a safe absorbing count cap reduce the state space; at the product corner, only the
total count is needed. For the $m=10$, $N=31$ NCI design, these reductions replace $32^{10}$
final-layer states with at most $92{,}378$ live multiset states, or $311$ states at the product
corner. Unequal caps or thinned accrual require the asymmetric state space.
\suppref{Appendix}{app:bstar}{Appendix~C} sets out the algorithms,
exactness arguments and measured cost.
\item \emph{Monte Carlo.} Otherwise, simulate $M$ null paths of the $m$ streams and record
$\max_{t\le N}E_{S,t}$ on each.  The returned boundary $\widehat b(\alpha,N,m)$ is the smallest
recorded path maximum no greater than $1/\alpha$ at which the one-sided Clopper--Pearson upper confidence bound on the crossing
probability, at confidence $\gamma$, is at most $\alpha$; if no such recorded value qualifies the routine
returns $1/\alpha$.  We write $\widehat b$ rather than $b^\star$ because it is a random variable
computed from the calibration sample: it is not the population minimum, and the sample hit fraction
at $\widehat b$ is not the population crossing probability.  What the construction delivers is that,
with probability at least $\gamma$ over the calibration draw, $\widehat b$ is feasible in the sense
of~\eqref{eq:bfeasible}; the calibration sample must then be independent of the data the procedure
is run on.  The Gaussian designs of \S\ref{sec:sims} use $M=50{,}000$ paths at $\gamma=0.999$.
\end{itemize}

For Bernoulli response probabilities $p_k$ and a null upper bound $p_0$, the composite marginal
null is $H_k:p_k\le p_0$. The boundary is computed at $p_k=p_0$ for
every $k\in S$; a monotone coupling of the data shows that $(p_0,\dots,p_0)$ is the least-favorable
point of the intersection null (\S\ref{sec:validity}).  The supplement also contrasts Ville monitoring with the
alternative of a pointwise test followed by a union bound over time, which pays a $\log t$ penalty
that Ville's inequality avoids.

\paragraph{Calibration is schedule-specific.}  Definition~\ref{def:bstar} fixes the schedule
because the null law of $\max_{t\le N}E_{S,t}$ depends on it, and a boundary calibrated under
synchronous full accrual need not be feasible under another.  Factor-one idle updates after a
stream is capped or skipped preserve the martingale property, hence Ville's threshold, but they do
not preserve a smaller calibrated one. Freezing can prevent further decay of null evidence,
so its effect on crossing probability has no fixed direction;
\suppref{Appendix}{app:bstar}{Appendix~C} gives an exact counterexample.

What closed testing needs is that the size-$|S_0|$ boundary be feasible under the schedule of the
true-null set $S_0$, for every configuration in the claimed guarantee. The calibration routine
recomputes each boundary's crossing probability under the cap and accrual profiles it checks,
and raises the boundary to the next grid candidate where necessary.
Coverage of the frozen results differs
by design: the synchronous NCI design and both five-basket designs, the realized unequal caps
included, are verified at every nonempty true-null subset, whereas the NCI thinned-accrual
sensitivity is certified only for the true-null profile that scenario realizes.  Prospective
calibration checks every nonempty subset profile under each prespecified schedule and uses
Ville's threshold when exact calculation exceeds the state budget; this safeguard does not
enlarge the guarantee of the frozen results. \suppref{Appendix}{app:bstar}{Appendix~C} and
\suppref{Table}{tab:scheduleaudit}{Table~S1} record the complete coverage
and what remains outside the verified class.

Unknown null crossing laws require the Ville fallback $1/\alpha$ for the finite-population benchmark
and paired-score advertising processes of \S\ref{sec:apps}; their family-wise error control
holds at every stopping time.

\paragraph{What a simulated boundary guarantees.} Exact calibration gives an unconditional
level guarantee. Monte Carlo calibration gives feasibility on an event of probability at least
$\gamma=1-\delta$. For a fixed configuration, only its true-null set's boundary enters the FWER
proof, yielding the unconditional bound $\alpha+(1-\alpha)\delta$; at $\alpha=0.05$ and
$\delta=10^{-3}$ this is $0.05095$. Simultaneous feasibility of all $K$ boundaries instead has
confidence at least $1-\sum_m\delta_m$ when size $m$ uses confidence $1-\delta_m$.
We report the conditional guarantee and $\gamma$. Searching the nested crossing events requires
no additional correction: each nonfallback boundary is an observed order statistic selected by
a prespecified maximum hit count. \suppref{Appendix}{app:bstar}{Appendix~C} proves this
and gives the internal calibration target for an exact unconditional level, together with tie
and fallback conventions.

\subsection{Closed testing and the least-favorable shortcut}\label{sec:closure}

Closed testing converts the intersection rejections into elementary rejections with no correction
over subsets or over time.  For each size $m$, prespecify the exchangeable weight vector $w^{(m)}$
and use it in~\eqref{eq:mixmart} for every $S$ with $|S|=m$.  The rejection times $\tau_k$ of
\S\ref{sec:procedure} are read with $\inf\varnothing=\infty$: $H_k$ is not rejected if some
intersection containing it fails to cross its boundary by $N$.
Proposition~\ref{prop:bstarvalid} shows this controls family-wise error through $N$ whenever
the deployed $b_m$ are feasible under the required schedules, and Proposition~\ref{prop:fwer} gives control at every stopping time for the
flat threshold.  Each $E_{S,t}$ costs $O(|S|^2)$
via the ESP prefix recursion $e_l(x_{1:j})=e_l(x_{1:j-1})+x_je_{l-1}(x_{1:j-1})$, where
$x_{1:j}=(x_1,\dots,x_j)$, $e_0=1$, and $e_l(x_{1:j})=0$ for $l>j$ (including the empty
prefix when $j=0$). Literal closure
sweeps $2^K-1$ intersections; the $K=6$ simulation and the basket applications do so.

Exchangeability permits a conservative shortcut that
monitors $K$ statistics per hypothesis instead of $2^{K-1}$.  The statistic $E_{S,t}$ depends on $S$
only through the multiset $\{L_{k,t}\}_{k\in S}$ and is coordinatewise increasing, so among size-$m$
subsets containing $k$ it is minimized at
$S_{k,m,t}=\{k\}\cup\{\text{the }m-1\text{ smallest other }L\text{'s}\}$.  If this least-favorable
statistic reaches $b_m$, every size-$m$ intersection containing $k$ has been rejected,
because $b_m$ is the same for all size-$m$ subsets.  Write
$s_{k,m}=\inf\{t\le N:E_{S_{k,m,t},t}\ge b_m\}$ for the first crossing of the size-$m$
least-favorable statistic.  Since rejections are permanent, the sizes need not clear at a common
time, and the \emph{latched} shortcut rejects $H_k$ at
\[
\tau_k^{\mathrm{LF}}=\max_{1\le m\le K}s_{k,m}\ \ge\ \tau_k .
\]
The shortcut is conservative, because the least-favorable subset can change with $t$, so it rejects
only hypotheses rejected by the full closed test and inherits its error control.  It costs $O(K^2)$
per hypothesis per step using sorted-prefix ESPs.  We deploy the latched form whenever literal
closure is infeasible.  The literal algorithm, the comparison with a non-latched rule, and the exact
condition for shortcut equality are in
\suppref{Appendix}{app:algorithm}{Appendix~A}.

\section{Theoretical guarantees}\label{sec:theory}

We establish error control, sufficiency and evidence optimality, then compare intersection and
elementary rejection limits.  Section~\ref{sec:boundaryshape} connects these results to the stopping
objective~\eqref{eq:costobj}. The main text states and interprets the results; full proofs are
collected in the online supplement.

\subsection{Family-wise error control}\label{sec:validity}

Whatever subset of hypotheses is in fact false, the probability that the
procedure rejects any true hypothesis by the horizon is at most $\alpha$, provided the deployed
boundary is feasible for that configuration's true-null schedule.
This is the guarantee~\eqref{eq:fwerN} that the problem demands, and it is obtained without a
Bonferroni-type correction over the $2^K-1$ intersections or over the $N$ analysis times: closed
testing handles the subsets, and the boundary handles the analysis times.  Two versions are proved.
The first is for feasible deployed boundaries $b_m$ and holds through the horizon $N$.  The second is
for the flat threshold $1/\alpha$ and holds at every stopping time, bounded or not; it is the
guarantee that remains when a smaller feasible boundary cannot be certified.

\begin{proposition}[Closed testing with feasible boundaries]\label{prop:bstarvalid}
Fix a configuration and let $S_0$ be its true-null set.  Suppose $H_S$ is rejected at the first
$t\le N$ with $E_{S,t}\ge b_{|S|}$, and $H_k$ once every $S\ni k$ has been rejected.  If
$S_0$ is empty, or its deployed size-$|S_0|$ threshold is feasible in the sense of~\eqref{eq:bfeasible}
under the joint null-observation and accrual law for that configuration, then PEM closure satisfies
\eqref{eq:fwerN}.  Feasibility must hold for every nonempty possible $S_0$ to obtain strong FWER.
For synchronous observations the exact size-specific calibration establishes this simultaneously;
for Monte Carlo boundaries the conclusion is conditional on their calibration-feasibility event.
\end{proposition}
Any false elementary rejection requires rejecting the intersection of all true nulls, so
one feasible local test controls the family-wise error. The proof is in
\suppref{Appendix}{app:closureproofs}{Appendix~C.1}.

Boundary feasibility must match the actual sampling schedule of the true-null streams. Only the
true-null set's own threshold is used, so a
design need not certify all $K$ boundaries against all schedules; but the one it uses must be
feasible under the schedule that set follows, and under unequal caps or thinned accrual that is a
different law from the one the boundary was calibrated against.  Section~\ref{sec:calibration} states
the re-pricing that the basket designs perform, and
\suppref{Appendix}{app:bstar}{Appendix~C} records exactly which schedules
have been verified.

For composite Bernoulli nulls $H_k:p_k\le p_0$, an outcome-independent schedule and
likelihood-ratio factors increasing in the response make $(p_0,\dots,p_0)$ least favorable for the entire intersection.
The joint monotone-coupling argument is in
\suppref{Appendix}{app:bstar}{Appendix~C}.

The flat-threshold version requires only that each intersection process be an e-process in the full
filtration.  That is why it transfers, unchanged, to the application-specific marginal processes of
\S\ref{sec:apps}, whose null laws are not those of the product model.

\begin{proposition}[Strong anytime-valid FWER]\label{prop:fwer}
Suppose that, for every nonempty $S\subseteq[K]$, $(E_{S,t})_{t\ge0}$ is a nonnegative
full-filtration e-process under every distribution in the intersection null $H_S$.  The closed test
that rejects $H_S$ at its first crossing of $1/\alpha$ and rejects $H_k$ only after every $H_S$ with
$S\ni k$ has been rejected controls strong anytime-valid FWER at level $\alpha$ under every
data-generating distribution.  In particular, under Assumptions~\ref{ass:exch}--\ref{ass:indep}, the
PEM statistics in~\eqref{eq:mixmart} satisfy this premise and PEM closure obeys~\eqref{eq:fwerN}.
\end{proposition}
The proof is in
\suppref{Appendix}{app:closureproofs}{Appendix~C.1}.

\subsection{Marginal sufficiency and permutation-invariant coordinates}\label{sec:sufficiency}

Under exchangeability and conditional independence, the \emph{labeled} vector of marginal
likelihood ratios is minimal sufficient for which hypotheses are false, and its elementary
symmetric polynomials generate exactly the permutation-\emph{invariant} information in it.  The two
statements do different work, and the second is weaker than the first: passing from the labeled
vector to its ESPs discards the labels, so the ESPs are \emph{not} sufficient for the labeled active
set, and a rule that must name which streams are active cannot in general be written as a function
of them.  At a fixed horizon, symmetric functions of that labeled vector factor through the ESPs.
The intersection null $H_S$ asserts that every stream in $S$ is null; its rejection asserts that at
least one is active.  Under an exchangeable prior the scalar decision problems specified below admit
a symmetric representative, so ESP coordinates suffice for those problems.
Equation~\eqref{eq:mixmart} monitors a linear function of these coordinates.
This fixed-horizon reduction does not recover historical crossings from the current ESP vector:
sequential closure also stores permanent marks.  Closed testing recovers labeled elementary
decisions by combining the invariant tests over the labeled subsets.

Write $P_0^S$ for the all-null law of the streams in $S$, and $P_0^{(t)}$ and $Q_c^{(t)}$,
$\varnothing\ne c\subseteq S$, for the restrictions to $\F_t^S$
of $P_0^S$ and of the law under configuration $c$; by Assumption~\ref{ass:indep} each
$Q_c^{(t)}$ is a product law and $dQ_c^{(t)}/dP_0^{(t)}=L_{c,t}$, the quantity accumulated
in~\eqref{eq:mixmart}.

\begin{lemma}[Minimal sufficiency of the marginal likelihood ratios]\label{lem:sufficient}
Fix a nonempty $S\subseteq[K]$ and a horizon $t$.  Under Assumption~\ref{ass:indep} the vector
$T_t=(L_{k,t})_{k\in S}$ is minimal sufficient for the family $\{P_0^{(t)}\}\cup\{Q_c^{(t)}:\varnothing
\ne c\subseteq S\}$ on $\F_t^S$.
\end{lemma}

\begin{proposition}[The ESPs are the exchangeable coordinates]\label{prop:espcoords}
Let $m=|S|$ and fix a horizon $t$.
\begin{enumerate}
\item[(a)] \emph{Symmetrization.}  Suppose the action is a scalar, the feasible class of decision
rules is convex, consists of integrable actions, and is closed under both the relabeling action of
$\mathfrak S_m$ and the Rao--Blackwell map $\delta\mapsto\E_{P_0}[\delta\mid T_t]$.  Suppose the design prior is
exchangeable, and the loss is convex in the action and invariant under relabeling.  If a Bayes
optimal rule in that class exists, then one can be chosen to be a symmetric measurable function of
$T_t$: Rao--Blackwellization followed by group averaging is feasible and has Bayes risk no
larger.  These closure conditions hold for scalar terminal e-values.  For labeled vector actions,
the corresponding symmetry requirement is equivariance rather than invariance.
\item[(b)] \emph{The invariant coordinates.}  Without any decision-theoretic hypothesis,
\[
\sigma\big(e_1(T_t),\dots,e_m(T_t)\big)\;=\;\sigma\big(T_t\big)^{\mathfrak S_m},
\]
the sub-$\sigma$-algebra of events invariant under the symmetric group $\mathfrak S_m$ of
permutations of the streams in $S$, so \emph{every} symmetric measurable function of $T_t$ factors
through $\big(e_1(T_t),\dots,e_m(T_t)\big)$, with statistical versions understood modulo null sets.
\end{enumerate}
\end{proposition}

Proposition~\ref{prop:espcoords} is a statement about $\sigma$-algebras, not about polynomial
rings, so it applies to a far wider class of functions than the fundamental theorem of symmetric
polynomials.  The proofs, by the Halmos--Savage factorization criterion and by a Borel-isomorphism
argument from Vieta's formulas, are in \suppref{Appendix}{app:sufficiency}{Appendix~D}.

\subsection{Optimality of the intersection e-process}\label{sec:opt}

The prior-matched configuration mixture maximizes expected log evidence at every horizon.
The comparison is with all e-processes for $H_S$ that use only the data in $S$, and the expectation
uses the specified prior over which hypotheses are false. Under exchangeability this is
the PEM statistic with level weights equal to the prior on the number of false hypotheses.  If instead
one asks for the largest expected log growth in the worst case over configurations, the answer is
the arithmetic mean.  Expected log growth is the criterion for anytime-valid evidence formalized
by~\citet{koolengrunwald2022}, growth-rate optimality (GRO) against a fixed alternative; its
average over a design prior is the Bayes criterion above, and its worst case over configurations
is growth-rate optimality in the worst case (GROW)~\citep{grunwald2024,ramdas2023}.  In the one-hypothesis problem it selects the likelihood-ratio
process monitored by Wald's sequential probability ratio test.  The worst-case result is
a theorem about an existing method, the arithmetic mean of
\citet{vovk2021} and \citet{hartoglei2025}: it is the exact worst-case optimum, so every interior
point of the family is a deliberate use of prior information, whose cost \S\ref{sec:regret} bounds.

At each horizon, the accumulated configuration mixture in~\eqref{eq:mixmart} is an exact likelihood
ratio, which permits comparison with all local e-processes.  Under Assumption~\ref{ass:indep},
\begin{equation}\label{eq:lr_identity}
E^w_{S,t}=\sum_c w_c L_{c,t}
=\frac{d\bar Q_w^{(t)}}{dP_0^{(t)}},
\qquad \bar Q_w^{(t)}=\sum_c w_c\,Q_c^{(t)},
\end{equation}
the likelihood ratio of the \emph{mixture-of-products} alternative against the null (ratios of
measures denote Radon--Nikodym derivatives throughout).  Since $\sum_cw_c=1$, $\bar Q_w^{(t)}$ is a
probability measure and $(E^w_{S,t})_t$ is a nonnegative mean-one martingale under the all-null law
of $S$.  For a horizon $t$ and design prior $\omega$, define on any local e-process $E$ the
cumulative expected-log objective
\[
\Psi_{t,\omega}(E):=\E_{\bar Q_\omega^{(t)}}[\log E_t]
=\sum_c\omega_c\,\E_{Q_c^{(t)}}[\log E_t].
\]
The optimality problem is posed in the local experiment, on e-processes adapted to $(\F^S_t)$.
If processes could use outside streams whose laws were correlated with the configuration inside $S$
under the design prior, those streams would carry information about it and change the Bayes
likelihood ratio.  Under any fixed configuration, which is the frequentist setting
of~\eqref{eq:fwerN}, the outside streams are independent of those in $S$ and carry no such
information; Proposition~\ref{prop:twosided} and Lemma~\ref{lem:ceiling} below therefore allow the
competing procedure to use all $K$ streams.

The generic point-null likelihood-ratio optimality step in the next theorem is the
stopping-rule-independent result of~\citet[Theorem~8]{koolengrunwald2022}.  The direct argument
establishes terminal uniqueness and identifies the configuration weights in the present product
experiment.

\begin{theorem}[Simultaneous finite-horizon Bayes optimality]\label{thm:bayes}
For every horizon $t\ge1$ and design prior $\omega=(\omega_c)_{c\in\mathcal C}$, among all
$(\F_s^S)_{s\ge0}$-adapted e-processes for the local null $P_0^S$, the objective
\[
\Psi_{t,\omega}(E)=\E_{\bar Q_\omega^{(t)}}[\log E_t]
\]
is maximized by the accumulated configuration mixture
\[
E_{S,t}^{\omega}=\sum_c\omega_c L_{c,t}
=\frac{d\bar Q_\omega^{(t)}}{dP_0^{(t)}}.
\]
One test martingale
$(E_{S,t}^{\omega})_{t\ge0}$ therefore attains the Bayes optimum at every horizon. The maximizing
terminal e-value is unique $P_0^{(t)}$-a.s.; under the standing assumption $D>0$, the product measures
$\{Q_c^{(t)}\}$ are linearly independent, so its representation in configuration weights is also
unique and equals $w=\omega$.
\end{theorem}
The proof is in
\suppref{Appendix}{app:bayes}{Appendix~F}.

Exchangeability therefore turns the design prior into ESP weights: beliefs about the number
of active streams determine how much weight to place on each polynomial degree. Familiar
merging rules correspond to particular choices of these beliefs, as the next two corollaries show.

\begin{corollary}[Exchangeable ESP form]\label{cor:esp}
If the design prior is exchangeable, $\omega_c=\omega_{l_c}$, the maximizer is the PEM statistic
$E_{S,t}=\sum_{l}\omega_l\,e_l(\{L_{k,t}\}_{k\in S})$ of~\eqref{eq:mixmart}, with level prior
$\beta_l=\binom{m}{l}\omega_l$: the Bayes-optimal level prior equals the design prior over the number
of alternatives, and the map from prior to weights is one-to-one.
\end{corollary}

\begin{corollary}[The product, the arithmetic mean, and the balanced test are each optimal for one
prior]\label{cor:corners}
The product $\prod_{k\in S}L_{k,t}$ is the maximizer of $\Psi_{t,\omega}$ for the prior
concentrated on the all-alternatives configuration, the arithmetic mean $m^{-1}\sum_{k\in S}L_{k,t}$
for the prior uniform on the singleton configurations, and the balanced test
$\tfrac12\tilde e_1+\tfrac12\tilde e_m$ of \citet{bharti2026} for the two-atom level prior
$\beta_1=\beta_m=\tfrac12$.  At $m=1$ the three statistics coincide, the two atoms of the balanced
prior fall on the same configuration and combine to total mass one, and the common optimum is
$L_{1,t}$ itself.  By the uniqueness in Theorem~\ref{thm:bayes}, none of them is optimal
for any other design prior; and because the theorem ranges over all local e-processes, no
data-adaptive reweighting of the PEM statistic improves on the fixed prior-matched weights either
(the self-financing conditions under which predictable weights may change at all are in
\suppref{Appendix}{app:imposs}{Appendix~I}).
\end{corollary}
Proofs of both corollaries are in
\suppref{Appendix}{app:bayes}{Appendix~F}.

\begin{theorem}[Finite-horizon minimax representation]\label{thm:grow}
At every horizon $t\ge1$, among all $(\F_s^S)_{s\ge0}$-adapted e-processes $E$ for the local null $P_0^S$,
\[
\sup_{E}\ \min_{c\in\mathcal C}\E_{Q_c^{(t)}}[\log E_t]
=\min_{\lambda\in\Delta(\mathcal C)}\KL(\bar Q_\lambda^{(t)}\|P_0^{(t)}),
\]
where $\Delta(\mathcal C)$ is the set of probability vectors on $\mathcal C$, attained by the
accumulated mixture with any
$w=\lambda^\star_t\in\argmin_{\lambda\in\Delta(\mathcal C)}
\KL(\bar Q_\lambda^{(t)}\|P_0^{(t)})$. This is the $Q\|P_0$ information projection over
$\mathrm{conv}\{Q_c^{(t)}\}$; projection terminology varies with the direction of KL\@.
\end{theorem}
Maximizing worst-case expected log evidence therefore reduces to finding the configuration
mixture closest to the null in KL divergence. The likelihood ratio for that mixture attains
the worst-case optimum. The proof is in
\suppref{Appendix}{app:grow}{Appendix~G}.

\begin{proposition}[The arithmetic mean is the horizon-uniform GROW solution]\label{prop:growmean}
Under Assumptions~\ref{ass:exch}--\ref{ass:indep}, for every horizon $t\ge1$ the minimizing mixture is
the uniform-singleton mixture
\[
\bar Q_{\lambda^\star}^{(t)}=m^{-1}\sum_{k\in S}Q_{\{k\}}^{(t)}.
\]
Thus the GROW-optimal local intersection e-process is, at every horizon,
\[
\bar e_1(t)=m^{-1}\sum_{k\in S}L_{k,t}.
\]
The minimizing $\lambda$ is the uniform singleton distribution; under $D>0$ its mixture
representation is unique.
\end{proposition}

Singleton configurations are the least favorable, so moving into the interior of the ESP simplex is
a deliberate use of configuration information whose cost is bounded: Theorem~\ref{thm:bayes}
selects the Bayes-matched interior point, and Theorem~\ref{thm:regret} guarantees that any
full-support point lags behind the configuration-aware oracle's log evidence by at most
$\log(1/w_c)$ on every path.  Symmetry equalizes growth on singleton alternatives, while a size-bias
coupling in \suppref{Appendix}{app:grow}{Appendix~G} shows that every larger
configuration has at least that growth under the arithmetic mean, so higher-order ESP terms cannot
raise the worst-case value.  What is specific here is that the information projection is
closed-form and horizon-free in the exchangeable product experiment; the projection minimizes
$\KL(Q\|P_0)$ for the simple null, whereas the reverse-information-projection literature for
composite nulls uses the opposite optimization role~\citep{grunwald2024,larsson2024}.

\subsection{From evidence to rejection time: the rate gain from pooling}\label{sec:regret}

Pooling attains the fastest first-order local rejection delay as the error level tends to zero.
With the dimension, configuration and weights fixed, a full-support PEM statistic has log-growth
rate $lD$ when $l$ of the $m$ streams in $S$ are active. No valid procedure, even using all $K$
streams, improves its first-order rejection delay in this small-level limit.
The route is a pathwise oracle inequality, and it is one-sided: on every sample path the PEM
statistic's log evidence \emph{falls short}, by at most $\log(1/w_l^{(m)})$, of the log evidence
an oracle that knew the true configuration would hold. The PEM statistic therefore rejects no later
than that oracle run at a deflated level.  The mixture can exceed the unweighted oracle likelihood ratio, but the inequality
does not assert that it does so.  The inequality holds with
no expectation and no asymptotics, and its price splits into two additive parts: one for the prior
mass on the true cardinality, which the analyst controls, and one for spreading that mass over the
$\binom ml$ subsets of that size.  It is the reason full support is kept: a prior that gives every
cardinality positive mass pays a bounded constant for being wrong, whereas one that omits the true
cardinality can lose the exponential rate.

Fix a nonempty $S$, write $m=|S|$, and suppose
$w_l^{(m)}>0$ for every $l$. Let $\beta_l^{(m)}=w_l^{(m)}\binom{m}{l}$, so
$\sum_l\beta_l^{(m)}=1$. Define the hitting times
\[
\tau_S^E(a)=\inf\{t:E_{S,t}\ge1/a\},\qquad
\tau_c^L(a)=\inf\{t:L_{c,t}\ge1/a\},
\]
with the infimum of the empty set equal to infinity.

\begin{theorem}[Pathwise regret]\label{thm:regret}
For every $t$, every sample path, and every nonempty $c\subseteq S$ with $|c|=l$, using
$\log0=-\infty$,
\[
\log E_{S,t}\ \ge\ \log L_{c,t}-\log\tfrac1{w_l^{(m)}},
\qquad
\log\tfrac1{w_l^{(m)}}=\underbrace{\log\tfrac1{\beta_l^{(m)}}}_{\text{level-prior cost}}+
\underbrace{\log\tbinom{m}{l}}_{\text{which-subset cost}},
\]
and consequently $\tau_S^E(\alpha)\le\tau_c^L(\alpha w_l^{(m)})$ pathwise: the PEM statistic at level
$\alpha$ rejects no later than the oracle that knows $c$, run on $L_{c,t}$ at the deflated level
$\alpha w_l^{(m)}$.
\end{theorem}
The proof is in
\suppref{Appendix}{app:regretproof}{Appendix~H.1}.

The bound separates the cost of choosing a cardinality from the cost of not knowing its active
subset. For a fixed level mass, uniform allocation minimizes the worst which-subset cost within the mixture class.
Across all e-processes, $\log\binom ml$ is also unavoidable for a deterministic pathwise shortfall
that must hold uniformly over every horizon and every size-$l$ configuration; this is not a
finite-horizon or expected-delay lower bound. The geometric weights minimize the drift-normalized
worst-case surcharge $\max_l\log(1/w_l)/(lD)$, rather than exact expected delay.

Full support also makes the local log-regret bound asymptotically sharp: under a fixed
configuration $c$ with $w_c>0$, $E_{S,t}/L_{c,t}\to w_c$ and
$t^{-1}\log E_{S,t}\to lD$ almost surely. Omitting the true configuration can change that rate.
\suppref{Appendix}{app:regretproof}{Appendix~H.1} gives these derivations.
Local log regret does not determine elementary closure delay, where the singleton eventually
binds under full density (Proposition~\ref{prop:sandwich}).

The next result pairs the oracle inequality, an upper bound on the PEM statistic's rejection time, with a
lower bound on the rejection time of \emph{any} valid procedure, by the change-of-measure argument that underlies Wald's lower bound on
the expected sample size of a sequential test~\citep{wald1945}, and evaluates the corners.  Fix a configuration $\vec h$ of all $K$
streams whose restriction to $S$ is $c$ with $|c|=l\ge1$; the streams outside $S$ are arbitrary.
Write $H_S$ also for the set of configurations that are null on all of $S$, and put
$t_\alpha=\log(1/\alpha)/(lD)$ for the first-order local delay benchmark.

\begin{remark}[The asymptotic regime, stated once]\label{rem:regime}
Every first-order rate statement in this paper is a statement as $\alpha\to0$ with the dimension $K$,
the intersection $S$, the configuration $\vec h$ and the weight vectors $w^{(m)}$ \emph{held fixed};
only the level, and with it the horizon, varies.  Two rejection times must be distinguished.  The
\emph{uncapped} time $\tau_S^E(\alpha)=\inf\{t:E_{S,t}\ge1/\alpha\}$ monitors the flat Ville
threshold over an unbounded horizon and is what the asymptotics below describe.  The \emph{deployed}
time is $\inf\{t\le N:E_{S,t}\ge b_m\}$, with infinity on nonrejection, for a deterministic feasible
boundary $b_m\le1/\alpha$.  The same statements hold conditional on an independent calibration
whose chosen boundary satisfies this local validity premise.  The two times agree asymptotically
under the sufficient growing-horizon condition~\eqref{eq:growinghorizon} of Corollary~\ref{cor:rate}.
At a \emph{fixed} $N$ this rejection-time equivalence cannot hold: a rejection occurs by $N$ or its
time is infinity, whereas $t_\alpha\to\infty$.
Finally, $lD$ is information per synchronous round, a round supplying one observation from each of
the $m$ monitored streams, with the information supplied by its $l$ active streams; it is not a
rate per individual observation.  That is
why a rate of $lD$ per round is consistent with the single-observation ceiling of
\S\ref{sec:ceiling}, and it is the sense in which pooling saves rounds and therefore observations.
\end{remark}

\begin{proposition}[First-order rejection delay for an intersection]\label{prop:twosided}
Let $S$ have $m$ streams, $l\ge1$ of them active, and let $t_\alpha=\log(1/\alpha)/(lD)$, in the
regime of Remark~\ref{rem:regime}.
\begin{enumerate}
\item[(i)] \emph{Lower bound, any procedure.}  Let $(\sigma_\alpha)$ be $\F_t$-stopping times, one for
each level, and let $N_\alpha\ge t_\alpha$ be horizons such that
$\Pr_{\vec h'}(\sigma_\alpha\le N_\alpha)\le\alpha$ for every $\vec h'\in H_S$.  Then for every
$\epsilon\in(0,1)$,
\[
\Pr_{\vec h}\bigl(\sigma_\alpha\le(1-\epsilon)\,t_\alpha\bigr)\longrightarrow0
\qquad\text{as }\alpha\to0 .
\]
\item[(ii)] \emph{Upper bound, the uncapped PEM statistic.}  If $w_l^{(m)}>0$, then
$\limsup_{\alpha\to0}\tau_S^E(\alpha)/t_\alpha\le1$ holds $P_{\vec h}$-almost surely; hence
$\tau_S^E(\alpha)/t_\alpha\to1$ in $P_{\vec h}$-probability, and the uncapped PEM statistic rejects
$H_S$ at the fastest first-order rate any valid procedure can achieve.
\item[(iii)] \emph{The corners.}  The arithmetic mean has $t^{-1}\log\bar e_1(t)\to D$ almost
surely, so its uncapped rejection delay is $\log(1/\alpha)/D$ to first order at every $l$: a factor
$l$ slower than the bound in (i) whenever $l\ge2$.  The product has
$t^{-1}\log\prod_{k\in S}L_{k,t}\to lD-(m-l)\bar D$ almost surely.  When that drift is strictly
negative its supremum is almost surely finite, so its uncapped crossing probability tends to zero as
$\alpha\to0$; at a fixed $\alpha$ that probability is generally positive, not zero.  When the drift
is strictly positive and $l<m$, its first-order delay exceeds $t_\alpha$ by the factor
$lD/\{lD-(m-l)\bar D\}>1$.  In the nondegenerate zero-drift case $lD=(m-l)\bar D$ the log evidence is
a nondegenerate integrable mean-zero random walk, so it oscillates and crosses any fixed upper level with
probability one, but with no positive linear growth rate and therefore no first-order delay of the
form $\log(1/\alpha)/\text{rate}$.  Each corner attains the rate $lD$ at exactly one value of $l$:
the arithmetic mean at $l=1$, and the product at $l=m$.  This identifies the single active count at
which each corner is rate optimal; it is not a claim that either corner dominates or is dominated at
every other $l$ in finite samples.
\end{enumerate}
\end{proposition}
The proof is in
\suppref{Appendix}{app:localrateproofs}{Appendix~H.2}.

\begin{corollary}[The rate gain from multiple streams]\label{cor:rate}
For an intersection of $m$ streams of which $l$ are active, pooling the streams raises the
attainable exponential rate of evidence from $D$ to $lD$ per synchronous round and divides the
first-order rejection delay by $l$.  The PEM statistic with any full-support prior, in particular
PEM closure with either default, attains the rate $lD$ at every $l=1,\dots,m$ simultaneously, with no
knowledge of $l$.  No valid procedure, whether or not it uses the streams outside $S$, attains a
faster rate.  The arithmetic mean attains $lD$ only at $l=1$ and the product only at $l=m$.
Moreover, if the horizons satisfy
\begin{equation}\label{eq:growinghorizon}
\liminf_{\alpha\to0}\ \frac{N_\alpha}{\log(1/\alpha)/(lD)}\ >\ 1 ,
\end{equation}
then $\Pr_{\vec h}(\tau_S^E(\alpha)\le N_\alpha)\to1$, and because
the deterministic feasible boundary satisfies $b_m\le1/\alpha$, the deployed rejection time is no later than
$\tau_S^E(\alpha)$ on that event; hence the deployed time also satisfies the first-order statement of
Proposition~\ref{prop:twosided}(ii).  For an elementary claim the same condition is imposed with $D$
in place of $lD$.
\end{corollary}
The proof is in
\suppref{Appendix}{app:localrateproofs}{Appendix~H.2}.

For fixed $K$, the small-level limit gives a factor-$l$ reduction in first-order local rejection
delay: at $l=2$ it halves, and at $l=3$ it falls to a third.  Sections~\ref{sec:sims} and~\ref{sec:apps} illustrate it at finite levels, where
the boundary $b^\star$ adds a second, horizon-specific reduction on top of it; those are
finite-$\alpha$ illustrations of the pattern the theorem describes, not evidence for the asymptotic
statement itself.  What the rate gain does not do
is accelerate the rejection of a single hypothesis beyond its own stream's evidence; the next
subsection makes that ceiling precise.

\subsection{Elementary claims: the single-stream ceiling and the multiplicity penalty}\label{sec:ceiling}

No procedure that controls family-wise error rejects an individual hypothesis
$H_k$ by any deadline with higher probability than the best level-$\alpha$ sequential test of stream
$k$ alone \emph{for that deadline}.  The bound is a deadline-specific power envelope drawn from the
stream's own evidence: the supremum is taken separately at each $t$, and there need be no single
sequential rule that attains it at every $t$ at once, so the lemma does not assert the existence of a
uniformly most powerful sequential test.  The other $K-1$ streams therefore cannot buy rate for an
elementary claim; what they can buy is the removal of the multiplicity penalty, and closed testing
removes it entirely under dense evidence.  This separates the paper's two contributions.  The rate gain of
\S\ref{sec:regret} is available for intersection claims and is the reason to pool.  For elementary
claims the honest comparison with $e$-Bonferroni is in the additive constant, and Corollary~\ref{cor:penalty} below
gives that constant for every configuration.

\begin{lemma}[Single-stream ceiling]\label{lem:ceiling}
Let $\boldsymbol\tau$ satisfy~\eqref{eq:fwerN} under every configuration, and let $\vec h$ be a
configuration with $h_k=1$.  Let $\mathcal T_\alpha$ be the class of stopping times of the filtration
generated by stream $k$ and an independent randomization, with $\Pr_{f_0}(\sigma\le N)\le\alpha$,
where $\Pr_{f_0}$ and $\Pr_{f_A}$ denote the law of stream $k$ under $f_0$ and under $f_A$.
Then for every $t\le N$,
\[
\Pr_{\vec h}(\tau_k\le t)\ \le\ \sup_{\sigma\in\mathcal T_\alpha}\Pr_{f_A}(\sigma\le t).
\]
The right-hand side is an envelope, one number for each $t$, not the power function of a single
rule.  In particular $\Pr_{\vec h}(\tau_k\le N)$ is at most the power of the Neyman--Pearson test of
$f_0$ against $f_A$ at $N$ on stream $k$, and, by Proposition~\ref{prop:twosided}(i) with $S=\{k\}$,
$\Pr_{\vec h}(\tau_k\le(1-\epsilon)\log(1/\alpha)/D)\to0$ as $\alpha\to0$ whenever
$N\ge\log(1/\alpha)/D$.
\end{lemma}
The proof is in
\suppref{Appendix}{app:elementaryproofs}{Appendix~H.3}.

\begin{corollary}[The multiplicity penalty]\label{cor:penalty}
Let uncapped literal PEM closure run with fixed full-support weights at the threshold $1/\alpha$,
and let $\vec h$ be a configuration with active set $c$, $k\in c$, $|c|=l$.  Then pathwise
\[
\tau_k\ \le\ \max_{S\ni k}\ \tau^L_{c\cap S}\bigl(\alpha\,w^{(|S|)}_{|c\cap S|}\bigr),
\]
and consequently:
\begin{enumerate}
\item[(a)] $\tau_k/\bigl(\log(1/\alpha)/D\bigr)\to1$ in probability: PEM closure attains the
single-stream ceiling of Lemma~\ref{lem:ceiling} at first order, as does $e$-Bonferroni.
\item[(b)] The eventual log-boundary surcharge, made precise in Proposition~\ref{prop:sandwich}, is at most
$\max_{m\le K-l+1}\log(1/w_1^{(m)})$ for PEM closure, which equals
$\log(K-l+1)+\log(1/\beta_1^{(K-l+1)})$ for either default; it is zero when every hypothesis is false
($l=K$), in the sense that $\tau_k=\tau^L_{\{k\}}(\alpha)$ for all sufficiently small $\alpha$, almost
surely.  For $e$-Bonferroni the penalty is $\log K$ under every configuration.  For arithmetic-mean
closure the pathwise upper bound is also $\log K$, because $\bar e_1(t)\ge L_{k,t}/|S|$ for every
$S\ni k$.  This bound need not be an actual positive delay penalty, even under full density.
\end{enumerate}
\end{corollary}
The proof is in
\suppref{Appendix}{app:elementaryproofs}{Appendix~H.3}.

The log-boundary comparison is pathwise, and it is sharpest when stated as a sandwich on the
uncapped closure time rather than as an expected-delay expansion.

\begin{proposition}[Pathwise sandwich for the uncapped elementary closure time]\label{prop:sandwich}
Fix a configuration with active set $c$, $|c|=l\ge1$, and $k\in c$, and run literal closure with
fixed full-support weights at the common Ville threshold.  Write $x=\log(1/\alpha)$,
\[
T_k(x)=\inf\{t:\log L_{k,t}\ge x\},\qquad
C=\max_{m\le K-l+1}\log\bigl(1/w_1^{(m)}\bigr),
\]
and let $\tau_k(e^{-x})$ be the uncapped literal-closure rejection time of $H_k$, that is
$\max_{S\ni k}\inf\{t:E_{S,t}\ge e^{x}\}$ with no deadline.  Then, $P_{\vec h}$-almost surely, for
all sufficiently large $x$,
\[
T_k(x)\ \le\ \tau_k(e^{-x})\ \le\ T_k(x+C).
\]
\end{proposition}
The proof is in
\suppref{Appendix}{app:elementaryproofs}{Appendix~H.3}.

The sandwich makes the multiplicity comparison precise. Under dense evidence, $l=K$ gives
$C=0$, so the uncapped closure time eventually equals the singleton crossing time, almost surely
as $\alpha\downarrow0$. More generally, $C$ is a \emph{log-boundary surcharge}: closure at level
$\alpha$ is eventually no slower than the single stream at level $\alpha e^{-C}$. This is a
pathwise statement, not a second-order expansion of expected delay, which would require an
overshoot analysis.

Since $e$-Bonferroni crosses at $T_k(x+\log K)$, PEM closure is eventually no slower whenever
$C\le\log K$. This condition need not hold under sparsity: the level-uniform default has
$C=2\log(K-l+1)$, which is $2\log K$ when only one stream is active. Full-support PEM,
arithmetic-mean closure and $e$-Bonferroni share the first-order elementary delay, so that rate
alone does not order their finite-level performance.

The prior also explains the sparse-case cost relative to arithmetic-mean closure. For an
intersection with only one active stream, the PEM statistic's log evidence is asymptotically
lower by $\log(1/\beta_1^{(m)})$. Under dense evidence, its higher-order terms instead clear
multi-active intersections at rate $2D$ or more. These log-growth comparisons do not themselves
order finite-level rejection times. Default-specific calculations and further comparisons are in
\suppref{Appendix}{app:elementaryproofs}{Appendix~H.3}.

\subsection{The stopping rule and the rejection-time objective}\label{sec:boundaryshape}

The local PEM threshold rule does not solve the rejection-time objective~\eqref{eq:costobj} exactly: it stops when
a linear function of the sufficient coordinates crosses a flat boundary, and what it gives up is,
at $m=1$, only the flatness of the boundary and, at $m\ge2$, the flatness and the linear reduction
together.  At the horizon, the PEM statistic at its exact null quantile is the most powerful
terminal test against the prior mixture.  For a single stream, every positive-multiplier
Lagrangian optimum has a time-varying evidence boundary; attaining a prescribed constrained level
can additionally require randomization.  For two or more streams the continuation value need not
be a function of $E_{S,t}$ alone.

\begin{corollary}[Terminal Neyman--Pearson optimality at the design horizon]\label{cor:np}
Fix the horizon $N$ and a configuration prior $w=(w_c)_{c\in\mathcal C}$.  Among all level-$\alpha$
tests of $H_S$ based on $\F_N^S$, the most powerful test against the mixture alternative
$\bar Q_w^{(N)}$ rejects when $E^w_{S,N}>c_\alpha$ and randomizes when $E^w_{S,N}=c_\alpha$, where
$c_\alpha$ is the $(1-\alpha)$ quantile of $E^w_{S,N}$ under $P_0^{(N)}$ and the randomization
probability is chosen so that the size is exactly $\alpha$; when $c_\alpha$ is not an atom the
randomization is vacuous and the test rejects on $\{E^w_{S,N}\ge c_\alpha\}$.  If $w$ is
exchangeable, $E^w_{S,N}$ is the ESP form of~\eqref{eq:mixmart} monitored throughout the paper.
\end{corollary}
The proof is in
\suppref{Appendix}{app:stoppingproofs}{Appendix~K.1}.

The deployed local rule in \S\ref{sec:calibration} instead rejects on a crossing of a boundary
$b_m$.  For a deterministic feasible choice with $b_m\le1/\alpha$, one has
$c_\alpha\le b_m\le1/\alpha$.  The terminal Neyman--Pearson test bounds the by-$N$ mixture power
of any level-$\alpha$ local stopping rule, since its rejection event is $\F_N^S$-measurable.
This power bound does not make the crossing event a subset of the terminal rejection region:
evidence can cross early and subsequently decrease.  Lowering a feasible flat threshold advances
crossings, but does not turn the rule into the terminal Neyman--Pearson test.  The terminal
criterion is commensurable with fixed-sample optimal-FWER policies~\citep{rosset2022,dubeyhuo2026};
Remark~\ref{rem:degenerate} explains why it cannot by itself justify stopping before $N$.

Because $\sigma\wedge N$ is bounded and $E^\omega_{S,t}=d\bar Q_\omega^{(t)}/dP_0^{(t)}$, the Wald
likelihood identity rewrites the objective in~\eqref{eq:costobj} under the null measure.  Define the
reward to be zero when the claim is never rejected, and evaluate the evidence at the bounded
stopping time $\sigma\wedge N$, which is defined on every path including $\{\sigma=\infty\}$:
\begin{equation}\label{eq:waldreduction}
\E_{\bar Q_\omega}\big[(N-\sigma)^+\big]
=\E_{P_0}\big[E^\omega_{S,\sigma\wedge N}\,(N-\sigma\wedge N)\big].
\end{equation}
When $m\ge2$ the optimal region need not be a level set of $E^\omega$ at all: the stopping payoff
depends on the state only through $E^\omega_{S,t}$, while the continuation value need not.
The distinction arises once a multiplier $\kappa>0$ for the constraint turns the problem into an
unconstrained stopping problem with payoff $g_t=E^\omega_{S,t}(N-t)-\kappa$.

\begin{theorem}[Single-stream boundary shape]\label{thm:levelset}
Let $m=1$ and consider~\eqref{eq:costobj} in its Lagrangian form, with stopping payoff
$g_t(E)=E\,(N-t)-\kappa$ for $\kappa>0$ and the option never to stop.  Then for each $t<N$ the
optimal stopping region is $\{E_t\ge b_t\}$ for a finite threshold $b_t$, and $b_N=\infty$: the
optimal rule is a level set of the evidence with a time-varying boundary, and it never exercises at
the deadline.
\end{theorem}
The proof is in
\suppref{Appendix}{app:stoppingproofs}{Appendix~K.1}.

The single-stream boundary characterization need not extend to $m\ge2$, because given $\F_t$
the future of $E^\omega_{S,\cdot}$ depends on the
posterior over configurations and not on $E^\omega_{S,t}$ alone.  Theorem~\ref{thm:levelset} is a
statement about the \emph{shape} of the optimum at $m=1$, and it says the optimal boundary varies
with $t$; it therefore does not certify the local PEM rule's flat boundary, and in
particular $b_N=\infty$ means the optimum never exercises at the deadline whereas the local PEM
threshold rule can.  When the exact minimum is attained, the local PEM rule at $b^\star$ is optimal among the
deterministic flat-threshold rules on $E^\omega_{S,t}$ with the stated crossing convention, since
the least feasible representative gives the earliest crossing event.  It solves neither the time-varying level-set version nor the
unrestricted version of~\eqref{eq:costobj} exactly.  It attains the first-order rate of
Proposition~\ref{prop:twosided}, which no rule can beat, and \S\ref{sec:dpvalidation} reports the
feasible policies and primal--dual bounds for the solved Bernoulli problems, together with
numerical Gaussian comparisons.

\section{Comparison with existing methods}\label{sec:phase}

We compare PEM closure with its product and arithmetic-mean corners, with $e$-Bonferroni, which
skips closed testing, and with the unrestricted local stopping optimum as a benchmark.

\subsection{The corners forgo the rate gain}\label{sec:corners}

Product closure attains the rate $lD$ only when every stream in the
intersection is active, and it loses positive log growth when null streams are sufficiently numerous.  At the full intersection $S=[K]$
the point mass on the all-alternatives configuration specializes~\eqref{eq:mixmart} to the product
$E_t=\prod_{k\in[K]}L_{k,t}$, its unique Bayes-optimal e-process (Corollary~\ref{cor:corners}),
whose per-step log growth under $l$ alternatives is $lD-(K-l)\bar D$
(Proposition~\ref{prop:twosided}(iii), Table~\ref{tab:streams}): below the attainable rate $lD$
whenever a single null is present, and negative once the nulls outnumber the actives by more than
the factor $D/\bar D$.  For equal-variance Gaussian shifts this is $(2l-K)D$, negative below
half-density, and the product then rejects with probability tending to zero as $\alpha\downarrow0$.  This is the sequential consequence of targeting a
dense configuration: the higher-order terms that produce the rate gain are all the product has, and
they carry the nulls with them.

Arithmetic-mean closure keeps a strictly positive evidence rate at every configuration with an
active stream, so it crosses any fixed threshold eventually with probability one, and it is the
exact worst-case expected-log optimum (Proposition~\ref{prop:growmean}); but its evidence for an intersection
grows at rate $D$ whatever the number of active streams, so it forgoes the factor $l$ entirely.  Under $l$
alternatives its first-order full-intersection delay is $\log(1/\alpha)/D$ against the PEM statistic's
$\log(1/\alpha)/(lD)$.  The two agree at $l=1$, the least-favorable configuration at which
the mean also attains the first-order lower bound, and separate by a growing
margin as $l$ increases.

A companion fixed-sample study in preparation reports the same tradeoff~\citep{dubeyhuo2026prior};
the present results do not depend on it.  When sparse configurations matter, the oracle inequality
motivates positive mass at every cardinality; when they do not, the product is the Bayes optimum
for the prior that says so.

\subsection{\texorpdfstring{$e$}{e}-Bonferroni pays the multiplicity penalty}\label{sec:ebonf}

At the full intersection, full-support PEM has first-order rate $lD$, whereas rejecting when any
$e$-Bonferroni marginal crosses has rate $D$.  For elementary claims, the dense log-boundary
comparison instead suggests the drift benchmark $1+\log K/\log(1/\alpha)$, which tends to one for
fixed $K$; it is not a proved delay ratio.  The observed finite-level ratios are in
\suppref{Table}{tab:kscale}{Table~S6}.  For geometric weights
$w_l=\gamma_K^l$, $\gamma_K=2^{1/K}-1$, the drift-plus-oracle-constant proxy for the
$l$-alternative full-intersection crossing time is
\[
\frac{\log(1/\alpha)+l\log(1/\gamma_K)}{lD},
\]
the rate $lD$ of Proposition~\ref{prop:twosided} with the oracle-inequality constant of
Theorem~\ref{thm:regret}.  It is no larger than the corresponding $e$-Bonferroni proxy
$\log(K/\alpha)/D$ exactly when
\[
\log\!\frac1{K\gamma_K}\le(1-1/l)\log\!\frac1\alpha .
\]
Because $K\gamma_K\ge\log2$, this holds for every $K$, $l\ge2$, and $\alpha\le1/3$; for $l=1$ and
$K\ge2$, $\gamma_K<1/K$ and $e$-Bonferroni is favored, the two coinciding when $K=1$.  Simulation
shows the same sparse-to-dense transition
(\suppref{Table}{tab:competitors}{Table~S7}).  Elementary
rejection under sparsity is harder, as Corollary~\ref{cor:penalty}(b) says: closure must also clear
the null-padded intersections, whose penalty is $\log(1/w_1^{(K-l+1)})$, so both full-support closure
and $e$-Bonferroni retain an order-$\log K$ term there and neither has a first-order advantage.

\subsection{The cost of the linear reduction and the flat boundary}\label{sec:dpvalidation}

For finite-state Bernoulli experiments, backward induction on the full $m$-dimensional state
brackets the unrestricted optimum of~\eqref{eq:costobj} against the calibrated local PEM threshold
rule.  Certified upper gains are below $1.3\%$ for the solved NCI intersections ($N=31$,
$m=2,\dots,4$) and below $6.7\%$
for the planned five-basket design ($N=13$, $m=2,\dots,5$).  These are the designs' actual horizons;
the larger intersections of the $K=10$ NCI design remain unsolved, and no common-horizon benchmark
is claimed for the unequal-cap realized five-basket design.  Gaussian comparisons reach $m=3$
and are numerical approximations.  The tables and computational conventions are in
\suppref{Appendix}{app:dpvalidation}{Appendix~K}.

\paragraph{What is certified, and how.} For a Lagrangian-optimal policy at multiplier
$\kappa\ge0$, let $J_\kappa$ and $p_\kappa\le\alpha$ be its reward and null crossing probability.
With $J_{\rm flat}$ the feasible flat rule's reward, weak duality brackets the constrained optimum
$J^\star$:
\begin{equation}\label{eq:dpbracket}
\max(J_{\rm flat},J_\kappa)\ \le\ J^\star\ \le\ J_\kappa+\kappa\,(\alpha-p_\kappa).
\end{equation}
We report the upper endpoint as a percentage gain over $J_{\rm flat}$, certified by
outward-rounded interval arithmetic. A feasible lattice policy can leave error budget unspent,
so its reward alone need not attain the constrained optimum.
\suppref{Appendix}{app:localbenchmarks}{Appendix~K.2} gives the
policy-feasibility and numerical-certification details.

\paragraph{The Gaussian panel is a numerical approximation.} Feasible-policy gains range from
about $12.5\%$ at $N=15$ to below $0.1\%$ at $N=100$. These lattice comparisons do not certify
continuous-model optimality bounds, and short-horizon gains are sensitive to grid spacing.
\suppref{Appendix}{app:localbenchmarks}{Appendix~K.2} reports approximation
details, sensitivity checks and an independent continuous-model evaluation of one policy, which
checks its performance rather than optimality. We do not extrapolate to the $N=2500$ simulations.

\paragraph{What the gap measures.}  At $m=1$ only flatness restricts the rule
(Theorem~\ref{thm:levelset}); for $m\ge2$ the comparisons assess flatness and scalar reduction
together.  The local PEM rule already attains the first-order rate
(Proposition~\ref{prop:twosided}); these comparisons supply no second-order delay expansion.
We deploy the calibrated flat rule; the dynamic program benchmarks its local cost.

\section{Simulations}\label{sec:sims}

Four Gaussian experiments examine local pooling, elementary delays, error control and prior sensitivity. We use
\(f_0=\mathcal N(0,1)\), \(f_A=\mathcal N(\mu,1)\) with mean shift \(\mu=0.5\), and
\(\alpha=0.05\), so \(D=0.125\) and the drift-based single-stream crossing benchmark is
\(t_{\mathrm{cross}}=\log(1/\alpha)/D=24.0\) observations. This first-order approximation omits
overshoot and does not equal a finite-level expected stopping time.
We study sparsity and prior sensitivity at \(K=10\) (\(3{,}000\) paths per configuration,
horizon \(2{,}500\)); dense scaling under all alternatives for \(K=2\) to \(12\)
(\(6{,}000\) paths); and literal closure at \(K=6\) over all \(63\) intersections and configuration
classes (\(20{,}000\) paths, horizon \(2{,}500\)).

Every ESP merge is monitored at its own Monte Carlo boundary \(\widehat b\): \(20\)--\(33\%\)
below \(1/\alpha\) over the six sizes of the \(K=6\) level-uniform closed test, and \(39.6\%\),
\(16.8\%\) and \(57.0\%\) for the level-uniform, arithmetic and product merges of the \(K=10\)
sweep.  The matched-level reference uses only the normalized ESP \(\tilde e_l\) corresponding to the
true active count \(l\); it knows that count but not the active subset. It retains threshold
\(1/\alpha\), while \(e\)-Bonferroni retains \(K/\alpha\).  In the sparsity sweep and the \(K=6\) closed test, PEM
closure uses the level-uniform prior \(\beta_l^{(m)}=1/m\) rather than the geometric default of the
applications, exposing the cost of equal mass over cardinalities. Dense scaling uses the product
corner, Bayes optimal for the all-alternatives configuration (Corollary~\ref{cor:corners}).
Complete designs, boundaries, seeding conventions and results are in
\suppref{Appendix}{app:simdetails}{Appendix~L}.

Pooling divides the full-intersection delay by roughly the number of active streams, the
arithmetic mean does not, and the product stalls under sparsity, as Proposition~\ref{prop:twosided}
predicts.  In the sparsity sweep the median first crossing of the full-intersection statistic falls
from \(46\) at \(l=1\) to \(5\) at \(l=7\) and \(3\) at \(l=10\) for PEM closure, tracking the
product's \(4\) and \(2\), whereas the arithmetic mean falls only from \(37\) to \(15\) and \(13\).
The product crosses by the horizon on fewer than \(28\%\) of paths at \(l\le3\), where its log drift
is negative, and PEM closure crosses on every simulated path at every displayed \(l\). PEM closure is not uniformly
fastest: at \(l=1\) the arithmetic mean leads, \(37\) against \(46\), consistent with the arithmetic mean's expected-log advantage at the least-favorable
configuration, although Proposition~\ref{prop:growmean} does not order stopping times.
The level-uniform mixture also assigns the singleton configuration less prior mass; the
geometric prior reduces the observed gap.  Figure~\ref{fig:phase} shows the
transition from sparse to dense evidence.

\begin{figure}[!tbp]
\centering
\includegraphics[width=0.95\linewidth]{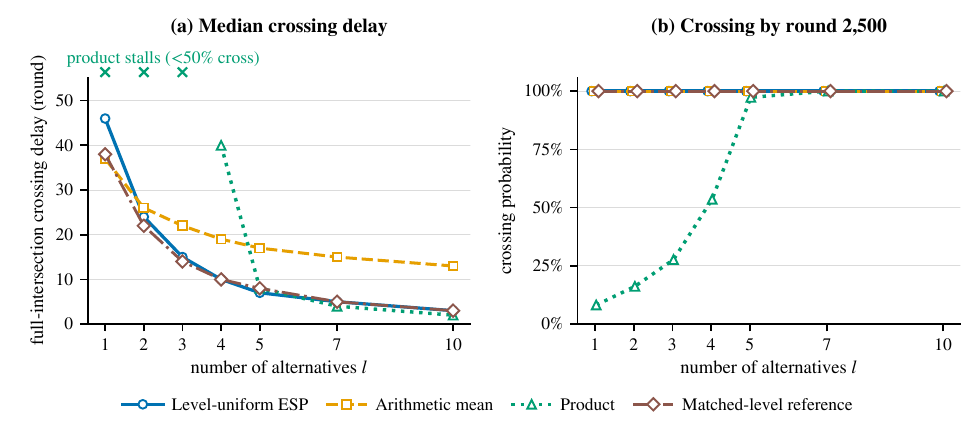}
\caption{Pooling divides the full-intersection delay by roughly the number of active streams;
the arithmetic mean does not benefit, and the product stalls below half density.
Level-uniform PEM statistic and its comparators at
$K=10$ ($3{,}000$ common-random-number paths per active count, horizon $2{,}500$).
Each merge is monitored at its own Monte Carlo boundary $\widehat b$; the matched-level
reference keeps the flat threshold by construction.  Panel (a) reports the censor-aware median
full-intersection crossing delay; product medians are deliberately omitted and marked as stalls when
fewer than half the paths cross.  Panel (b) reports crossing probabilities.  The exact counts,
intervals, and additional comparator values appear in
\suppref{Table}{tab:competitors}{Table~S7}.}
\label{fig:phase}
\end{figure}

At the elementary level, the dense-scaling study illustrates small closure costs and a growing
advantage over \(e\)-Bonferroni as \(K\) increases. These finite-level calibrated comparisons
complement the Ville-boundary asymptotics in Corollary~\ref{cor:penalty} and do not directly
measure the deadline-specific envelope in Lemma~\ref{lem:ceiling}.  Run at the product corner, the
mean rejection delay of a hypothesis is \(25.1\) at \(K=2\) and \(29.5\) at \(K=12\), close to \(t_{\mathrm{cross}}=24.0\) and
nearly flat in \(K\), whereas \(e\)-Bonferroni rises from \(32.1\) to \(46.0\) as its penalty
\(\log K\) grows.  The delay ratio rises from \(1.28\) to \(1.56\), below the drift-based ratio
\(1+\log K/\log(1/\alpha)\) from \(K=4\) upward. That ratio omits calibration, boundary
overshoot and residual closure costs, so it is only a reference for these finite-level results.

The literal closed test's error estimates are consistent with its guarantee under feasible local
boundaries. Monte Carlo calibration supplies the confidence qualification stated in
Section~\ref{sec:calibration}; the subsequent simulations independently check the implemented rule. At \(K=6\) the largest FWER estimate over the six null-containing
classes is \(0.0461\), and its Bonferroni-adjusted one-sided \(95\%\) Clopper--Pearson upper bound
is \(0.0498<0.05\); rerun at the flat threshold \(1/\alpha\), the largest estimate is \(0.0369\),
about four fifths of the calibrated figure, so the lower boundaries use more of the nominal error budget.
Every false null is rejected by the horizon on every simulated path, and the pooled median
elementary rejection delay falls from \(43\) at \(l=1\) to \(22\) at \(l=6\).
This is consistent with reduced closure cost under dense evidence; it does not establish that
the singleton alone binds on every finite-level path or replace the theoretical guarantees.

Full-support priors share the asymptotic log-growth rate but differ in finite-sample delay.
The prior sweep illustrates this sensitivity rather than a universal ordering of the family.  In the prior sweep every full-support
member's full-intersection delay falls by a factor of at least four from \(l=1\) to \(l=10\), and
the arithmetic mean's by less than three, consistent with the rate comparison in
Corollary~\ref{cor:rate}.
The priors also differ at dense configurations: at \(l=10\) the
level-uniform prior crosses at \(3\) against the geometric prior's \(8\), while at \(l=1\) the two
differ by \(10\) rounds in the other direction, illustrating the observed sparse--dense tradeoff. The expected-log minimax result in
Proposition~\ref{prop:growmean} alone does not require this ordering of delays.  We retain the
geometric prior as the primary choice on the bound-minimax criterion, a design-time property of the
weights rather than a measured delay, and read the family as the sensitivity set.

\section{Applications}\label{sec:apps}

Three studies use literal closed testing with prespecified geometric weights
\(w_l^{(m)}=(2^{1/m}-1)^l\), with the hypothesis family, information clock, marginal processes and
reporting rule frozen before confirmatory outcomes. We call this geometric PEM closure, or simply
``geometric.'' Basket-trial operating-characteristic simulations use finite-state calibrated
boundaries; the language-model and randomized-advertising replays use Ville's \(1/\alpha\).
Their marginal processes are not product-model likelihood ratios, and a null crossing law for
calibration is unavailable. Under their joint-validity assumptions, the FWER guarantee transfers;
product-model optimality need not. The true configurations in the real-data studies are unknown:
observed evidence patterns do not establish the model's dense or sparse configurations. Each study
reports full-intersection crossings and elementary rejections separately (\S\ref{sec:setup}).
The online supplement gives complete protocols, marginal-process constructions, audit identities
and result grids.

\subsection{Multi-cohort basket trials}\label{sec:basket}

Basket trials test a treatment across several disease cohorts, each contributing a separate
response stream. The primary design follows the National Cancer Institute's Molecular Analysis
for Therapy Choice (NCI-MATCH) program: \(K=10\) cohorts, nominal level \(0.05\),
a null response rate \(p_0=0.05\), a
working alternative \(p_{\rm alt}=0.25\), and a cap of 31 analyzable patients per cohort, motivated
by the signal-seeking subprotocol design~\citep{flaherty2020,odwyer2023}; two five-basket designs
reproduce the planned and realized cohort sizes of~\citet{maulikzhou2026} at \(p_0=0.15\),
\(p_{\rm alt}=0.45\), and level \(0.10\).  Each cohort's marginal evidence is the Bernoulli
likelihood ratio against the working alternative, a supermartingale throughout the composite null
\(p_k\le p_0\), and the boundaries are computed on finite response-count states. Strong FWER control requires
feasibility for every true-null subset's sampling profile, as detailed in the supplement;
thinned-accrual simulations also serve as sensitivity checks and do not by themselves certify
arbitrary asynchronous sampling.
The NCI grid spans one to ten active cohorts. These are prospective operating-characteristic
simulations, not patient-level reanalyses, over 42
nonnull configurations and six million audited paths
(\suppref{Appendix}{app:basketresults}{Appendix~P}).

The comparisons measure power and patient costs separately. Geometric PEM closure's marginal-power
point estimate exceeds product closure's in every one of the 42
nonnull configurations, and its capped mean full-intersection crossing round is smaller than
arithmetic-mean closure's in 41 of them; the exception is the single-active-cohort NCI
configuration at $p=0.25$, where geometric trails by $0.012$ rounds.  These are paired point-estimate
comparisons under common random numbers; the frozen aggregates do not retain paired-difference
variances, so the counts are descriptive and are not significance tests. Against Bonferroni at the
protocol's prespecified interim analyses, geometric also has a higher marginal-power point estimate
in all 42 nonnull configurations: \(0.795\) versus \(0.458\) in the dense five-basket configuration.
Across the 41 deterministic-accrual configurations with identifiable costs, the median reduction in
mean patient outcomes to the \emph{first} correct rejection is \(31\%\), charging paths with no
correct rejection the full design total (\suppref{Table}{tab:basketsavings}{Table~S14}).
This capped cost does not measure total accrual after individual stream discontinuations, and the
comparison concerns the specified interim-look rule rather than optimized group-sequential designs
generally. The evaluated rule has no futility stopping for lack of response.

The exact boundary is what makes the error budget usable.  Across the simulated
null-containing configurations, the largest FWER estimates were \(0.0437\) at nominal \(0.05\) and
\(0.0886\) and \(0.0644\) at nominal \(0.10\).  Under the flat threshold the same three estimates
were \(0.0198\), \(0.0379\), and \(0.0343\): a single response multiplies a cohort's evidence by
\(p_{\rm alt}/p_0\), which is \(5\) in the NCI design and \(3\) in the two five-basket designs, so
the process overshoots \(1/\alpha\), and monitoring at Ville's threshold used only between a third and two
fifths of the nominal error budget in these simulations. At the primary geometric prior the deployed boundaries are
\(40\)--\(52\%\) below \(1/\alpha\) across the three designs and all intersection sizes, and every
operating characteristic above uses those boundaries; over the full set of closure priors, including
the product corner, the reduction ranges from \(38\) to \(80\%\)
(\suppref{Table}{tab:boundarytable}{Table~S2}).  The result is a prespecified robustness pattern, not
universal dominance: full-intersection and elementary rankings need not agree, and the published
Maulik--Zhou thresholds target terminal weak FWER (control under the all-null configuration)
rather than continuous strong FWER\@.

\subsection{Benchmark-stratified language-model evaluation}

Sequential benchmark evaluation asks whether a candidate language model improves on an incumbent
within each prespecified task stratum. We compare GPT-Neo 1.3B with GPT-Neo 125M and, as a
separate replication family, GPT-2 large with GPT-2 medium, on a fixed version of the HellaSwag
validation population, a multiple-choice sentence-completion benchmark~\citep{zellers2019,eleuther2021cards,radford2019}, with ten outcome-blind activity
strata per contrast, itemwise paired correctness differences, and a common 161-item prefix per
stratum. The null in each stratum is that the candidate's accuracy over its full finite population
is no greater than the incumbent's. We use level \(\alpha=0.05\) and betting fraction
\(\eta=0.25\), which sets the evidence multiplier's sensitivity to each paired outcome.
Because the population is fixed and revealed in uniform random order, each stratum's
evidence is an exact finite-population e-process
(\suppref{Proposition}{prop:finitepop}{Proposition~N.1})
rather than a likelihood ratio: the family-wise guarantee of Proposition~\ref{prop:fwer} transfers,
the optimality results do not, and the flat threshold \(1/\alpha\) is used.  The two contrasts are
separate families; no joint twenty-hypothesis claim is made.

Observed contrasts are positive in all twenty cells. Averaged over the ten strata, the balanced
prefixes give GPT-Neo 1.3B a paired-accuracy advantage of \(22.80\) percentage points and GPT-2 large
one of \(9.01\) points. These describe scored prefixes, not the full populations of up to 2,627 items;
the random-order construction supplies inference for those populations. At the intersection level,
which asserts only that \emph{at least one} stratum improves, geometric PEM closure reaches the
full-intersection claim at \(41.6\%\) and \(36.1\%\) lower mean cost than arithmetic-mean closure
across 2,000 within-prefix order replays, and on the prespecified random orders (the canonical replays) its
crossing uses \(79.5\%\) and \(54.7\%\) fewer scored items than the complete prefix.  At the
elementary level it rejects the finite-population nulls in \(10/10\) and \(9/10\) strata, never
later than arithmetic-mean closure in any of the 20 cells and strictly earlier in 17, and never
later than \(e\)-Bonferroni, which rejects only \(4/10\) nulls in the GPT-2 family under its fixed multiplicity correction.  The least-favorable shortcut reproduces literal
closure in all 20 cells. The single GPT-2 nonrejection is inconclusive; a full-intersection claim
does not identify which slices improve. Replay costs condition on the fixed scored prefixes and
are counterfactual: recorded forward passes were not avoided. The geometric prior was prespecified
(\suppref{Appendix}{app:llmprotocol}{Appendix~O}).

\subsection{Stratified randomized-advertising replay}

The CRITEO-UPLIFTv2.1 release records \(13{,}979{,}592\) users from randomized advertising
experiments~\citep{diemert2021criteo}.  We use the randomized treatment indicator, split the users
outcome-blind, fit a model to rank predicted treated-minus-control differences outside the
confirmatory holdout, and freeze ten
score strata before their outcomes are used. Within each stratum, treated and control users are
paired; under the prospective i.i.d.\ paired-arm model and its nonpositive-mean null, the betting
evidence constructed from paired outcome differences is a supermartingale
(\suppref{Appendix}{app:criteoresults}{Appendix~Q}).  The visit endpoint is
primary and conversion a separate secondary family, each tested at \(\alpha=0.05\).
One outer round processes one matched pair from each nonexhausted stratum. This is the concentrated-evidence regime, the
empirical counterpart of sparsity: across \(1{,}467{,}725\) matched holdout pairs the visit
difference is \(0.68\) percentage points, and the highest-score stratum holds \(8.9\%\) of the pairs
but \(74.0\%\) of the net paired-score total.  Which stratum-level nulls are in fact true is not
known, and nothing below assumes it.

Geometric PEM closure rejects the highest-score visit stratum at outer round $3{,}441$, which is $34{,}410$
processed matched pairs out of $1{,}467{,}725$, or $2.344\%$ of the matched-pair replay, excluding
unmatched users. Product closure neither crosses the full-intersection threshold nor
rejects any stratum.  Arithmetic-mean closure and $e$-Bonferroni also reject at round $3{,}441$.
The stratum's own level-$\alpha$ evidence first crosses at round $2{,}943$, leaving a closure delay
of $498$ rounds on this path. This is an observed singleton benchmark, not a pathwise lower bound
for every valid test. Meanwhile,
the level-uniform prior crosses the full intersection earliest but rejects the stratum only
after $109{,}900$ pairs. In the conversion
family all four exact closures reject the same four strata against two for $e$-Bonferroni, and
geometric's last rejection comes $10.7\%$ earlier than arithmetic's in outer rounds; measured
instead in processed pairs, which charges only strata that are not yet exhausted, the reduction is
$9.9\%$.  These are outcome-blind random-order replays, since the public data carry no timestamps, and causal
interpretation requires the original randomization and sampling assumptions.

All ten empirical visit contrasts are positive, and the terminal paired-sign Holm comparator
rejects seven under its own i.i.d.\ sign model. Thus a single sequential rejection does not
establish that the remaining nine nulls are true. Positive empirical contrasts do not ensure positive growth of a fixed betting process: eight
visit strata have negative empirical log growth, contributing to the product's poor performance
on this replay. \suppref{Appendix}{app:criteoresults}{Appendix~Q} gives
the calculation. This is a descriptive calculation; the sparse-configuration likelihood-ratio
theorem does not apply directly to these betting processes.

\FloatBarrier
Table~\ref{tab:appsummary} collects the headline results and their comparison-specific scope.

\begin{table}[!tbp]
\centering
\footnotesize
\setlength{\tabcolsep}{3pt}
\caption{Headline results from the three prespecified audited application studies.  ``Earlier,'' ``lower,''
and ``above'' refer to the prespecified point estimates or frozen replay paths, not universal
dominance.  The language-model percentages compare conditional within-prefix replay means; basket rows are
Monte Carlo operating characteristics; Criteo is one canonical random-order replay.}
\label{tab:appsummary}
\begin{tabular}{@{}L{2.2cm}L{2.9cm}L{2.3cm}L{6.2cm}@{}}
\toprule
Study & Frozen contrast & Geometric PEM closure & Main comparative finding\\
\midrule
GPT-Neo & \(1.3\)B vs.\ 125M; \(+22.80\) pp
& \(10/10\) strata
& No later than arithmetic in \(10/10\) cells; replay full-intersection cost \(41.6\%\) lower than arithmetic\\
GPT-2 & Large vs.\ medium; \(+9.01\) pp
& \(9/10\) strata
& Arithmetic rejects \(8/10\), $e$-Bonferroni \(4/10\); replay full-intersection cost \(36.1\%\) lower than arithmetic\\
Basket grids & 42 nonnull configurations; six million paths
& Continuous monitoring
& Marginal power above product in \(42/42\) and full-intersection crossing time below arithmetic in \(41/42\); first-correct cost below product and planned interim analyses in \(41/41\)\\
Criteo visit & \(1{,}467{,}725\) pairs; \(+0.6796\) pp
& \(1/10\) strata
& Product rejects \(0/10\); geometric rejection uses \(2.344\%\) of matched pairs\\
Criteo conversion & Separate secondary family; \(+0.09709\) pp
& \(4/10\) strata
& \(e\)-Bonferroni rejects \(2/10\); last rejection \(10.7\%\) earlier in outer rounds than arithmetic\\
\bottomrule
\end{tabular}
\end{table}

The same prespecified full-support rule accommodates these different evidence patterns, with
different local and elementary benefits in each; neither the simulations nor the replays establish
universal dominance.

\section{Discussion}\label{sec:disc}

Pooling and elementary identification face different limits. In the product model, $l$ active
streams yield optimal first-order local delay $\log(1/\alpha)/(lD)$, while elementary rejection
faces a deadline-specific single-stream power envelope. Literal full-support PEM closure at Ville's
boundary eventually attains the singleton crossing pathwise under dense alternatives; its general-configuration sandwich has a fixed
log-threshold width. These results do not order finite-level rejection times. Horizon calibration
offers a separate improvement for a specified design.

Deployment begins with joint validity, in this order. Among valid ESP mixtures, use a scientific configuration
prior when defensible, or a full-support default with prespecified sensitivity checks. Next,
calibrate against the design's null crossing law or use Ville's $1/\alpha$. Use literal closure
when feasible; at larger $K$, label the least-favorable shortcut conservative and audit it on
tractable instances. Expected-log optimality does not order all multiple-testing policies by
power: terminal-power optimization and group-sequential design address different objectives.
The weights' posterior interpretation also requires the stated mixture model.

Conditional independence both validates higher products and identifies configuration likelihood
ratios, supporting the Bayes, worst-case and rate results. Under arbitrary contemporaneous
dependence, full-filtration marginal e-processes still support arithmetic-mean closure and
$e$-Bonferroni. Admissible dependence-robust merges are weighted arithmetic means, allowing weight
on the constant one; under permutation symmetry these mix $\bar e_1$ and
$1$~\citep{vovk2021,wang2025}. Higher products need joint conditional-moment conditions; a known
dependent null may instead supply a joint density process without ESP factorization
(\suppref{Appendix}{app:dependence}{Appendix~J}). After changing the sampling or dependence law, establish
intersection validity before recalibrating or using Ville: that threshold cannot repair an invalid
process. Dependence robustness applies to Ville monitoring; reduced boundaries need calibration
under the allowed joint law. Valid application-specific intersection processes transfer error control, but not automatically the
likelihood-ratio model's optimality, rates or logarithmic multiplicity costs. Staggered arrivals,
composite alternatives and weighted error criteria likewise need new validity
arguments~\citep{songfellouris2019}.

The horizon enters the results unevenly. Martingale and Ville validity need none; sufficiency and
expected-log optimality hold at each fixed time, and the oracle inequality holds pathwise. Rates
use an uncapped small-level limit or a growing horizon containing the crossing. The savings
objective, terminal Neyman--Pearson comparison, stopping benchmark and calibrated boundaries use
$N$ explicitly. A possible extension therefore requires prospective calibration for the extended
horizon or Ville monitoring from the outset.
Beyond the calibrated horizon, the generic bound is only $1/b_m>\alpha$ when $b_m<1/\alpha$.

Section~\ref{sec:sota} and Table~\ref{tab:streams} locate the contribution within established
merging and closed-testing frameworks. The evidence arguments specialize likelihood-ratio log
optimality and GROW theory~\citep{koolengrunwald2022,grunwald2024,larsson2024,ramlarsson2026,sahaharamdas2026}.
Related limits concern asymptotic evidence optimality~\citep{ramramdas2026}, changes of
filtration~\citep{choeramdas2025}, and sequential change of measure~\citep{wald1945}.

Three questions remain. First, for $m\ge2$, when is a time-varying level set of $E^\omega$ optimal,
and when does the full state improve constrained rejection time? Theorem~\ref{thm:levelset} treats
only the one-stream Lagrangian problem; the numerical brackets combine scalar reduction and
flatness rather than separate them. Second, a general second-order theory of elementary expected
delays is missing. The pathwise sandwich supplies neither an exchange of limits and expectations
nor a sharp renewal constant. The uniform-in-time oracle comparison with every size-$l$
configuration forces a $\log\binom ml$ log-evidence loss, but does not establish that lower bound for expected
delays. Third, can the sufficient-state reduction extend beyond the product experiment?
Dependence can destroy its factorization, and even here the fixed-time ESPs omit the permanent
intersection marks needed to summarize a closed-testing history.

\fi 

\ifavmain\ifavsupp
  \clearpage
  \bibliographystyle{imsart-nameyear}
  \IfFileExists{references.bib}{\bibliography{references}}{\bibliography{manuscript/references}}
  \global\avbibdonetrue
\fi\fi

\ifavsupp
\ifavmain\clearpage\fi
  \begin{frontmatter}
  \title{Supplementary Material for\\``\papertitle''}
  \runtitle{Supplement: \runningtitle}

  \ifdefined\AVBLIND
\begin{aug}
\author{\fnms{}~\snm{Anonymous}}
\end{aug}
  \else
\begin{aug}
\author[A]{\fnms{Prasanjit}~\snm{Dubey}\ead[label=e1]{pdubey31@gatech.edu}}
\and
\author[A]{\fnms{Xiaoming}~\snm{Huo}\ead[label=e2]{huo@gatech.edu}}
\address[A]{H.~Milton Stewart School of Industrial and Systems Engineering,
Georgia Institute of Technology,
\printead{e1,e2}}
\end{aug}
  \fi

  \end{frontmatter}

  \noindent This supplement contains the exact algorithm, complete proofs and supporting
  derivations, additional simulations, the finite-population construction, and the three application
  audits.  Equation, table, figure and algorithm numbers carry an S prefix; theorem-family results
  are numbered within the appendix that states them; a number carrying neither prefix is the main
  article's. Notation follows the main article: $K$ is the number
  of hypotheses, $S$ an intersection of size $m$, $L_{k,t}$ the marginal evidence in stream $k$,
  $N$ the monitoring horizon, and $\alpha$ the error level. PEM closure combines the
  prior-weighted elementary-symmetric-polynomial (ESP) mixtures by closed testing; geometric
  PEM closure uses weights $w_l^{(m)}=(2^{1/m}-1)^l$.
  \medskip
  \appendix
  \setcounter{equation}{0}
  \setcounter{table}{0}
  \setcounter{figure}{0}
  \setcounter{algorithm}{0}
  \renewcommand{\theequation}{S\arabic{equation}}
  \renewcommand{\thetable}{S\arabic{table}}
  \renewcommand{\thefigure}{S\arabic{figure}}
  \renewcommand{\thealgorithm}{S\arabic{algorithm}}

\noindent\textbf{Roadmap.}
Appendix~\ref{app:algorithm} gives the exact closed-testing algorithm.
Appendices~\ref{app:spending}--\ref{app:dependence} collect the supporting derivations, proofs,
self-financing analysis, and dependence extension. Appendix~\ref{app:dpvalidation} gives the
stopping-rule proofs, finite-state stopping certificates, and Gaussian numerical benchmarks,
Appendix~\ref{app:simdetails} the boundaries, tables, and
conventions of the Gaussian simulations, and Appendix~\ref{app:hetero} the heterogeneity
diagnostics. Appendix~\ref{app:finitepop} establishes the
random-order construction used by the benchmark application. Appendices~\ref{app:llmprotocol}--\ref{app:criteoresults}
give the three application protocols, computational audits, and complete results; the small-GPT
outcome tables begin in Subsection~\ref{app:llmresults}. Administrative statements close the supplement.
\medskip

\section{Exact closed-testing pseudocode}\label{app:algorithm}

The main article defines both literal closure and the conservative latched
least-favorable shortcut.
Algorithm~\ref{alg:main} gives the literal procedure used for every reported exact-closure result; it is written with the flat threshold $1/\alpha$, and the horizon-$N$ results replace $1/\alpha$ by the deployed $b_{|S|}$ and stop at $t=N$, under the corresponding feasibility guarantee.

\begin{algorithm}[!tbp]
\caption{PEM closure}
\label{alg:main}
\begin{algorithmic}[1]
\State \textbf{input:} prespecified marginal evidence processes $\{L_{k,t}\}$ whose required
intersection products are full-filtration e-processes; level $\alpha$; nonnegative weights $w^{(m)}$ for each $m$,
$\sum_{l=1}^m\binom{m}{l}w_l^{(m)}=1$
\State initialize $L_{k,0}\gets1$ for all $k$; $\mathcal A\gets[K]$ (unrejected hypotheses);
  $\mathcal M\gets\varnothing$ (permanently marked intersections)
\For{$t=1,2,\dots$}
  \For{$k\in\mathcal A$} \Comment{only unrejected streams are read; see the note below}
    \State update $L_{k,t}$ \Comment{baseline model: $L_{k,t}\gets L_{k,t-1}f_A(X_{k,t})/f_0(X_{k,t})$}
  \EndFor
  \For{each nonempty $S\subseteq[K]$ with $S\notin\mathcal M$} \Comment{$O(K^2)$ each, ESP recursion; marked $S$ are skipped}
    \State $E_{S,t}\gets\sum_{l=1}^{|S|} w^{(|S|)}_l\,e_l\big(\{L_{k,t}\}_{k\in S}\big)$
    \If{$E_{S,t}\ge1/\alpha$} $\mathcal M\gets\mathcal M\cup\{S\}$ \Comment{the mark is permanent} \EndIf
  \EndFor
  \State $\mathcal B_t\gets\varnothing$ \Comment{elementary hypotheses newly rejected at time $t$}
  \For{$k\in\mathcal A$}
    \If{every $S\ni k$ is marked rejected} $\mathcal B_t\gets\mathcal B_t\cup\{k\}$ \EndIf
  \EndFor
  \State reject every $H_k$ with $k\in\mathcal B_t$; $\mathcal A\gets\mathcal A\setminus\mathcal B_t$
  \Comment{and discontinue those streams}
\EndFor
\State \textbf{guarantee:} $\Pr_{\vec h}(\sup_t V_t>0)\le\alpha$ for every configuration
  (\mainref{Proposition}{prop:fwer})
\end{algorithmic}
\end{algorithm}

\paragraph{Why reading only the unrejected streams is enough.}  The loop reads $X_{k,t}$ only for
$k\in\mathcal A$, which is what makes the promised saving of $(N-\tau_k)^+$ observations real rather
than notional.  It costs nothing, because no unmarked intersection ever contains a discontinued
stream.  Suppose $H_k$ was rejected at $\tau_k$.  By the closure rule that is precisely the event
that every $S\ni k$ carried a mark at $\tau_k$, and marks are permanent, so at every $t>\tau_k$ every
$S\ni k$ lies in $\mathcal M$ and is skipped.  The inner loop therefore never asks for a likelihood
ratio it has stopped updating.  Two conventions make this exact rather than approximate: a mark is a
persistent bit, never re-derived from the current value of $E_{S,t}$, which may have fallen back
below the threshold; and a stream that is discontinued for any other reason, such as exhausting its
cap before the horizon, receives factor-one idle updates, which leave $E_{S,t}$ a full-filtration
e-process and therefore leave Proposition~\mainnum{prop:fwer} in force at the threshold $1/\alpha$.
Those idle updates do \emph{not} preserve a smaller calibrated boundary; see
Appendix~\ref{app:bstar}.

\paragraph{Latched versus current-time shortcut.} Requiring all $K$ least-favorable statistics to exceed their thresholds at one common time gives a current-time rule, which discards a statistic that crossed earlier and subsequently fell back.  The latched rule instead records each $s_{k,m}$ permanently.  If a stream is discontinued after a shortcut rejection, its likelihood ratio is frozen at its last observed value and continues to enter the sort at that value.  For the following comparison, evaluate the current-time rule on this same frozen array; with $\tau_k$ the literal closure time on the underlying data path,
\[
\tau_k\le\tau_k^{\mathrm{LF}}
\le\inf\bigl\{t\le N:E_{S_{k,m,t},t}\ge b_m\text{ for every }m=1,\dots,K\bigr\}.
\]
At time $s_{k,m}$ the least-favorable statistic crosses, so every size-$m$ intersection containing $k$ is then above threshold on the frozen array.  To compare with literal closure, induct on shortcut discontinuations.  Each previously discontinued stream has already been rejected by literal closure; all intersections containing it were therefore marked before freezing could affect them.  An intersection containing no such stream still has its original observations, so a new crossing is also a literal crossing.  Thus every $S\ni k$ is marked by $\tau_k^{\mathrm{LF}}=\max_ms_{k,m}$, proving the first inequality; the second is the first-crossing definition on the common array.  This argument also establishes validity for a current-time shortcut run with its own discontinuations, but does not compare two different freezing schedules pathwise.  The latched shortcut equals literal closure on a path exactly when $s_{k,m}\le\tau_k$ for every $m$.  Latching adds $K$ bits per hypothesis and no arithmetic.

\smallskip\noindent(This appendix is cited in \S\mainnum{sec:closure} of the main article.)

\section{Ville monitoring versus pointwise time-union calibration}\label{app:spending}

Ville monitoring rejects no later than pointwise time-union calibration on every sample path.  Both routes monitor the same evidence $E_t=L_{1,t}=\prod_{s\le t}r_{1,s}$, a nonnegative martingale
with $\E_{f_0}[E_t]=1$ under the null, and differ only in the rejection boundary. Ville monitoring rejects
at $\tau_{\mathrm e}=\inf\{t:E_t\ge1/\alpha\}$, valid since
$\Pr_{f_0}(\sup_tE_t\ge1/\alpha)\le\alpha$ by Ville. The pointwise time-union route, with
$\delta_t=6\alpha/(\pi^2t^2)$ so that $\sum_t\delta_t=\alpha$, rejects at
$\tau_{\mathrm s}=\inf\{t:E_t\ge1/\delta_t\}$ and is valid by a union bound over time. Since
$\delta_t\le\delta_1=6\alpha/\pi^2<\alpha$ for all $t\ge1$, we have $1/\delta_t\ge1/\alpha$, hence
\[
\tau_{\mathrm e}=\inf\{t:E_t\ge1/\alpha\}\ \le\ \inf\{t:E_t\ge1/\delta_t\}=\tau_{\mathrm s}
\qquad\text{on every sample path.}
\]

The boundary inflation is $\log(1/\delta_t)-\log(1/\alpha)=2\log t+O(1)$: this union-bound comparator
pays a $\log t$ penalty that Ville avoids. In $2\times10^5$ Gaussian simulations at
$\alpha=0.05$, every path crossed both boundaries by the specified horizon. At $\mu=1.0$ and
horizon $2{,}000$, the estimated mean delays were $7.44$ (standard error $0.01$) and $19.63$ ($0.02$), a
ratio of $2.64$; at $\mu=0.3$ and horizon $8{,}000$, they were $70.89$ ($0.13$) and $337.26$ ($0.31$), a
ratio of $4.76$.

\smallskip\noindent(This appendix is cited in \S\mainnum{sec:calibration} of the main article.)

\section{The horizon-\texorpdfstring{$N$}{N} boundary: attainment, schedules, and composite marginal nulls}\label{app:bstar}

\paragraph{Attainment.} Feasibility in \mainref{Equation}{eq:bfeasible} is what error control uses.
For finite-state enumeration the prespecified candidate grid is finite and contains the feasible
Ville threshold, so its feasible subset has a minimum.  For Gaussian increments, the running
maximum over $0\le t\le N$ is atomless above $1$.  If $\Pr_0(\max_{t\le N}E_{S,t}>1)>\alpha$,
the infimum of the feasible support candidates is above $1$, where tail continuity makes it
feasible.  The condition matters: for $m=N=1$, a Gaussian shift $\mu=4$ and $\alpha=.05$ give
$\Pr_0(L_{1,1}>1)=1-\Phi(2)<\alpha$, so every $b>1$ is feasible but $b=1$ is not.  Even this
Gaussian population candidate set has no minimum.  The implemented finite Monte Carlo candidate
set with its Ville fallback is well-defined in either case.  Restricting to the support of an
arbitrary law does not by itself ensure attainment.

Attainment is not automatic, and it does not follow from the feasible set being nonempty.  If the law has an atom at some $v$ whose mass carries the crossing probability from strictly above $\alpha$ to at most $\alpha$, with support extending continuously above $v$, then every $b>v$ is feasible while $v$ itself is not, and the feasible subset of $\mathcal V_{N,m}$ has no least element.  This is realizable inside the standing assumptions.  Let $m=1$, $N=1$ and $\alpha=0.15$, and let the per-observation likelihood ratio $r$ have null law with mass $0.7$ at $0.5$, mass $0.2$ at $2$, and mass $0.1$ spread uniformly over $[2,3]$.  Then $\E_{f_0}[r]=0.7(0.5)+0.2(2)+0.1(2.5)=1$, so $r$ is a genuine likelihood ratio and both $D$ and $\bar D$ are finite and positive.  Here $\max_{t\le1}E_{S,t}=\max(1,r)$, so for $b>1$ the crossing probability is $\Pr_0(r\ge b)$, which equals $0.3$ at $b=2$ and $0.1(3-b)\le0.1$ for $b\in(2,3]$.  Every $b>2$ is feasible, $b=2$ is not, and although $1/\alpha=20/3$ is feasible the infimum $2$ of the feasible set is not attained.  In this situation the stated fallback deploys $1/\alpha$, which is feasible whether or not the process can attain it.

For a finite-state process, let $v_j<v_{j+1}$ be consecutive possible running maxima.  Every
threshold in $(v_j,v_{j+1}]$ induces the same rejection event, whereas the inclusive threshold
$v_j$ can have a larger crossing probability.  The implementation searches a grid of gate values,
which need not equal the support of the running maximum: a gate can occur only after a larger
previous value, and rectangular state arrays can also include unreachable states.  These extra
candidates select numerical representatives within intervals on which the crossing event is
unchanged.  The recursion starts at the reachable initial state and computes each candidate's
crossing probability on the actual observation law.  Definition~\mainnum{def:bstar} therefore fixes
the candidate set explicitly.  A schedule guard re-evaluates the numerical threshold itself under
the new schedule; test equivalence under the original schedule is not assumed to transfer.
For the next counterexample take the support of the running maximum as the candidate set,
augmented by $1/\alpha$; the reachable-gate candidate set gives the same two selected values.

\paragraph{Horizon monotonicity.} At a fixed numerical threshold $b$, $\Pr_0(\max_{t\le N}E_{S,t}\ge b)$ is nondecreasing in $N$ because the running maximum is. Calibration at $N$ alone supplies no guarantee beyond $N$; Ville still bounds the crossing probability by $1/b$, which may exceed $\alpha$.  The selected value $b^\star(\alpha,N,m)$ is a different object and is not monotone in $N$.  Take Bernoulli streams with $p_0=0.2$ and $p_{\mathrm{alt}}=0.6$ at $\alpha=0.05$, $m=1$, so that a response multiplies the likelihood ratio by $3$ and a nonresponse by $1/2$.  At $N=1$ the attained values are $\{1,3\}$ and $\Pr_0(\max\ge3)=0.2>\alpha$, so no attained value below $1/\alpha$ is feasible and $b^\star=20$.  At $N=2$ the attained values are $\{1,3/2,3,9\}$ and $\Pr_0(\max\ge9)=0.04\le\alpha$, so $b^\star=9$.  The selected boundary therefore falls from $20$ to $9$ as the horizon grows.  The $N=1$ case also shows that $b^\star<1/\alpha$ does not follow from the process failing to reach $1/\alpha$: the running maximum never exceeds $3$ there, yet no attained value is feasible, because the selected finite candidate set need not offer a strictly smaller feasible representative.

\paragraph{Observation schedules.} \mainref{Definition}{def:bstar} fixes the observation schedule, and it must: the null law of the running maximum, and hence feasibility, depends on which streams are read at each round and on how many observations each contributes in all.  Idle updates with factor one preserve the martingale property of $E_{S,t}$, and therefore Ville's threshold at every stopping time, but they do not preserve a smaller calibrated threshold.  Take two independent Bernoulli streams with $p_0=1/2$ and $p_{\mathrm{alt}}=3/4$, horizon $N=6$, $\alpha=0.1$, and the full-support statistic $E_t=\tfrac14(L_{1,t}+L_{2,t})+\tfrac12L_{1,t}L_{2,t}$.  Exhaustive enumeration gives the exact synchronous boundary $b^\star=30051/8192\approx3.668335$, whose synchronous crossing probability is $51/512\approx0.099609$.  Under the schedule that gives stream~1 one observation and then freezes it while stream~2 continues to six, the crossing probability at that same threshold is $13/128\approx0.101563$, above the nominal level.  Freezing spares a null likelihood ratio the further decay that continued sampling would bring, so thinning can raise or lower the level, and the direction has to be computed rather than argued.

For a fixed configuration, closure needs feasibility only for its true-null set.  The frozen
basket calibration checked those sets in the prespecified simulation scenarios.  This is enough for
the error statement at those scenarios, but does not establish strong FWER across other signal
assignments under the same unequal-accrual design. The sensitivity priors include a
\emph{conditional-binomial} family: each stream is independently alternative with probability $q$,
conditional on at least one alternative, giving
$w_l^{(m)}=q^l(1-q)^{m-l}/\{1-(1-q)^m\}$, $0<q<1$.
The one recorded guard increase was the NCI
conditional-binomial $q=0.5$ entry at $m=3$: accrual $(1,0.95,0.90)$ gave level $0.050017$ at
$10.049901$, so the boundary was increased to $10.143601$, giving thinned level $0.049410$.

A separate exact audit checks every nonempty subset of the five-basket designs at all four closure
priors.  All $31$ subsets pass in each design: the worst levels are $0.099801$ for planned accrual
and $0.099787$ for the realized caps $(20,10,8,18,7)$.  Under synchronous NCI accrual, exchangeability
already covers every subset of a given size.  For the additional NCI accrual vector
$(1,.95,.9,.85,.8,.9,.85,1,.95,.9)$, the audit covers all five singleton profiles, all $14$ distinct
size-two profiles, and the size-three true-null profile used in the thinned scenario.
All entries pass except the same level-uniform $m=2$ boundary under three profiles:
$(.85,1)$, $(.9,1)$ and $(.95,1)$ give $0.0500254$, $0.0500315$ and $0.0500255$.
These failures concern assignments absent from the frozen thinned scenario, whose true-null set
has size three; its seven prior-specific thresholds all pass, with maximum level $0.049955$.
The thinned sensitivity results therefore retain their scenario-specific interpretation, not an
all-configurations strong-FWER claim.

The current boundary generator addresses prospective deployment separately.  Its default enumerates
all nonempty subset profiles under each prespecified accrual/cap design, independently of simulated
signal labels.  It increases any failing threshold; for example the level-uniform $m=2$ NCI
threshold increases from $10.257182$ to $10.611539$.  If a required asymmetric state space exceeds
the declared computational budget, it uses $1/\alpha$ for that size and records an analytic Ville
certificate instead of skipping the profile.  A historical reproduction option retains the original
scenario-only guard and labels that narrower coverage.  These prospective changes do not alter the
frozen boundaries or reported operating characteristics.  Table~\ref{tab:scheduleaudit} describes
the coverage of those frozen results.

\begin{table}[!tbp]
\centering
\footnotesize
\setlength{\tabcolsep}{3pt}
\caption{Schedule and subset coverage of the frozen basket boundaries.  Full accrual and both
five-basket designs have all-subset coverage; the NCI thinned sensitivity has the narrower coverage
listed below.  Outside the verified class a design must be re-priced or
monitored at $1/\alpha$, which Ville's inequality makes valid at every stopping time and under every
schedule.}
\label{tab:scheduleaudit}
\begin{tabular}{@{}L{2.3cm}L{2.4cm}L{5.2cm}L{3.4cm}@{}}
\toprule
Design & Caps and accrual & Verified class & Worst verified level\\
\midrule
NCI-MATCH-style, $K=10$, $N=31$, $\alpha=0.05$
& caps all $31$; full accrual, and one thinned accrual vector
& every true-null set under full accrual; the realized thinned $m=3$ profile.  Other thinned
  assignments are not covered uniformly: the level-uniform $m=2$ entry fails on three profiles.
& $\le0.05$ under full accrual; $0.049955$ across priors for the realized thinned $m=3$ profile\\
\addlinespace
Five-basket planned, $K=5$, $N=13$, $\alpha=0.10$
& caps all $13$; full accrual
& all $31$ nonempty true-null subsets, each following the calibration schedule
& $0.099801$\\
\addlinespace
Five-basket realized, $K=5$, $N=20$, $\alpha=0.10$
& caps $(20,10,8,18,7)$; full accrual
& all $31$ nonempty true-null subsets, at all four closure priors, checked exactly
& $0.099787$\\
\bottomrule
\end{tabular}
\end{table}

Outside the verified class a calibrated boundary carries no guarantee and must not be assumed to
transfer.  A design with other caps, other accrual probabilities, or a true-null family the audit did
not price must either be re-priced in the same way before deployment or monitored at $1/\alpha$.  One
case needs no computation: at $m=1$ the statistic is the stream's own likelihood ratio, and any
schedule with cap $c\le N$ lets it visit only a prefix of the path it would visit under full accrual,
so its running maximum is pathwise no larger and the full-accrual boundary stays feasible.  For
$m\ge2$ no such comparison holds, as the two-stream example above shows.  A schedule-free bound does
exist, since the gate is coordinatewise nondecreasing and the maximum over any schedule respecting the
caps is therefore at most the gate evaluated at the independent per-stream running maxima over their
own caps; but on these designs it is far too conservative to certify anything, reaching $0.19$ against
$\alpha=0.10$ at the geometric prior with $m=5$, so it is recorded here as a structural observation
rather than as a usable certificate.

\paragraph{Composite marginal nulls.} For a composite marginal null such as $H_k:p_k\le p_0$, the boundary is computed at $p_k=p_0$ for every $k\in S$.  Justifying this requires a monotone coupling of the \emph{data}, not merely of their marginal law: construct the Bernoulli draws from a common uniform variable $U_{k,t}$, setting the response to $1$ iff $U_{k,t}\le p_k$, so raising $p_k$ raises the realized data pathwise.  When the per-observation likelihood-ratio factor is increasing in the response, as it is for the basket model of the main article, the coupling makes $E_{S,t}$, and hence $\max_{t\le N}E_{S,t}$, pathwise nondecreasing in every $p_k$, so crossing $b^\star$ by time $N$ is an increasing function of the coupled data, not merely an event whose marginal probability happens to be ordered at each $t$.  Closure needs this jointly, for the intersection null $\{p_k\le p_0:k\in S\}$ rather than each marginal $H_k$ in isolation: the coupling raises every coordinate $p_k$ at once, so $(p_k)_{k\in S}=(p_0,\dots,p_0)$ is the least-favorable configuration for $H_S$, and the boundary is evaluated there.

\paragraph{Computation.} Both calibration routes search a finite sorted candidate list by
bisection (Algorithm~\ref{alg:bstar}). The population crossing probability
$\pi(b)=\Pr_0(\max_{t\le N}E_{S,t}\ge b)$ is nonincreasing in $b$ and, for a finite-state law,
constant between consecutive possible running maxima. On the Monte Carlo route, it is the
empirical hit count and its upper confidence bound that are nonincreasing step functions of the
boundary; the continuous population crossing law need not be stepwise.

The exact route applies to Bernoulli marginals with null response probability $p_0$ and working
alternative $p_{\rm alt}$.  A response multiplies $L_{k,t}$ by $r_1=p_{\rm alt}/p_0$ and a
nonresponse by $r_0=(1-p_{\rm alt})/(1-p_0)$. If $n_k$ is the number of responses among
$t$ observations in stream $k$, then $L_{k,t}=r_1^{n_k}r_0^{\,t-n_k}=c^{\,n_k}r_0^{\,t}$ with $c=r_1/r_0$, and $E_{S,t}=G_t(n_1,\dots,n_m)$
is a function of the count vector on the lattice $\{0,\dots,t\}^m$.  Step~1 tabulates $\log G_t$ by
the ESP prefix recursion in log scale.  Step~2 evaluates $\pi(b)$ by backward recursion with
absorption. Write $n=(n_1,\dots,n_m)$ and let $F_t(n)$ be the probability of a crossing from
state $n$ at time $t$ through $N$, counting a crossing at $t$ itself. Then
$F_N(n)=\mathbf 1\{G_N(n)\ge b\}$, and for $t<N$, $F_t(n)=1$ if $G_t(n)\ge b$ and
otherwise
\[
F_t(n)=\sum_{u\in\{0,1\}^m}p_0^{|u|}(1-p_0)^{m-|u|}\,F_{t+1}(n+u),
\]
where $|u|=\sum_k u_k$ counts the next round's responses. This is the expectation over $m$
independent Bernoulli$(p_0)$ draws; then
$\pi(b)=F_0(0,\dots,0)$.  Step~3 collects the distinct tabulated values $G_t(n)\le1/\alpha$ as
candidates, checks that the largest is feasible, returning $1/\alpha$ if it is not, and bisects to
the smallest feasible candidate.

Independence lets the implementation apply the one-coordinate Bernoulli transitions successively,
without explicitly summing over all $2^m$ response vectors at each state. Swept literally this costs
$(N+1)^m$ states per time layer and $O(mN(N+1)^m)$ per evaluation of
$\pi$, which is affordable up to about four million lattice sites and is what the implementation uses
there, so that published entries do not move in their last bits over a change of engine.  Beyond
that threshold three exact reductions are used instead.

\emph{(i) Permutation-orbit collapse.}  Under the all-null law the $m$ cohorts are i.i.d.\ and
$G_t$ is symmetric, so the transition kernel and the gate both factor through the \emph{multiset}
of counts.  Lumping the lattice onto multisets is therefore an exact aggregation of the Markov
chain, not an approximation.  It requires exchangeability and is switched off whenever caps or
accrual differ across cohorts.

\emph{(ii) Absorbing count cap.}  For any prior with $w_1>0$,
\[
E_{S,t}\;\ge\;w_1e_1\bigl(\{L_{k,t}\}\bigr)\;\ge\;w_1\max_kL_{k,t}
   \;=\;w_1c^{\,\max_kn_k}r_0^{\,t}\;\ge\;w_1c^{\,\max_kn_k}r_0^{\,N},
\]
using $r_0<1$ and $t\le N$.  Let $x^\star$ be the least integer with
$w_1c^{\,x^\star}r_0^{\,N}>1/\alpha$.  A cohort whose count reaches $x^\star$ therefore puts the gate
strictly above every candidate, at that look and at every later one, since the displayed lower bound
does not decrease as counts accumulate.  The recursion keeps live counts only through
$x^\star-1$ and absorbs transitions to $x^\star$.  Nothing live is discarded: the bound is uniform in $t$, so no path is
removed that could still fail to cross, and the candidate pool is unaffected because every discarded
gate value exceeds $1/\alpha$ and was never a candidate.  The implementation re-checks this on the
last look, where the gate is smallest, rather than trusting the algebra.  The cap is vacuous when
$w_1=0$, and the code refuses to apply it there.

\emph{(iii) Product chain.}  At the product corner $w_1=0$ and the cap is unavailable, but the gate
is $w_m\prod_kL_{k,t}=w_mc^{\,\sum_kn_k}r_0^{\,mt}$, a function of the total response count alone,
whose null law after $t$ rounds is $\mathrm{Binomial}(mt,p_0)$.  The state is then the scalar total.

For the NCI design, $m=10$ and $N=31$, the literal final layer has $32^{10}\approx1.13\times10^{15}$
states.  Reductions (i) and (ii) leave $8{,}008$ live multiset states at the arithmetic prior,
$19{,}448$ at the geometric and the two sparser conditional-binomial priors, $43{,}758$ at
level-uniform and $92{,}378$ at conditional-binomial $q=0.5$, the largest of the seven; reduction
(iii) leaves $311$.  The per-entry calibration times recorded in the frozen configuration files sum
to $228$, $26$ and $541$ seconds for the NCI, planned and realized designs, respectively.  These
historical timings cover synchronous calibration and exclude the later schedule-guard evaluations;
they are provenance records rather than a fresh hardware benchmark.  Table~\ref{tab:boundarytable} lists every deployed
entry.

A separate retained benchmark runs only the NCI $m=10$ geometric calibration in a fresh Python
process. It took $3.12$ seconds for the calibration call and reached $115$ MiB peak resident memory,
including the interpreter and imported libraries, on macOS x86-64 with Python~3.12.7,
NumPy~1.26.4 and SciPy~1.14.1.  The benchmark record includes the command, source hashes and
environment.  Concurrent verification work may affect this single timing; it is not extrapolated
to other priors or schedule guards.

The regression suite compares $57$ boundary calculations and $583$ crossing probabilities from
the collapsed and full engines on tractable sizes, requiring discrepancies below $10^{-13}$ in
both log threshold and level.  The engines can name different floating-point copies of the same
gate value.  Every calibration comparison $G_t(n)\ge b$ uses a tolerance of $10^{-9}$ in log scale, which absorbs floating-point noise
between permutation-equivalent lattice states without merging genuinely distinct values, since
distinct gate values near the boundary are about $10^{-4}$ apart; erring toward counting a crossing
can only raise the selected boundary, never lower it.  There is deliberately no Monte Carlo route
here: Bernoulli increments put the accumulated log likelihood ratio on a lattice, simulated path
maxima tie in bulk, and a hit count taken at exact floating-point equality undercounts, after which a
Clopper--Pearson bound accepts thresholds that overspend $\alpha$.

The Monte Carlo route applies when no finite lattice exists, as for the Gaussian model.  Step~1
simulates $M$ independent null paths of the $m$ streams over $N$ steps, computes $\log E_{S,t}$ at
every step by the log-ESP recursion, and records the path maximum.  Step~2, for a candidate $b$,
counts the paths whose maximum is at least $\log b$, with the same tolerance, and declares $b$
feasible if the one-sided Clopper--Pearson upper bound at confidence $\gamma$ on the crossing
probability, the $\gamma$-quantile of the $\mathrm{Beta}(\mathrm{hits}+1,\,M-\mathrm{hits})$
distribution, is at most $\alpha$.  A point estimate alone does not certify the true
crossing probability.  Step~3 takes the distinct path maxima at most $\log(1/\alpha)$
as candidates, checks the largest, and bisects to the smallest feasible one; if none is feasible,
$\widehat b=1/\alpha$, valid by Ville's inequality regardless of the simulation.  The Gaussian
designs use $M=50{,}000$ and $\gamma=0.999$.

\paragraph{Why searching the calibration paths needs no multiplicity correction.}
Let $Z_1,\dots,Z_M$ be independent path maxima with distribution function $F$, and write
$\pi(b)=1-F(b-)$ and $H(b)=\sum_i\mathbf1\{Z_i\ge b\}$.  Define the deterministic allowable count
\[
h^\star=\max\{h\in\{0,\dots,M-1\}:
  \mathrm{Beta}^{-1}_{h+1,M-h}(\gamma)\le\alpha\},
\]
with $h^\star=-1$ if this set is empty.  If $h^\star<1$, no observed maximum qualifies and the
routine uses the Ville fallback.  Otherwise every nonfallback return satisfies $H(\widehat b)\le
h^\star$; counting additional near-ties can only strengthen this inequality for the exact hit count.

For a proof valid with atoms, represent the calibration sample as $Z_i=F^{-1}(U_i)$ with
independent uniform $U_i$.  If $\pi(b)>\alpha$, then $F(b-)<1-\alpha$, so every $U_i>1-\alpha$
implies $Z_i\ge b$.  This implication holds simultaneously for all infeasible $b$.  Consequently,
with $B=\sum_i\mathbf1\{U_i>1-\alpha\}\sim\mathrm{Binomial}(M,\alpha)$,
\[
\{\pi(\widehat b)>\alpha,\ \widehat b\text{ is not the fallback}\}
 \subseteq\{B\le h^\star\},\qquad
\Pr(B\le h^\star)\le1-\gamma.
\]
The final inequality is the defining binomial-tail identity of the one-sided
Clopper--Pearson bound.  Ville's fallback is feasible on every calibration draw, so
$\Pr\{\pi(\widehat b)\le\alpha\}\ge\gamma$.  The proof does not condition on the selected
threshold or assume that counts remain binomial after selection.
In the continuous, tolerance-free case, before the restriction to candidates at most $1/\alpha$,
the selected order statistic is $Z_{(M-h^\star+1)}$, whose inclusive upper-tail mass has law
$\mathrm{Beta}(h^\star,M-h^\star+1)$.  Ties, tolerance and fallback are covered by the preceding
argument, without using this sharper special-case identity.

For data independent of calibration, the feasibility event yields conditional level $\alpha$;
integrating over calibration gives $\alpha+(1-\alpha)(1-\gamma)$.  The uncapped order-statistic
expected-tail identity alone does not establish exact unconditional level $\alpha$ for the
implemented rule, because the Ville fallback changes its distribution.  Calibration confidence
across intersection sizes is combined as stated in \S\mainnum{sec:calibration}.

\begin{table}[!tbp]
\centering
\footnotesize
\setlength{\tabcolsep}{3pt}
\caption{Every deployed basket boundary, by design and closure prior.  Each row ranges over the
intersection sizes $m=1,\dots,K$; the deployed $b_m$ is reported at the two ends, the null crossing
probability and the reduction below $1/\alpha$ as ranges over $m$.  ``Engine'' records which exact
computation produced the entries: \emph{lattice} is the literal $(N+1)^m$ count sweep,
\emph{collapsed} the permutation-orbit and count-cap reduction, or the total-count chain at the
product corner.  No entry is estimated, so none carries a Monte Carlo confidence caveat.  ``Lifted''
counts entries raised above the unguarded synchronous optimum so as to remain feasible under a
realized true-null accrual profile; the one lift is the conditional-binomial $q=0.50$ prior at $m=3$
described above.  Per-size entries, full digests, and per-profile levels are in the machine-readable
configuration files.}
\label{tab:boundarytable}
\begin{tabular}{@{}llrrrrlr@{}}
\toprule
Design & Closure prior & $b_m$ at $m{=}1$ & at $m{=}K$ & Null level & Reduction (\%) & Engine & Lifted\\
\midrule
NCI & Geometric (primary) & 11.236 & 9.709 & 0.0485--0.0500 & 43.8--51.5 & collapsed, lattice & ---\\
 & Level-uniform & 11.236 & 9.137 & 0.0458--0.0500 & 43.8--60.7 & collapsed, lattice & ---\\
 & Binomial $q=0.10$ & 11.236 & 9.829 & 0.0494--0.0500 & 43.8--50.9 & collapsed, lattice & ---\\
 & Binomial $q=7/27$ & 11.236 & 10.126 & 0.0492--0.0500 & 43.8--51.0 & collapsed, lattice & ---\\
 & Binomial $q=0.50$ & 11.236 & 9.779 & 0.0471--0.0500 & 43.8--56.5 & collapsed, lattice & 1\\
 & Arithmetic & 11.236 & 9.170 & 0.0477--0.0500 & 43.8--54.2 & collapsed, lattice & ---\\
 & Product & 11.236 & 4.027 & 0.0176--0.0496 & 43.3--79.9 & collapsed, lattice & ---\\
\addlinespace
Planned & Geometric (primary) & 5.945 & 4.882 & 0.0893--0.0997 & 40.6--51.6 & lattice & ---\\
 & Level-uniform & 5.945 & 4.907 & 0.0893--0.0998 & 40.6--57.1 & lattice & ---\\
 & Arithmetic & 5.945 & 4.619 & 0.0893--0.0994 & 40.6--53.8 & lattice & ---\\
 & Product & 5.945 & 3.466 & 0.0651--0.0987 & 38.0--65.3 & lattice & ---\\
\addlinespace
Realized & Geometric (primary) & 5.945 & 5.626 & 0.0951--0.0998 & 40.3--43.7 & lattice & ---\\
 & Level-uniform & 5.945 & 5.516 & 0.0939--0.0999 & 40.0--49.3 & lattice & ---\\
 & Arithmetic & 5.945 & 5.655 & 0.0993--0.1000 & 40.3--43.7 & lattice & ---\\
 & Product & 5.945 & 3.466 & 0.0656--0.0998 & 38.0--65.3 & lattice & ---\\
\bottomrule
\end{tabular}
\end{table}

\begin{algorithm}[!tbp]
\caption{Computing a deployed horizon-$N$ boundary $b_m$: exact enumeration or Monte Carlo}
\label{alg:bstar}
\begin{algorithmic}[1]
\State \textbf{input:} $\alpha$, $N$, $m$, weights $w^{(m)}$; and either $(p_0,p_{\rm alt})$ for
the exact route, or a null increment sampler with path count $M$ and confidence $\gamma$ for the
Monte Carlo route
\State \textbf{exact route, state space:} \textbf{if} $(N+1)^m$ is affordable, use the full count
lattice; \textbf{else if} $w^{(m)}_1>0$, collapse to multisets with live counts below
$x^\star=\min\{x: w_1^{(m)}c^{\,x}r_0^{\,N}>1/\alpha\}$, absorbing transitions to $x^\star$;
\textbf{else if} the prior is the product corner, use the scalar total response count;
\textbf{else} use an affordable exact state space or return $1/\alpha$
\State \textbf{exact route:} tabulate $\log G_t$ on that state space for $t=0,\dots,N$; let
$\pi(b)$ be the crossing probability from the backward recursion with absorption;
candidates $\gets$ distinct tabulated gate values $\le1/\alpha$
\State \textbf{Monte Carlo route:} simulate $M$ null paths and record $\max_{t\le N}\log E_{S,t}$
on each, including $E_{S,0}=1$; let $\pi(b)$ be the Clopper--Pearson upper bound at confidence $\gamma$ on the fraction
of maxima at least $\log b$; candidates $\gets$ distinct maxima at most $\log(1/\alpha)$, converted to evidence scale
\State sort the candidates; \textbf{if} none exist or $\pi(\text{largest candidate})>\alpha$ \textbf{then
return} $1/\alpha$
\State bisect over the sorted candidates for the smallest $b$ with $\pi(b)\le\alpha$; count calibration
crossings with a tolerance of $10^{-9}$ in log scale (deployment uses $E_{S,t}\ge b$, without lowering $b$)
\State \textbf{schedule guard (exact route):} \textbf{for} each true-null cap and accrual profile
the design admits, recompute $\pi$ on the asymmetric state space;
\textbf{while} any profile has $\pi(b)>\alpha$, advance $b$ to the next grid candidate, capped at
$1/\alpha$; if a required calculation is unaffordable, use $1/\alpha$ with its analytic certificate
\State \textbf{return} $b_m=b$ with the per-profile levels or Ville certificate. Before guarding,
the exact return is $b^\star$; on the Monte Carlo route record $\widehat b$, the hit count, $M$ and $\gamma$
\end{algorithmic}
\end{algorithm}

Write $\delta=1-\gamma$. For a procedure that must be level $\alpha$ unconditionally, calibrating to the internal target
$(\alpha-\delta)/(1-\delta)$ when $\delta<\alpha$ makes the integrated bound equal to $\alpha$.
The cost is about $\delta$ in spent error budget. For simultaneous feasibility across sizes,
confidence $1-\delta_m$ at size $m$ gives joint confidence at least $1-\sum_m\delta_m$ by a union
bound; a fixed configuration's FWER proof uses only its true-null set's boundary.

\subsection{Closed-testing error guarantees}\label{app:closureproofs}

Both guarantees reduce family-wise error to a crossing by the true-null intersection. Let $S_0$
be the set of true nulls and $V_t$ the number of false rejections by time $t$. For a nonempty
$S_0$, write $b=b_{|S_0|}$ for its deployed threshold in the first proof.

\begin{proof}[Proof of Proposition~\mainnum{prop:bstarvalid}]
If $S_0=\varnothing$, then $V_t=0$ identically and \mainref{Equation}{eq:fwerN} holds trivially.  Otherwise,
rejecting any $k\in S_0$ requires rejecting $H_{S_0}$.  Under Assumption~\mainnum{ass:indep} the law of
the streams in $S_0$, under the specified outcome-independent accrual design, is the all-null law
$P_0^{S_0}$ used by the feasibility hypothesis.  Hence
$\Pr(\exists t\le N: E_{S_0,t}\ge b)\le\alpha$.  The inclusion
$\{\exists t\le N:V_t>0\}\subseteq\{\exists t\le N: E_{S_0,t}\ge b\}$ finishes the proof.
\end{proof}

\begin{proof}[Proof of Proposition~\mainnum{prop:fwer}]
Let $S_0$ be the true-null set. If $S_0=\varnothing$, then $V_t=0$ identically. Otherwise, rejecting
any $k\in S_0$ requires rejecting the nonempty true intersection $H_{S_0}$.  Put
$\sigma_{S_0}=\inf\{t:E_{S_0,t}\ge1/\alpha\}$.  For every deterministic $T$, the e-process property at
the bounded stopping time $\sigma_{S_0}\wedge T$ and nonnegativity give
\[
1\ge \E[E_{S_0,\sigma_{S_0}\wedge T}]
  \ge \alpha^{-1}\Pr(\sigma_{S_0}\le T).
\]
Letting $T\to\infty$ yields $\Pr(\sigma_{S_0}<\infty)\le\alpha$.  Since
$\{\sup_tV_t>0\}\subseteq\{\sigma_{S_0}<\infty\}$, the result follows.  The full-filtration
test-martingale argument of \S\mainnum{sec:object} verifies the stated baseline specialization.
\end{proof}

\smallskip\noindent(This appendix is cited in \S\S\mainnum{sec:calibration} and \mainnum{sec:validity} of the main article.)

\section{Proofs of Lemma~\mainnum{lem:sufficient} and Proposition~\mainnum{prop:espcoords}: the sufficiency results}\label{app:sufficiency}

\paragraph{Proof of Lemma~\mainnum{lem:sufficient} (minimal sufficiency).} The family is dominated by $P_0^{(t)}$, and conditional independence gives $dQ_c^{(t)}/dP_0^{(t)}=L_{c,t}=\prod_{k\in c}L_{k,t}$, a measurable function of $T_t=(L_{k,t})_{k\in S}$ for every $c$; the Halmos--Savage factorization criterion for dominated families gives sufficiency.  For minimality, if $U$ is any sufficient statistic for the family then, again by Halmos--Savage, each density $dQ_c^{(t)}/dP_0^{(t)}$ has a version measurable with respect to $\sigma(U)$; the singleton configurations $c=\{k\}$ give $dQ_{\{k\}}^{(t)}/dP_0^{(t)}=L_{k,t}$, so every coordinate of $T_t$ is a $\sigma(U)$-measurable function, i.e.\ $T_t$ is a function of $U$.  Since this holds for every sufficient $U$, $T_t$ is minimal sufficient. \qed

\paragraph{Proof of Proposition~\mainnum{prop:espcoords}(a) (symmetrization).}
For a feasible integrable scalar rule $\delta$, put $\delta_T=\E_{P_0}[\delta\mid T_t]$.
Sufficiency supplies the same conditional expectation under each configuration law.  Conditional
Jensen's inequality therefore gives no greater risk at each configuration, and the assumed
Rao--Blackwell closure makes $\delta_T$ feasible.  Next average $\delta_T$ over the finite permutation
group.  Convexity and relabeling closure preserve feasibility.  Convexity of the loss, its invariance,
and the exchangeable prior show that this average has no larger Bayes risk; the group step need
not improve risk separately at every configuration.  The result is a symmetric function of $T_t$.
Apply these two operations to an optimizer if one exists.

For the scalar e-value class, nonnegativity and the budget $\E_{P_0}E_t\le1$ are preserved by both
operations.  Relabeling preserves the null product law, and conditioning preserves the null mean,
so all the required closure properties hold.  For vector actions whose coordinates carry stream
labels, the appropriate group averaging also acts inversely on the output labels and gives an
equivariant rule; scalar invariance is not the conclusion for that different decision problem.
This is a fixed-horizon assertion and does not eliminate the history needed to implement a
sequential crossing rule. \qed

\paragraph{Proof of Proposition~\mainnum{prop:espcoords}(b) (the invariant coordinates).}

The ESP vector determines the unordered likelihood-ratio vector, and sorting recovers a
measurable representative. Let $C=\{x\in\mathbb R^m:x_1\le\cdots\le x_m\}$, a Borel set.  The restriction $e|_C$ of $e=(e_1,\dots,e_m)$ is injective (by Vieta $\prod_k(z-x_k)=\sum_{j=0}^m(-1)^je_j(x)\,z^{m-j}$, so a monic polynomial determines its root multiset, its coefficients being the $e_j(x)$ up to sign, and each multiset has exactly one nondecreasing representative) and Borel.  By the Lusin--Souslin theorem~\citep[Theorem~15.1]{kechris1995}, an injective Borel map between standard Borel spaces has Borel image and Borel inverse, so $e|_C$ is a Borel isomorphism onto $e(C)=e(\mathbb R^m)$.

The containment $\sigma(e_1(T_t),\dots,e_m(T_t))\subseteq\sigma(T_t)^{\mathfrak S_m}$ is immediate, since each $e_l$ is symmetric and Borel.
For the reverse containment, let $E\in\sigma(T_t)^{\mathfrak S_m}$; since $E\in\sigma(T_t)$, $E=T_t^{-1}(A)$ for some Borel $A\subseteq\mathbb R^m$, and invariance of $E$ gives $T_t^{-1}\big(\bigcap_{\rho\in\mathfrak S_m}\rho(A)\big)=\bigcap_{\rho\in\mathfrak S_m}\rho^{-1}(E)=E$, so $A^\star=\bigcap_{\rho\in\mathfrak S_m}\rho(A)$ is a symmetric Borel set with $E=T_t^{-1}(A^\star)$: every invariant event is $T_t^{-1}(A)$ for some symmetric Borel $A$, so we may take $A$ symmetric.  With $A$ symmetric and Borel, $A=e^{-1}\big(e(A\cap C)\big)$, and $e(A\cap C)$ is Borel by the above, so $E=T_t^{-1}(A)\in\sigma(e_1(T_t),\dots,e_m(T_t))$.  The two $\sigma$-algebras therefore coincide. \qed

\section{Per-step GRO derivation}\label{app:gro}

The per-step growth-rate-optimal e-value is the likelihood ratio, and under the product model its worst-case version is the arithmetic mean $\tilde e_1$ of the per-step ratios.  For the per-step e-value $M\ge0$ with $\E_{P_0}[M]\le1$, growth-rate optimality maximizes $\E_Q[\log M]$ for a fixed one-step alternative law $Q$.
Put $L=dQ/dP_0$.  Under the standing equivalence of $Q$ and $P_0$, Jensen's
inequality yields, with extended-real logarithms,
\[
\E_Q[\log M]-\KL(Q\|P_0)
=\E_Q[\log(M/L)]
\le \log\E_Q[M/L]
=\log\E_{P_0}[M]\le0.
\]
Equality requires the budget to be tight and $M/L$ to be constant $Q$-almost surely, hence
$M^\star=L$.  This argument covers arbitrary measurable candidates without a functional
stationarity assumption. For
the composite alternative $\mathcal C$, the worst-case (GROW) optimum is the likelihood ratio against
the KL-projection $\bar Q_{\lambda^\star}$~\citep{grunwald2024}. Since
$dQ_c/dP_0=\prod_{k\in c}r_k$,
\[
M^\star=\sum_c\lambda_c^\star\prod_{k\in c}r_k.
\]
When $\lambda_c^\star=a_{|c|}$ is exchangeable, this is the PEM statistic
$\sum_l a_l e_l(r_1,\dots,r_m)$ from the main article; equivalently, with total level mass
$\beta_l^\star=\binom{m}{l}a_l$, it is $\sum_l\beta_l^\star\tilde e_l(r_1,\dots,r_m)$.
Under the present product model,
\mainref{Proposition}{prop:growmean} identifies
$\lambda^\star$ as uniform
on the singletons, so $M^\star=m^{-1}\sum_{k\in S}r_k=\tilde e_1$.

\smallskip\noindent(This appendix is cited from Appendix~\ref{app:bayes} and supports \S\mainnum{sec:opt} of the main article.)

\section{Proofs of prior-matched optimality}\label{app:bayes}

The proof of \mainref{Theorem}{thm:bayes} has two steps: Gibbs' inequality identifies the unique terminal maximizer, and a linear-independence argument identifies its weights.  Fix $t\ge1$ and a local e-process $E$ adapted to $(\F_s^S)_{s\ge0}$. Define the finite measure $G_t$ on
$\F_t^S$ by $dG_t=E_t\,dP_0^{(t)}$. Because deterministic $t$ is a stopping time,
$G_t(\Omega)=\E_{P_0^{(t)}}E_t\le1$. The two finite KL assumptions imply that $f_0$ and $f_A$ are
mutually absolutely continuous, hence $\bar Q_\omega^{(t)}$ and $P_0^{(t)}$ are equivalent. With
extended-real logarithms (if the left side is $-\infty$, the inequality is immediate), Jensen's
inequality gives
\[
\begin{split}
\E_{\bar Q_\omega^{(t)}}[\log E_t]
-\KL\big(\bar Q_\omega^{(t)}\,\big\|\,P_0^{(t)}\big)
&=\E_{\bar Q_\omega^{(t)}}\!\left[
  \log\frac{dG_t}{d\bar Q_\omega^{(t)}}\right]\\
&\le \log G_t(\Omega)\le0.
\end{split}
\]
The reference value is finite because convexity of KL yields
$\KL(\bar Q_\omega^{(t)}\|P_0^{(t)})\le\sum_c\omega_c l_cDt$.
Equality requires both $G_t(\Omega)=1$ and
$dG_t/d\bar Q_\omega^{(t)}=1$ almost surely, hence
$G_t=\bar Q_\omega^{(t)}$. Thus the unique maximizing terminal e-value is
$d\bar Q_\omega^{(t)}/dP_0^{(t)}$. The likelihood-ratio identity~(\mainnum{eq:lr_identity}) of the main
article identifies it with the
accumulated mixture at $w=\omega$. The fixed configuration mixture on the infinite product
experiment has these consistent finite-horizon likelihood ratios, so the same test martingale is
optimal at every horizon simultaneously.

It remains to identify its weights. Put $\mu_0=f_0^{\otimes t}$ and
$\mu_1=f_A^{\otimes t}$. Since $D>0$, $\mu_0\ne\mu_1$, and the two probability measures are linearly
independent. Choose a measurable set $A$ on which they differ. Linear combinations of
$1_A$ and $1$ give bounded functions $u_0,u_1$ satisfying
$\int u_i\,d\mu_j=\mathbf 1\{i=j\}$. If a linear combination of the $m$-fold tensor products
$\bigotimes_{k=1}^m\mu_{b_k}$, $b\in\{0,1\}^m$, were zero, integrating it against
$\prod_{k=1}^m u_{b_k}$ would isolate each coefficient. The tensor products are therefore linearly
independent, as is the subfamily $\{Q_c^{(t)}:c\ne\varnothing\}$. Consequently
$\bar Q_w^{(t)}=\bar Q_\omega^{(t)}$ implies $w=\omega$. \qed

\paragraph{First-order growth of the Bayes objective.} Write $\psi_{t,\omega}(w):=\Psi_{t,\omega}(E^w)$, concave in $w$.  For every full-support $w$,
\[
\sum_c\omega_cl_cDt-\log\frac1{\min_c w_c}
\le \psi_{t,\omega}(w)
\le \KL(\bar Q_\omega^{(t)}\|P_0^{(t)})
\le \sum_c\omega_cl_cDt,
\]
so $t^{-1}\psi_{t,\omega}(w)\to\sum_c\omega_cl_cD$.  The pointwise bound $E^w_{S,t}\ge w_cL_{c,t}$ gives the lower bound, while Gibbs' inequality and convexity of KL give the upper bounds.  At $w=\omega$, even if the prior has zeros, summing the same lower bound over its support gives the finite remainder $\sum_{c:\omega_c>0}\omega_c\log\omega_c$, and hence the same limit.  The pathwise version is \mainref{Theorem}{thm:regret}.

\paragraph{Accumulated versus repeated one-step mixing.} Repeating the fixed one-step mixture $M_w(r_s)=\sum_{l=1}^{m}w_l e_l(\{r_{k,s}\}_{k\in S})$ gives a valid process $\prod_{s\le t}M_w(r_s)$ that is generally not the accumulated mixture of the main article.  At $S=\{1,2\}$ with configuration weights $(w_{\{1\}},w_{\{2\}},w_{\{1,2\}})=(a,a,b)$, $2a+b=1$, the accumulated process is $aL_{1,t}+aL_{2,t}+bL_{1,t}L_{2,t}$ while the repeated one is $\prod_{s\le t}\{ar_{1,s}+ar_{2,s}+br_{1,s}r_{2,s}\}$, whose expansion lets the active configuration change across times; the two therefore differ in general for $t\ge2$.  The repeated process is one of the local e-processes compared in \mainref{Theorem}{thm:bayes}, so for every fixed design prior $\omega$ the accumulated process $E^\omega_{S,t}$ has at least its expected log evidence under the design-mixture law at every finite horizon.  The one-step growth rate $g_c(w):=\E_{Q_c^{(1)}}[\log M_w(r)]$ governs the repeated per-step mixture, not the accumulated object; its Kelly derivation is Appendix~\ref{app:gro}.

\paragraph{Exchangeable priors and familiar merging rules.}
For \mainref{Corollary}{cor:esp}, grouping the terms in the configuration mixture by cardinality
gives $\sum_c\omega_cL_{c,t}=\sum_l\omega_l e_l(\{L_{k,t}\}_{k\in S})$ when
$\omega_c=\omega_{|c|}$. There are $\binom ml$ configurations of size $l$, so their total prior
mass is $\beta_l=\binom ml\omega_l$; this determines the level weights uniquely.

\begin{proof}[Proof of Corollary~\mainnum{cor:corners}]
Each statistic is $E^\omega_{S,t}$ for the stated $\omega$, and Theorem~\mainnum{thm:bayes} identifies the
maximizer uniquely with its weights.
\end{proof}

\section{Proof of Theorem~\mainnum{thm:grow}: the finite-horizon minimax representation}\label{app:grow}

The proof applies Sion's minimax theorem to a concave--linear objective and evaluates the inner maximum with \mainref{Theorem}{thm:bayes}.  Fix the horizon $t$ and write $G_c(w)=\E_{Q_c^{(t)}}[\log E^w_{S,t}]$, concave in $w$ (the log composed with
the linear map $w\mapsto E^w_{S,t}$) and continuous on the simplex. Indeed, all configuration
likelihood ratios are positive almost surely and
$\min_dL_{d,t}\le E^w_{S,t}\le\max_dL_{d,t}$. The finite configuration set and the two finite KL
assumptions make $\max_d|\log L_{d,t}|$ integrable under every $Q_c^{(t)}$, providing a common
dominating envelope. Meanwhile,
$\sum_c\lambda_cG_c(w)$ is linear in the index distribution
$\lambda\in\Delta(\mathcal C)$; both $w$ and $\lambda$ range over the full simplex $\Delta(\mathcal C)$.
The minimax interchange itself does not use exchangeability. Under
\mainref{Assumption}{ass:exch}, the saddle
identified below is exchangeable by symmetry. By Sion's minimax
theorem~\citep{sion1958},
\[
\max_w\min_{c}G_c(w)=\max_w\min_\lambda\sum_c\lambda_cG_c(w)=\min_\lambda\max_w\sum_c\lambda_cG_c(w).
\]
By \mainref{Theorem}{thm:bayes} (with
$\omega=\lambda$), the inner maximum is attained at $w=\lambda$, and
its value is $\KL(\bar Q_\lambda^{(t)}\|P_0^{(t)})$.

Minimizing over $\lambda$ gives the displayed $Q\|P_0$ divergence minimization over
$\mathrm{conv}\{Q_c^{(t)}\}$,
and the saddle has $w^\star_t=\lambda^\star_t$. (Sion gives equality of the optimal values;
attainment of the outer maximum at $w=\lambda^\star_t$ follows from the first-order condition of the
minimization,
\[
G_c(\lambda^\star_t)\ \ge\ \KL(\bar Q_{\lambda^\star_t}^{(t)}\|P_0^{(t)})\ \text{ for every }c,
\qquad\text{with equality on }\mathrm{supp}(\lambda^\star_t),
\]
so $\min_cG_c(\lambda^\star_t)$ equals the minimax value.) Finally, the value is optimal among
\emph{all} local $(\F_s^S)_{s\ge0}$-adapted e-processes, not only mixtures: for any such e-process $E$,
\[
\min_c\E_{Q_c^{(t)}}[\log E_t]
\le\E_{\bar Q_{\lambda^\star_t}^{(t)}}[\log E_t]
\le\KL(\bar Q_{\lambda^\star_t}^{(t)}\|P_0^{(t)}).
\]
The last inequality follows from
\mainref{Theorem}{thm:bayes}, and the mixture with
$w=\lambda^\star_t$ attains the
bound. \qed

\paragraph{Proof of Proposition~\mainnum{prop:growmean} (the horizon-uniform GROW solution).} Although
$\KL(\bar Q_\lambda^{(t)}\|P_0^{(t)})$ is not additive in $t$ ($\bar Q_\lambda^{(t)}$ is a mixture of
products, not a product), its minimizer is horizon-free. Let $\lambda^{u}$ be uniform on the $m$
singleton configurations, so that
$d\bar Q_{\lambda^u}^{(t)}/dP_0^{(t)}=\bar e_1(t)=m^{-1}\sum_{k\in S}L_{k,t}$. By exchangeability
$G_{\{j\}}(\lambda^u)$ is the same for every $j\in S$, and
\[
\KL(\bar Q_{\lambda^u}^{(t)}\|P_0^{(t)})
=\sum_{j\in S}\lambda^u_{\{j\}}G_{\{j\}}(\lambda^u)
=G_{\{j\}}(\lambda^u)
\qquad\text{for every }j\in S.
\]
Thus the first-order condition holds with equality on the support. For any configuration $c$ and any
$j\in c$: under $f_A$ each per-step ratio $r$ is the size-biased version of its $f_0$-law
($\Pr_{f_A}(r\in dx)=x\,\Pr_{f_0}(r\in dx)$), hence stochastically larger, and so is each
$\log L_{k,t}$, a sum of independent stochastically ordered increments. We can therefore couple the
coordinates of $c\setminus\{j\}$, independently of the remaining coordinates, so that their
accumulated likelihood ratios increase pathwise. Because $\log\bar e_1(t)$ is coordinatewise
nondecreasing, this coupling gives
\[
G_c(\lambda^u)\;\ge\;G_{\{j\}}(\lambda^u)\;=\;\KL(\bar Q_{\lambda^u}^{(t)}\|P_0^{(t)}).
\]

The first-order condition of the convex problem $\min_\lambda\KL(\bar Q_\lambda^{(t)}\|P_0^{(t)})$ is
therefore verified at $\lambda^u$ at every horizon, so $\lambda^\star_t=\lambda^u$ for every $t$.
Strict convexity of $Q\mapsto\KL(Q\|P_0^{(t)})$ makes the minimizing mixture law unique, and the
linear independence proved in Appendix~\ref{app:bayes} makes
$\lambda\mapsto\bar Q_\lambda^{(t)}$ injective, hence makes the weights unique. \qed

\section{Proofs for local and elementary rejection times}\label{app:twosided}

\subsection{Pathwise regret}\label{app:regretproof}

The mixture retains a positive multiple of the true-configuration likelihood ratio on every path.
Fix $S$ of size $m$ with positive weights $w_l^{(m)}$, and write
$\beta_l^{(m)}=\binom mlw_l^{(m)}$ and
$\bar E_{S,t}^{(l)}=e_l(\{L_{k,t}\}_{k\in S})/\binom ml$.
For a nonempty $c\subseteq S$ with $|c|=l$, $L_{c,t}=\prod_{k\in c}L_{k,t}$, and the hitting times are
$\tau_S^E(a)=\inf\{t:E_{S,t}\ge1/a\}$ and $\tau_c^L(a)=\inf\{t:L_{c,t}\ge1/a\}$.

\begin{proof}[Proof of Theorem~\mainnum{thm:regret}]
$E_{S,t}=\sum_{l'}\beta_{l'}^{(m)}\bar E_{S,t}^{(l')}
\ge w_l^{(m)}e_l(\{L_{k,t}\}_{k\in S})\ge w_l^{(m)}L_{c,t}$ since all dropped summands are
nonnegative. Taking logs and applying the definitions of the hitting times proves both claims.
\end{proof}

\paragraph{Sharpness and prior costs.} The which-subset term is minimax \emph{within the mixture class}: if a configuration mixture
allocates total mass $\beta_l$ arbitrarily among the $\binom{m}{l}$ size-$l$ configurations, some
configuration receives at most $\beta_l/\binom{m}{l}$, so the bound cannot be improved uniformly by
reallocating, and uniform allocation attains equality.  That is a statement about how to spend a
given level mass inside this family; it is not a finite-horizon lower bound asserting that every
symmetric e-process must pay $\log\binom ml$.  A uniform-in-horizon version does hold over all
e-processes.  For a finite family
$\mathcal C_0$ of distinct configurations, suppose a deterministic constant $a$ satisfies
$E_t\ge e^{-a}L_{c,t}$ almost surely for every $c\in\mathcal C_0$ and every $t$.  Then
\[
e^a\ge\E_0\max_{c\in\mathcal C_0}L_{c,t}
=\sum_{c\in\mathcal C_0}Q_c(\widehat c_t=c)\longrightarrow |\mathcal C_0|,
\]
where $\widehat c_t$ maximizes the configuration likelihood, with a fixed tie rule.  Under $Q_c$,
$t^{-1}\log(L_{c,t}/L_{d,t})\to |c\setminus d|D+|d\setminus c|\bar D>0$ for $d\ne c$, so the
finite-family maximum-likelihood classifier is eventually correct almost surely, proving the
limit.  Thus $a\ge\log|\mathcal C_0|$, and the uniform mixture over that family attains equality.
Taking all size-$l$ configurations gives $\log\binom ml$.  This is a bound on a deterministic
pathwise shortfall uniform over all horizons, not on finite-horizon or expected rejection delay.
The level-prior cost is set by the weights, which the analyst controls.  Its
\emph{drift-normalized} worst case $\max_l\log(1/w_l)/(lD)$, the surcharge measured in rounds at
the rate the evidence actually accumulates rather than the raw log cost, is minimized by the geometric
weights $w_l=\gamma_m^l$, $\gamma_m=2^{1/m}-1$, which equalize $\log(1/w_l)/l$ across $l$; this is a
minimax statement about the displayed bound, not an exact expected-delay optimum.  Under a fixed
configuration $c$ with $w_c>0$, every ratio $L_{d,t}/L_{c,t}$ for $d\ne c$ tends to zero almost
surely by the same drift calculation.  Consequently $E_{S,t}/L_{c,t}\to w_c$ and
$t^{-1}\log E_{S,t}\to lD$ almost surely: the local log-regret bound is asymptotically sharp,
including at dense configurations.  This local cost does not determine elementary closure delay,
where the singleton eventually binds under full density (Proposition~\mainnum{prop:sandwich}).
Omitting the true
configuration can change the exponential rate, which is what the corners do.

\subsection{Local rejection rates}\label{app:localrateproofs}

The lower bound is a change of measure at a bounded stopping time, and the upper bound and the
corner rates are the hitting-time law of a random walk with drift.  The setting is that of
\mainref{Proposition}{prop:twosided}: $S$ has $m$ streams, the configuration $\vec h$ of all $K$ streams restricts on $S$ to
$c$ with $|c|=l\ge1$, $t_\alpha=\log(1/\alpha)/(lD)$, $\tau_S^E(a)=\inf\{t:E_{S,t}\ge1/a\}$ and
$\tau_c^L(a)=\inf\{t:L_{c,t}\ge1/a\}$ are the hitting times of the PEM statistic and of the
configuration likelihood ratio $L_{c,t}=\prod_{k\in c}L_{k,t}$, and $\bar e_1(t)$ is the arithmetic
mean of the $L_{k,t}$, $k\in S$.  Let $\vec h^0$ agree with $\vec h$ outside $S$ and be null on
$S$, so $\vec h^0$ satisfies $H_S$.  By conditional independence the two laws differ only on the
streams in $c$, and $dP_{\vec h}/dP_{\vec h^0}=L_{c,t}$ on $\F_t$.

\paragraph{(i) Lower bound, any procedure.}  Let $(\sigma_\alpha)$ be $\F_t$-stopping times and
$N_\alpha\ge t_\alpha$ horizons with $\Pr_{\vec h'}(\sigma_\alpha\le N_\alpha)\le\alpha$ for every
configuration $\vec h'$ satisfying $H_S$.  Fix $\epsilon\in(0,1)$, put
$u_\alpha=\lfloor(1-\epsilon)t_\alpha\rfloor\le N_\alpha$, and choose $\eta>0$ with $(1+\eta)(1-\epsilon)<1$.  The
event $A=\{\sigma_\alpha\le u_\alpha,\ L_{c,\sigma_\alpha\wedge u_\alpha}\le e^{(1+\eta)lDu_\alpha}\}$ is
$\F_{\sigma_\alpha\wedge u_\alpha}$-measurable, and the change of measure at the bounded stopping
time $\sigma_\alpha\wedge u_\alpha$ gives
\[
\Pr_{\vec h}(A)=\E_{\vec h^0}\bigl[L_{c,\sigma_\alpha\wedge u_\alpha}\mathbf 1_A\bigr]
\le e^{(1+\eta)lDu_\alpha}\Pr_{\vec h^0}(\sigma_\alpha\le N_\alpha)
\le \alpha^{\,1-(1+\eta)(1-\epsilon)}\longrightarrow0 .
\]
On the complement of $A$ within $\{\sigma_\alpha\le u_\alpha\}$, $\max_{s\le u_\alpha}\log
L_{c,s}>(1+\eta)lDu_\alpha$.  Under $P_{\vec h}$, $\log L_{c,s}$ is a random walk with drift $lD>0$,
so by the strong law $\max_{s\le u}\log L_{c,s}/u\to lD$ almost surely as $u\to\infty$, and since
$u_\alpha\to\infty$ the probability of this event also tends to zero.  Hence
$\Pr_{\vec h}(\sigma_\alpha\le(1-\epsilon)t_\alpha)\to0$.

\paragraph{(ii) Upper bound, the PEM statistic.}  By \mainref{Theorem}{thm:regret},
the pathwise regret theorem,
$\tau_S^E(\alpha)\le\tau_c^L(\alpha w_l^{(m)})$ whenever $w_l^{(m)}>0$.  For a random walk with
drift $lD>0$ the hitting time of the level $x$ satisfies $\tau_c^L(e^{-x})/x\to1/(lD)$ almost surely
as $x\to\infty$, and $\log(1/(\alpha w_l^{(m)}))/\log(1/\alpha)\to1$; hence
$\limsup_{\alpha\to0}\tau_S^E(\alpha)/t_\alpha\le1$ almost surely.  Combining with (i), applied to
$\sigma_\alpha=\tau_S^E(\alpha)$ with $N_\alpha=\lceil t_\alpha\rceil$, which is admissible because the PEM
statistic is an e-process under $H_S$
(\mainref{Proposition}{prop:fwer}), gives
$\tau_S^E(\alpha)/t_\alpha\to1$ in $P_{\vec h}$-probability.  The deployed rejection time at
$b^\star(\alpha,N,m)\le1/\alpha$ is no later than $\tau_S^E(\alpha)$ whenever
$\tau_S^E(\alpha)\le N$, by the definition of $b^\star$.

\paragraph{(iii) The corners.}  $\log\bar e_1(t)$ lies between $\max_{k\in S}\log L_{k,t}-\log m$
and $\max_{k\in S}\log L_{k,t}$, and the maximum of the $m$ random walks has drift
$\max(D,-\bar D)=D$, so $t^{-1}\log\bar e_1(t)\to D$ almost surely and the arithmetic mean's
first-order delay is $\log(1/\alpha)/D$ at every $l$.  For the product, $\log\prod_{k\in S}L_{k,t}$
is a random walk with drift $lD-(m-l)\bar D$; when the drift is negative its supremum is finite
almost surely, so $\Pr_{\vec h}(\sup_t\prod_{k\in S}L_{k,t}\ge1/\alpha)\to0$ as $\alpha\to0$; when it
is positive, the hitting-time law used in (ii) gives first-order delay
$\log(1/\alpha)/(lD-(m-l)\bar D)$, which exceeds $t_\alpha$ by the factor $lD/(lD-(m-l)\bar D)$
when $l<m$. At zero drift, put $S_t=\log\prod_{k\in S}L_{k,t}$ and
$T(x)=\inf\{t:S_t\ge x\}$. Integrability and nondegeneracy imply oscillation of this random
walk, so every finite upper level is crossed almost surely. Its strong law gives $S_t/t\to0$;
at the first crossing of $x$, $x\le S_{T(x)}$, whence $T(x)/x\to\infty$ almost surely as
$x\to\infty$.  Thus crossing persists but there is no positive first-order linear rate. \qed

\begin{proof}[Proof of Corollary~\mainnum{cor:rate}]
Combine parts (i)--(iii) of Proposition~\mainnum{prop:twosided}, noting that (ii) requires only
$w_l^{(m)}>0$ for the true $l$, which full support guarantees for every $l$.  For the last claim,
(ii) gives $\tau_S^E(\alpha)/t_\alpha\to1$ in probability while \mainref{Equation}{eq:growinghorizon} gives
$N_\alpha/t_\alpha\ge1+\epsilon$ eventually for some $\epsilon>0$, so
$\Pr(\tau_S^E(\alpha)>N_\alpha)\le\Pr(\tau_S^E(\alpha)/t_\alpha>1+\epsilon)\to0$; on the complementary
event the deployed time is at most $\tau_S^E(\alpha)$ because the deployed boundary is no larger than
$1/\alpha$.  Its local level guarantee permits applying part (i) to the deployed time, which supplies
the matching lower bound in probability.  Thus its ratio to $t_\alpha$ converges to one in probability.
For elementary claims the same argument uses Corollary~\mainnum{cor:penalty} and the lower bound of
Lemma~\mainnum{lem:ceiling}.
\end{proof}

\subsection{Elementary rejection times}\label{app:elementaryproofs}

The deadline-specific ceiling follows by treating the other streams as independent randomization.
Let $\boldsymbol\tau$ be any procedure satisfying the strong FWER guarantee in
\mainref{Equation}{eq:fwerN} under every configuration, and let $\tau_k$ be its rejection time for $H_k$.
Fix a configuration $\vec h$ with active stream $k$. As in \mainref{Lemma}{lem:ceiling},
$\mathcal T_\alpha$ consists of stopping times based on stream $k$ and independent randomization
whose probability of stopping by $N$ under $f_0$ is at most $\alpha$.

\begin{proof}[Proof of Lemma~\mainnum{lem:ceiling}]
Let $\vec h'$ agree with $\vec h$ except that $h'_k=0$.  Under $\vec h'$ the hypothesis $H_k$ is
true, so \mainref{Equation}{eq:fwerN} gives $\Pr_{\vec h'}(\tau_k\le N)\le\alpha$.  Under both $\vec h$ and
$\vec h'$ the streams other than $k$ have the same law and, by Assumption~\mainnum{ass:indep}, are
independent of stream $k$.  Viewed as a function of stream $k$ and of the other streams, $\tau_k$ is
therefore a stopping time of the filtration generated by stream $k$ and an independent randomization,
and its level under $\vec h'$ is its level under $f_0$ on stream $k$: $\tau_k\in\mathcal T_\alpha$.
The displayed bound follows, the Neyman--Pearson statement is Remark~\mainnum{rem:degenerate}, and the
first-order statement is Proposition~\mainnum{prop:twosided}(i) with $S=\{k\}$, $l=1$, and
$N_\alpha=N$.
\end{proof}

For the remaining proofs, let $c$ be the active set, $|c|=l$, and $k\in c$. Run uncapped literal
closure with fixed full-support weights at $1/\alpha$, so that
$\tau_k=\max_{S\ni k}\tau_S^E(\alpha)$.

\begin{proof}[Proof of Corollary~\mainnum{cor:penalty}]
The displayed bound applies Theorem~\mainnum{thm:regret} to each $S\ni k$ with the true configuration
$c\cap S$ of $S$, and $\tau_k=\max_{S\ni k}\tau^E_S(\alpha)$.  For (a), the intersections with
$|c\cap S|=1$ contribute hitting times of random walks with drift $D$ and those with $|c\cap S|\ge2$
hitting times of walks with drift at least $2D$, so the maximum is governed by the intersections
with a single active stream, and its
first-order delay is $\log(1/\alpha)/D$ by the hitting-time law used in
Proposition~\mainnum{prop:twosided}(ii).  The lower bound is Lemma~\mainnum{lem:ceiling}.  For (b), an
intersection with $c\cap S=\{k\}$ has $|S|\le K-l+1$ and pays $\log(1/w_1^{(|S|)})$, which is
nondecreasing in $|S|$ for both defaults; when $l=K$ no such \emph{nonsingleton} intersection
exists, the singleton $S=\{k\}$ being precisely the binding case and costing nothing since
$w^{(1)}_1=1$, every $S\ni k$ with
$|S|\ge2$ has drift at least $2D$ and therefore crosses before $\{k\}$ for small $\alpha$, and
$E_{\{k\},t}=L_{k,t}$ with $w^{(1)}_1=1$.  The $e$-Bonferroni time is $\tau^L_{\{k\}}(\alpha/K)$ by
definition.  For arithmetic-mean closure, $\bar e_1(t)\ge L_{k,t}/|S|$ gives
$\tau_S\le\tau^L_{\{k\}}(\alpha/|S|)\le\tau^L_{\{k\}}(\alpha/K)$, and
Proposition~\mainnum{prop:twosided}(iii) gives the rate $D$ of each of its intersections.
\end{proof}

The pathwise comparison uses $x=\log(1/\alpha)$,
$T_k(x)=\inf\{t:\log L_{k,t}\ge x\}$, and
$C=\max_{m\le K-l+1}\log(1/w_1^{(m)})$, the largest singleton log-boundary surcharge among
intersections containing $k$ and no other active stream. Write
$\tau_k(e^{-x})=\max_{S\ni k}\tau_S^E(e^{-x})$ for the uncapped closure time at level $e^{-x}$.

\begin{proof}[Proof of Proposition~\mainnum{prop:sandwich}]
For the lower bound, $S=\{k\}$ is one of the intersections closure must clear and
$E_{\{k\},t}=L_{k,t}$ because $w_1^{(1)}=1$, so $\tau_k(e^{-x})\ge T_k(x)$ on every path.  For the
upper bound, take any $S\ni k$ and let $d=c\cap S$, which contains $k$ and so is nonempty.
Theorem~\mainnum{thm:regret} gives $\log E_{S,t}\ge\log L_{d,t}-\log(1/w^{(|S|)}_{|d|})$ pathwise.  If
$|d|=1$, then $d=\{k\}$ and every stream of $S$ other than $k$ is null, so $|S|\le K-l+1$ and
$\log E_{S,t}\ge\log L_{k,t}-C$; hence $S$ has crossed by $T_k(x+C)$.  If $|d|\ge2$, then
$\log L_{d,t}$ is a random walk with drift $|d|D\ge2D$, so its hitting time of $x+\log(1/w^{(|S|)}_{|d|})$
divided by $x$ tends almost surely to $1/(|d|D)\le1/(2D)$, while $T_k(x+C)/x\to1/D$; as there are
finitely many such $S$, almost surely for all large $x$ every one of them has crossed strictly before
$T_k(x+C)$.  Taking the maximum over the finitely many $S\ni k$ gives the upper bound.
\end{proof}

\paragraph{Default weights and comparison with the arithmetic mean.} The level-uniform default has $w_1^{(m)}=1/m^2$ and hence
$C=2\log(K-l+1)$, which at $l=1$ is $2\log K$; the geometric default has
$C=\log\{1/(2^{1/(K-l+1)}-1)\}$, which at $K=10$, $l=1$ is $2.64$ against $\log K=2.30$.  For
arithmetic-mean closure the corresponding surcharge is $\log K$, since $\bar e_1(t)\ge L_{k,t}/|S|$
for every $S\ni k$; that its intersections all grow at the same first-order rate $D$ does not by
itself prove that its nonsingleton intersections remain binding on every path, and lattice examples
exist in which $e$-Bonferroni and the singleton coincide along selected boundary sequences.

For a size-$m$ intersection containing only one active stream, the PEM statistic's log evidence
is asymptotically below the arithmetic mean's by $\log(1/\beta_1^{(m)})$ as time grows.
For the geometric default $\beta_1=m(2^{1/m}-1)$, which at $m=10$ equals $0.718$ and gives
$\log(1/\beta_1)=0.332$ nats; as $m$ grows $\beta_1$ decreases to $\log2=0.693$ and the constant
increases to $\log(1/\log2)=0.367$.  This is
consistent with Proposition~\mainnum{prop:growmean} and with the observed arithmetic-mean advantage
at the least-favorable singleton configuration in \S\mainnum{sec:sims}; the asymptotic log-evidence
comparison does not itself order finite-level rejection times.  Under
dense evidence the PEM statistic's higher-order terms clear the intersections containing several
active streams at rate $2D$ or more, while the arithmetic mean must wait for each of them at rate
$D$.  This explains the different log-boundary comparisons, while Proposition~\mainnum{prop:sandwich}
gives the precise eventual guarantee for full-support PEM closure.

\section{Self-financing reweighting within the accumulated ESP family}\label{app:imposs}

Predictable reweighting of the PEM statistic preserves the martingale property only under a
pathwise balance constraint. Under a fixed design prior, it cannot improve expected log evidence
over the Bayes-matched PEM statistic.  This appendix concerns adaptivity within the local accumulated-ESP family; it is not a statement about all adaptive e-processes.
Write $L_t=(L_{k,t})_{k\in S}$. Set $E_0=1$, and for $t\ge1$ suppose
$E_t=\sum_lw_{l,t}e_l(L_t)$, where the nonnegative weights $w_{\cdot,t}$ are predictable in the local
filtration ($\F_{t-1}^S$-measurable). Assume all displayed terms are integrable and impose the initial
unit-wealth condition $\sum_l\binom mlw_{l,1}=1$. Since each $e_l(L_t)$ is a martingale,
$\E[e_l(L_t)\mid\F_{t-1}^S]=e_l(L_{t-1})$; the initial condition gives
$\E[E_1\mid\F_0^S]=E_0$, and thereafter the martingale property is equivalent to the pathwise
\emph{balance constraint}
\begin{equation}\label{eq:balance}
\sum_l\,(w_{l,t}-w_{l,t-1})\,e_l(L_{t-1})=0
\qquad\text{almost surely, for every }t\ge2.
\end{equation}
Two consequences follow.

\emph{(i) Deterministic schedules must be constant.} If the schedule $w_{\cdot,t}$ is deterministic,
Equation~\eqref{eq:balance} is a linear identity with constant coefficients holding
$P_0^{(t-1)}$-almost surely.  No support condition on the likelihood-ratio vector is needed to
conclude from it.  Write $\Delta w_l=w_{l,t}-w_{l,t-1}$ and expand the elementary symmetric
polynomial into its configuration terms, $e_l(L_{t-1})=\sum_{|c|=l}L_{c,t-1}$ with
$L_{c,t-1}=dQ_c^{(t-1)}/dP_0^{(t-1)}$.  Then~\eqref{eq:balance} says
\[
\sum_{\varnothing\ne c\subseteq S}\Delta w_{|c|}\,\frac{dQ_c^{(t-1)}}{dP_0^{(t-1)}}=0
\qquad P_0^{(t-1)}\text{-a.s.},
\]
that is, the signed measure $\sum_c\Delta w_{|c|}Q_c^{(t-1)}$ vanishes.  Appendix~\ref{app:bayes}
shows that $\{Q_c^{(t-1)}:c\ne\varnothing\}$ is linearly independent whenever $D>0$, using only that
$f_0\ne f_A$; hence every $\Delta w_{|c|}$ is zero, so $w_{l,t}=w_{l,t-1}$ for every $l$ and every
$t\ge2$.  This covers both the Bernoulli basket experiments and the Gaussian experiments without requiring
the likelihood-ratio vector to have full support.  Here $w_{\cdot,1}$ is the initially
normalized choice of prior, and no pre-planned retuning of that prior after the first observation is
compatible with the martingale property.

\emph{(ii) Martingale-preserving predictable reweighting is self-financing.}
Path-dependent solutions to~\eqref{eq:balance} do exist: from an interior weight vector with $m=2$, take
$w_{1,t}=w_{1,t-1}+\varepsilon_te_2(L_{t-1})$ and
$w_{2,t}=w_{2,t-1}-\varepsilon_te_1(L_{t-1})$ with $\varepsilon_t\ge0$ predictable and small enough to
keep the weights nonnegative. This solves~\eqref{eq:balance}, hence keeps $E_t$ a mean-one test
martingale (and therefore an e-process),
though it does not preserve the pathwise level-normalization $\sum_lw_{l,t}\binom{m}{l}=1$ (that
constraint is not needed after initialization for validity: $E_0=1$ and the martingale property give
$\E_{P_0^S}[E_t]=1$;
for $m\ge3$ a normalized direction exists as described next). Equation~\eqref{eq:balance}
says precisely that every such move is \emph{self-financing}: mass bought on one level is paid for by
mass sold on the others at their current values $e_l(L_{t-1})$, so the realized value at the
switch time is unchanged (this constrains the value at the switch, not the whole future trajectory).

No self-financing reweighting improves on the fixed Bayes-matched PEM statistic in expected log evidence.  Because every such process $E$ is an e-process, \mainref{Theorem}{thm:bayes} gives
$\Psi_{t,\omega}(E)\le\Psi_{t,\omega}(E^\omega)$ for every fixed design prior $\omega$ and horizon $t$.
When $\omega$ is exchangeable, the comparator belongs to the level-weight family above. Thus
reweighting cannot improve on the Bayes likelihood ratio in expected log evidence under the specified
design law. This is not pathwise dominance or a universal, design-uniform impossibility theorem.

Under pathwise prior normalization, genuinely data-dependent accumulated weights exist for $m\ge3$ but generically not for $m=2$.  If pathwise prior normalization is also imposed, a move $\Delta w_t$ must satisfy
\[
\sum_{l=1}^m\Delta w_{l,t}e_l(L_{t-1})=0,
\qquad
\sum_{l=1}^m\binom ml\Delta w_{l,t}=0.
\]
For $m\ge3$, the orthogonal complement of the two displayed coefficient vectors has dimension at
least $m-2$. Starting from an interior weight vector, any nonzero predictable direction in that
complement can be scaled to preserve nonnegativity. Thus genuinely data-dependent, normalized
accumulated weights exist for $m\ge3$. For $m=2$, the two constraints generically leave only the zero
direction.

\paragraph{A predictable plug-in.} A less constrained data-driven alternative mixes fresh
per-step ratios $r_s=(r_{k,s})_{k\in S}$ using weights chosen from the past:
\[
E^{\rm plug}_{S,t}:=\prod_{s\le t}M_s,\qquad
M_s=\sum_l\pi_{l,s-1}\,\tilde e_l(r_s),
\]
with predictable weights $\pi_{l,s-1}\propto\pi^0_l\,\tilde e_l(L_{s-1})$, the running posterior on
$l=1,\ldots,m$, conditional on the nonempty composite alternative, under a fixed baseline level prior
$\pi^0$. It is valid for any predictable simplex-valued weights:
each $\tilde e_l(r_s)$ is a fresh mean-one per-step e-value and $\pi_{l,s-1}$ is
$\F_{s-1}^S$-measurable, so
$\E_{P_0^S}[M_s\mid\F_{s-1}^S]=\sum_l\pi_{l,s-1}=1$. Being an e-process, the plug-in likewise satisfies
$\Psi_{t,\omega}(E^{\rm plug})\le\Psi_{t,\omega}(E^\omega)$ for every fixed $(t,\omega)$ by
\mainref{Theorem}{thm:bayes}.
In particular,
if the baseline level prior $\pi^0$ is also the design level prior, then
$\omega_c=\pi^0_{|c|}/\binom{m}{|c|}$.

Except at the all-alternatives corner, no symmetric single-level merge attains the configuration-aware oracle per-step rate $lD$.  For a configuration $c$ of size $l$, write $R_c(r_s)=\prod_{k\in c}r_{k,s}$. For any fixed level
$q=1,\ldots,m$, independence gives
\[
\E_{Q_c^{(1)}}\!\left[\frac{\tilde e_q(r_s)}{R_c(r_s)}\right]
=\binom{m}{q}^{-1}\sum_{|d|=q}
\E_{Q_c^{(1)}}\!\left[
\prod_{k\in d\setminus c}r_{k,s}
\prod_{k\in c\setminus d}r_{k,s}^{-1}\right]=1,
\]
because $\E_{f_0}[r]=1$ and $\E_{f_A}[1/r]=1$. Jensen's inequality therefore gives
$\E_{Q_c^{(1)}}[\log\tilde e_q(r_s)]\le lD$. Equality requires
$\tilde e_q(r_s)=R_c(r_s)$ $Q_c^{(1)}$-almost surely, equivalently $P_0^{(1)}$-almost surely by the
standing mutual absolute continuity; the linear-independence argument in
Appendix~\ref{app:bayes} shows that this occurs only in the all-alternatives product corner
$c=S$, $q=m$. The predictable plug-in mixes levels with data-driven weights; its asymptotic rate
depends on posterior concentration and integrability conditions that we do not analyze here.

\smallskip\noindent(This appendix is cited in \S\mainnum{sec:opt} of the main article.)

\section{Dependence: validity of the intersection merges}\label{app:dependence}

This section proves the main article's dependence-validity claims while relaxing its conditional
independence assumption. It does not extend the Bayes or GROW optimality statements beyond the
product model. Throughout, the synchronously observed vectors
$X_{\cdot,t}=(X_{1,t},\dots,X_{K,t})$ are i.i.d.\ over $t$ from a joint law $P$ with \emph{arbitrary}
cross-hypothesis dependence whose $k$-th margin is $f_0$ (if $H_k$ is null) or $f_A$ (if false).
Temporal independence ($X_{\cdot,t}\perp\F_{t-1}$) gives, for any null stream $k$,
$\E_P[r_{k,t}\mid\F_{t-1}]=\E_{f_0}[r_{k,t}]=1$, so each $L_{k,t}$ is a mean-one $\F_t$-martingale
\emph{regardless} of the dependence across hypotheses. This fact drives the dependence-robust
arithmetic-mean result below.

\paragraph{(A) Weighted arithmetic means are valid under arbitrary dependence.}
For the intersection null $H_S$ with $|S|=m$, the merge
\[
\bar e_1(t)=m^{-1}\sum_{k\in S}L_{k,t}
\]
is a nonnegative $\F_t$-martingale, because each $L_{k,t}$ is a mean-one martingale. Hence the closed test whose
intersection e-processes are $\bar e_1$ controls anytime-valid FWER in the strong sense under arbitrary
cross-hypothesis dependence. More generally, for any prespecified coefficients
$a_0,\{a_k:k\in S\}\ge0$ satisfying $a_0+\sum_{k\in S}a_k=1$, the process
\[
a_0+\sum_{k\in S}a_kL_{k,t}
\]
is a test martingale. The displayed $\bar e_1$ is the equal-weight, permutation-symmetric member of
this family. Weighted arithmetic averages, allowing weight on the constant $1$, characterize the
admissible merging functions for arbitrarily dependent e-values; under permutation symmetry this
reduces to a mixture of $\bar e_1$ and $1$~\citep{vovk2021,wang2025}.

\paragraph{(B) Higher terms: necessary and sufficient conditions under temporal independence.}
A fixed nonnegative monomial mixture is an e-process for the specified temporally i.i.d.\ null law
if and only if every supported monomial has cross-moment at most one.  For $|c|=l\ge2$, temporal independence gives
\[
\E_P\Big[\prod_{k\in c}L_{k,t}\Big]=\rho_c^{\,t},\qquad
\rho_c=\E_P\Big[\prod_{k\in c}r_k\Big]=\int\prod_{k\in c}\frac{f_A}{f_0}(x_k)\,dP(x).
\]
Here $\rho_c$ is computed under the true law of the $c$-coordinates, which for a genuine strong-FWER
guarantee must be considered over every joint law consistent with the intersection null (nuisance
coordinates outside $S$ unconstrained). If $\rho_c\le1$ for every $c$ in the support of $w$
(configurations with $w_c=0$ impose no constraint), then
\[
\E_P[E_{S,t}\mid\F_{t-1}]=\sum_c w_c L_{c,t-1}\rho_c\le\sum_c w_c L_{c,t-1}=E_{S,t-1},
\]
so $E_{S,t}$ is a supermartingale and is anytime valid by Ville.

Conversely, the coefficientwise condition is necessary, and under temporal independence this needs
no support assumption at all.  Because the vectors are i.i.d.\ over $t$ and the weights are fixed and
nonnegative, expectations add:
\[
\E_P[E_{S,t}]=\sum_c w_c\,\E_P[L_{c,t}]=\sum_c w_c\,\rho_c^{\,t}.
\]
If some $c$ in the support of $w$ has $\rho_c>1$, then $\E_P[E_{S,t}]\ge w_c\rho_c^{\,t}\to\infty$,
so at some deterministic time $t$, which is a bounded stopping time, the e-process budget
$\E_P[E_{S,t}]\le1$ fails.  Hence $\rho_c\le1$ for every supported $c$ is necessary as well as
sufficient, and the characterization is complete for temporally i.i.d.\ vectors and fixed weights.
For temporally dependent increments the moments may instead depend on the past; neither this
fixed-moment characterization nor its expectation calculation transfers without an additional
conditional-moment argument.

Three regimes follow for each monomial $L_{c,t}$: $\rho_c=1$ (conditional independence, or any
copula with $\E_P[\prod_{k\in c}r_k]=1$) makes it a martingale; $\rho_c<1$ (a moment
condition, met for instance under suitable negative dependence) makes it a conservative
supermartingale; and $\rho_c>1$ makes its mean grow as $\rho_c^t$ and invalidates that e-process.
The whole mixture is an exact martingale when every supported $c$ has $\rho_c=1$. With a nondecreasing likelihood ratio
$g=f_A/f_0$ and positively associated null coordinates, meaning that increasing functions of
these coordinates have nonnegative covariance,
(e.g.\ nonnegatively correlated Gaussians, associated by Pitt's theorem~\citep{pitt1982};
multivariate total positivity of order two, $\mathrm{MTP}_2$, suffices), the
Esary--Proschan--Walkup inequality~\citep{esary1967} gives
$\E_P[\prod_{k\in c}g(X_k)]\ge\prod_{k\in c}\E_P[g(X_k)]=1$, i.e.\ $\rho_c\ge1$; association
alone does not imply strictness. In the one-sided Gaussian model
$r_k=\exp(\mu X_k-\mu^2/2)$ with null correlation matrix $(\rho_{ij})$, the exact expression is
\[
\rho_c=\exp\!\left(\mu^2\sum_{\substack{i<j\\i,j\in c}}\rho_{ij}\right),
\]
which is strictly above one whenever the displayed correlation sum is positive.
Conditional independence is thus sufficient for $\rho_c=1$, not necessary.

\paragraph{(C) A known dependent null, and why the naive latent mixture fails.} Integrating conditional likelihood ratios over a latent mixing law does not in general produce an e-process; the valid object is the ratio of the two marginal mixtures.  One might hope to integrate conditional ESP likelihood ratios over an alternative latent mixing law
$\Pi$ on a factor $U$. Let $\Pi_0$ be the law of $U$ under the null, and suppose the conditional-law
kernels below are specified on their common parameter space, at least $(\Pi_0+\Pi)$-almost
everywhere. Write $P_{0,u}^{(t)}=P_0^{(t)}(\cdot\mid u)$ and
$\bar Q_u^{(t)}=\bar Q^{(t)}(\cdot\mid u)$. A conditional Radon--Nikodym derivative
$d\bar Q_u^{(t)}/dP_{0,u}^{(t)}$ is defined only $P_{0,u}^{(t)}$-almost surely, whereas the proposed
integral is evaluated under the marginal null. To make that proposal unambiguous, suppose for this
calculation that the conditional laws have densities $p_{0,u}^{(t)}$ and $\bar q_u^{(t)}$ with respect
to a common $\sigma$-finite measure $\nu_t$, that $\bar Q_u^{(t)}\ll P_{0,u}^{(t)}$ for
$\Pi$-almost every $u$, and that, for $\Pi$-almost every $u$ and $\nu_t$-almost every $x$,
\[
p_{0,u}^{(t)}(x)>0\quad\text{whenever}\quad
p_0^{(t)}(x):=\int p_{0,u'}^{(t)}(x)\,d\Pi_0(u')>0
\]
holds. Equivalently, up to null sets, the marginal null is dominated by each denominator used in the
integral. The naive construction is
\[
\tilde E_{S,t}(x)
=\int\frac{\bar q_u^{(t)}(x)}{p_{0,u}^{(t)}(x)}\,d\Pi(u).
\]
Tonelli's theorem gives
\[
\begin{split}
\E_{P_0^{(t)}}[\tilde E_{S,t}]
&=\int\!\!\int
  \frac{\bar q_u^{(t)}(x)}{p_{0,u}^{(t)}(x)}
  p_0^{(t)}(x)\,d\nu_t(x)\,d\Pi(u)\\
&=\int\!\!\int\!\!\int
  \frac{p_{0,u'}^{(t)}(x)}{p_{0,u}^{(t)}(x)}
  \bar q_u^{(t)}(x)\,
  d\nu_t(x)\,d\Pi_0(u')\,d\Pi(u),
\end{split}
\]
which need not equal one. If the displayed domination condition fails, the proposal has the additional
problem that it can depend on arbitrary versions of the conditional likelihood ratios outside their
conditional-null supports.

The valid construction uses one fixed marginal denominator. Let
\[
\bar q^{(t)}(x)=\int\bar q_u^{(t)}(x)\,d\Pi(u).
\]
If the resulting marginal laws $\bar Q$ and $P_0$ are consistent on the full sequence and
$\bar Q^{(t)}\ll P_0^{(t)}$ for every $t$, then
\[
E_{S,t}(x)
=\frac{d\bar Q^{(t)}}{dP_0^{(t)}}(x)
=\frac{\bar q^{(t)}(x)}{p_0^{(t)}(x)}
\]
is the nonnegative $P_0$ density martingale and hence an anytime-valid e-process. Thus latent and
configuration mixtures are allowed, but the valid object is the ratio of the two marginal mixtures,
not the mixture of conditional ratios. The marginal denominator generally does not factorize, so the
ESP closed form and its prefix-recursion implementation are generally lost.

Transferring such a construction into closed testing needs more than a valid marginal object.
\mainref{Proposition}{prop:fwer} requires, for each $S$, an e-process in the \emph{full} filtration
$(\F_t)$, valid under \emph{every} law in the intersection null family $H_S$, with the streams
outside $S$ unconstrained.  A density process built for one specified dependent law of the
coordinates in $S$ delivers neither by itself: it is adapted to the local filtration, and it is
calibrated against one member of the null family rather than the whole of it.  Both premises have to
be re-established for the dependent model in question before the closure argument may be applied:
the full-filtration property by checking that the nuisance coordinates cannot invalidate the
stopping-time expectation bound, and the null-family property by proving that bound for each
allowed law.

\paragraph{(D) Prespecified mixed deployment.} Closed testing needs only a valid e-process for each
intersection, so different e-processes may be preassigned using model assumptions or external
dependence information: for example, an ESP e-process where the joint model is justified and the
arithmetic mean elsewhere. Selecting among them after inspecting their realized values is not covered
by \mainref{Proposition}{prop:fwer} and requires a separate valid e-process
construction. Moreover, a
robust e-process at one intersection can bind several elementary decisions through closure; its power cost
need not remain confined to a single elementary hypothesis.

\smallskip\noindent(This appendix is cited in \S\mainnum{sec:disc} of the main article.)

\section{Stopping-rule proofs and local optimal-stopping benchmarks}\label{app:dpvalidation}

\subsection{Terminal tests and single-stream stopping}\label{app:stoppingproofs}

The terminal test uses the exact likelihood ratio against the design mixture. For a configuration
prior $w$, recall $\bar Q_w^{(N)}=\sum_cw_cQ_c^{(N)}$ and
$E^w_{S,N}=\sum_cw_cL_{c,N}$.

\begin{proof}[Proof of Corollary~\mainnum{cor:np}]
By \mainref{Equation}{eq:lr_identity}, $E^w_{S,N}=d\bar Q_w^{(N)}/dP_0^{(N)}$.  The claim is the
Neyman--Pearson lemma for the simple pair $(P_0^{(N)},\bar Q_w^{(N)})$.
\end{proof}

The single-stream stopping problem rewards early rejection. For $m=1$, write $E_t=L_{1,t}$.
The Lagrangian payoff under the null is $g_t(E)=E(N-t)-\kappa$, where $\kappa>0$ penalizes a
rejection, and never stopping earns zero. The proof compares this payoff with the value of
continuing to sample.

\begin{proof}[Proof of Theorem~\mainnum{thm:levelset}]
At $t=N$ the payoff is $g_N\equiv-\kappa<0$ while never stopping earns $0$, so the value is
$U_N\equiv0$ and the stopping region at $N$ is empty, that is $b_N=\infty$.  For $t<N$, under $P_0$
we have $E_s=E_tR$ for $s>t$ with $R\ge0$ independent of $\F_t$ and $\E R=1$, so the value function
is a function $U_t$ of $(E_t,t)$; as a supremum of affine functions of $E_t$ with nonnegative
slopes, together with the constant $0$, it is convex, nondecreasing and nonnegative.  Write
$C_t(E)=\E[U_{t+1}(ER)]$ for the continuation value, so $U_t=\max(g_t,C_t)$.

We show by backward induction that $U_t$ is Lipschitz with constant at most $N-t$.  The base case
$U_N\equiv0$ has constant $0$.  Assume $U_{t+1}$ is $L$-Lipschitz with $L\le N-t-1$.  Then for
$E'>E\ge0$,
\[
|C_t(E')-C_t(E)|\le\E\bigl[\,\bigl|U_{t+1}(E'R)-U_{t+1}(ER)\bigr|\,\bigr]
\le L\,\E[R]\,(E'-E)=L\,(E'-E),
\]
so $C_t$ is $L$-Lipschitz with $L\le N-t-1$.  No differentiability is used: the estimate is a
difference quotient, which is all a convex function supplies.  Since $g_t$ has slope exactly $N-t$,
\[
(g_t-C_t)(E')-(g_t-C_t)(E)\ \ge\ (N-t)(E'-E)-(N-t-1)(E'-E)\ =\ E'-E\ >\ 0,
\]
so $g_t-C_t$ is strictly increasing; and $g_t(0)=-\kappa<0\le C_t(0)$.  It therefore changes sign
exactly once, the stopping region $\{g_t\ge C_t\}$ is an upper interval $\{E\ge b_t\}$, and $b_t$ is
finite because $g_t$ grows at rate $N-t$ and $C_t$ at rate at most $N-t-1$.  Finally
$U_t=\max(g_t,C_t)$ is Lipschitz with constant at most $\max(N-t,N-t-1)=N-t$, which is the
hypothesis needed at $t-1$, closing the induction.
\end{proof}

\subsection{Certified and numerical local benchmarks}\label{app:localbenchmarks}

This section reports what the unrestricted optimum of the local rejection-time objective gains over
the local PEM threshold rule at its calibrated boundary, on the problems where that optimum can be bracketed.  The
benchmark is \mainref{Problem}{eq:costobj} with the conventions fixed there: a single claim, the
claim's own filtration enlarged by an independent randomization, the design mixture as the reward
measure, and the option never to stop.  It is not the vector problem under closure, and its value is
not total saved observations across streams.

\paragraph{Primal and dual bounds.}  The solver bisects the Lagrange multiplier and retains a
numerically feasible endpoint.  For the exact-arithmetic derivation, let $J_\kappa$ and $p_\kappa$
be the reward and null crossing probability of a Lagrangian-optimal policy with
$p_\kappa\le\alpha$.  On a
finite lattice $p_\kappa$ is a step function of $\kappa$, so the returned policy generally lands
strictly inside the constraint; its value $J_\kappa$ is then the value of one feasible policy and is
a lower bound on the constrained optimum $J^\star$, never an estimate of it.  The flat rule is also
feasible, so $\max(J_{\rm flat},J_\kappa)\le J^\star$.  For the upper bound, weak duality at that
same multiplier gives, for every feasible $\sigma$,
\[
\E[(N-\sigma)^+]\ \le\ \E[(N-\sigma)^+]-\kappa\bigl(\Pr_0(\sigma\le N)-\alpha\bigr)
\ \le\ J_\kappa+\kappa(\alpha-p_\kappa),
\]
the first step because $\kappa\ge0$ and $\Pr_0(\sigma\le N)\le\alpha$, the second because the
Lagrangian-optimal policy maximizes the middle expression over all stopping times.  Hence
\mainref{Equation}{eq:dpbracket}.  The bracket is reported as a percentage of $J_{\rm flat}$: the
left column of each pair is $J_\kappa/J_{\rm flat}-1$, the feasible-policy gain, and the right is
$\{J_\kappa+\kappa(\alpha-p_\kappa)\}/J_{\rm flat}-1$, the certified upper endpoint for the
optimum's gain.  A negative left entry says only that the returned feasible policy is worse than the
flat rule at a smaller error spend; the corresponding right entry is what bounds the optimum.

Outward-rounded interval arithmetic certifies the Bernoulli policies' error probabilities and
the upper bounds on their possible gains: lower endpoints are rounded down and upper endpoints
up. The decimal Bernoulli probabilities are treated as exact rationals.
The geometric weight base $\gamma_m$ is enclosed between binary floating-point numbers whose
$m$th powers after adding one bracket $2$, checked using integer arithmetic; the level-uniform
weights are rational.  Positive ESP, transition and reward recursions then round each addition
and multiplication outwards.  A separate upper Bellman recursion uses the maximum of stopping,
continuation and zero, so the dual upper bound remains valid even if a floating-point near-tie
selects a suboptimal policy.  The fixed count-state decisions of the returned and flat policies
are evaluated with the same interval arithmetic, including terminal rejection for the flat
policy.  Every reported Bernoulli policy has an upper null-probability endpoint below its nominal
level.  Dividing the dual upper endpoint by the flat reward's lower endpoint gives the reported
upper gain; the table rounds this endpoint upwards.  These checks certify the implemented
count-state policies, including their numerical tie convention, rather than silently identifying
them with a different real-arithmetic decision at a boundary atom.

The alternative to bracketing is to solve the constrained problem itself, by locating the critical
multiplier and randomizing between the two adjacent optimal policies.  We do not do that here;
interpolating arbitrary feasible and infeasible policies without verifying optimality at the
critical multiplier would not establish the constrained optimum either.

\paragraph{Comparison conventions.}  Every entry favors the dynamic program slightly.  The flat rule
is calibrated and evaluated at every $t\le N$ including the deadline, which is what is deployed,
whereas the program never exercises at $t=N$, where its payoff is $-\kappa<0$
(\mainref{Theorem}{thm:levelset} with $b_N=\infty$).  The flat rule therefore spends part of its
budget at the deadline for no return, and the comparison charges it for that.

\paragraph{Stopping-region diagnostic.} On states with positive surviving null mass, no violation of the up-set property was detected:
increasing the evidence coordinates never changed a stopping decision to continuation. This diagnostic
does not establish the monotonicity used by the shortcut of \S\mainnum{sec:closure}, which instead
uses the scalar PEM statistic's coordinatewise monotonicity and a threshold common to subsets of each
size.

\paragraph{Which designs are covered.}  Table~\ref{tab:dpvalidation} solves the NCI design at its
actual horizon $N=31$ and the five-basket planned design at its actual horizon $N=13$, together with
the nominal five-basket cells at the additional horizon settings $N=24$ and $N=31$.  Two limits should
be read with it.  The Bernoulli cells reach $m=4$, and $m=5$ at $N=13$, whereas the NCI design
deploys $K=10$; its larger intersections are unsolved, and we do not extrapolate to them from the
ratio $N/t_{\mathrm{cross}}$ alone.  The five-basket realized design has unequal caps
$(20,10,8,18,7)$, so no common-horizon cell describes it and none is claimed.  Across the NCI
design at $N=31$ together with the nominal five-basket cells at $N=24$ and $N=31$, the largest
certified upper endpoint is $1.69\%$, so on that set of solved local Bernoulli problems the
constrained optimum improves on the flat rule by less than $1.7\%$.  That is a statement about
those cells, not about every deployment, and it is separate from the $6.7\%$ endpoint of the
planned design at its own shorter horizon.

\paragraph{The Gaussian panel is numerical.}  Its solver discretizes the log-evidence increment onto
a lattice of spacing $h=0.25$, truncates the increment kernel at six standard deviations,
exponentially tilts it so that its exponential has mean one before state truncation,
and clamps the finite state grid.  Its left edge is the null mean at $N$ minus four cumulative
standard deviations; its right edge is $\log(1/\alpha)+\max(4\mu,5)$.
Clamping does not preserve the exact martingale identity at the edges.  The reported lattice
reward therefore approximates the null-measure expression in
\mainref{Equation}{eq:waldreduction}; the clamped gate is not asserted to be an exact mixture
density process for the clamped chain.
Table~\ref{tab:dpconvergence} reports a targeted
spacing and domain study at four cells.  Widening the domain from four to six standard deviations
does not change the displayed two-decimal gains in the checked two-dimensional cases; this is a
sensitivity check, not a proved truncation-error bound.  Halving the spacing does move
the entries, and by more than the reported precision at the shortest horizon: the $m=2$, $N=15$
level-uniform cell falls from $7.92\%$ to $6.99\%$ and is then stable to $0.01$ percentage points
under a further halving, while the geometric cell at the same horizon moves by about $0.19$ and the
$N=50$ cell by about $0.04$ percentage points without settling monotonically.  The one $m=3$ cell
checked behaves like the $m=2$ ones, moving from $8.85\%$ to $8.74\%$ when the spacing is halved; a
further halving there costs a $313^3$ state space per look and was not run.  The Gaussian entries
are therefore reported as numerical comparisons with the observed sensitivity, not a certified
discretization-error estimate.  A
continuous-model certificate would need a bound on the approximation error.

We also evaluated the $m=2$, $N=15$ level-uniform policy at $h=0.0625$ on $200{,}000$ fresh
continuous Gaussian null paths and $200{,}000$ fresh design-mixture paths, without recalibration.
At each look the continuous log likelihood-ratio vector is mapped to its nearest lattice point, clamped at the
edges, before applying the learned decision.  The null crossing estimates were $0.049705$ for
this policy and $0.050145$ for the flat rule at the retained lattice threshold, with two-sided
$99\%$ exact binomial intervals $[0.04846,0.05098]$ and $[0.04889,0.05142]$, rounded outwards.
The respective mixture rewards were $4.3140$ and $4.0432$ rounds (standard errors $0.0108$ and
$0.0097$); the paired reward difference was $0.2708$ (standard error $0.0041$).  These independent
checks do not certify level $0.05$ in the continuous model or an optimality gap there.

\begin{table}[!tbp]
\centering
\small
\caption{Local stopping comparisons relative to the calibrated flat PEM rule.  Bernoulli cells
give percentage gains and the bound on the unrestricted optimum from
\mainref{Equation}{eq:dpbracket} as \emph{feasible gain / certified upper endpoint}, with upper
endpoints rounded upwards after interval certification; a negative
feasible gain is a feasible policy landing inside the constraint, not a worse optimum.  Gaussian
cells give the feasible-policy gain only and are numerical, with the discretization error assessed
in Table~\ref{tab:dpconvergence}.  \(t_{\mathrm{cross}}=\log(1/\alpha)/D\), with \(D=\mu^2/2=0.125\)
for the Gaussian panel and \(D=\KL(p_{\rm alt}\|p_0)\) for the Bernoulli panel.  Rows marked
\emph{actual} are solved at the horizon the design in \S\mainnum{sec:apps} actually uses.}
\label{tab:dpvalidation}
\begin{tabular}{rrrrrr}
\multicolumn{6}{l}{\emph{Gaussian} (numerical), \(\mu=0.5\), \(\alpha=0.05\), \(t_{\mathrm{cross}}=24.0\), \(h=0.25\)}\\
\toprule
 & & \multicolumn{2}{c}{\(m=2\)} & \multicolumn{2}{c}{\(m=3\)}\\
\cmidrule(lr){3-4}\cmidrule(lr){5-6}
\(N\) & \(N/t_{\mathrm{cross}}\) & geometric & level-unif. & geometric & level-unif.\\
\midrule
15 & 0.63 & 12.51 & 7.92 & --- & ---\\
20 & 0.83 & 7.05 & 3.95 & 8.85 & 2.49\\
25 & 1.04 & 5.14 & 3.02 & --- & ---\\
30 & 1.25 & 3.21 & 1.54 & 3.94 & 1.08\\
40 & 1.67 & 1.78 & 1.21 & --- & ---\\
50 & 2.09 & 0.81 & 0.58 & 0.92 & 0.37\\
100 & 4.17 & 0.04 & 0.08 & --- & ---\\
\bottomrule
\end{tabular}

\vspace{1.0em}

\footnotesize
\setlength{\tabcolsep}{3.5pt}
\begin{tabular}{@{}L{4.0cm}rrrrr@{}}
\multicolumn{6}{@{}l}{\emph{Bernoulli} (exact lattice), feasible gain / certified upper endpoint}\\
\toprule
design and design prior & \(N/t_{\mathrm{cross}}\) & \(m=2\) & \(m=3\) & \(m=4\) & \(m=5\)\\
\midrule
NCI, geometric ($\alpha=0.05$, $N=31$, \emph{actual}) & 2.33 & 1.11 / 1.26 & 0.95 / 1.04 & 1.23 / 1.25 & ---\\
NCI, level-uniform ($\alpha=0.05$, $N=31$, \emph{actual}) & 2.33 & 0.55 / 0.65 & $-0.03$ / 0.48 & $-0.19$ / 0.18 & ---\\
Five-basket planned, geometric ($\alpha=0.10$, $N=13$, \emph{actual}) & 1.44 & 4.75 / 4.87 & 4.05 / 4.10 & 4.06 / 5.47 & 4.02 / 6.62\\
Five-basket, geometric ($\alpha=0.10$, $N=24$) & 2.66 & 0.80 / 0.81 & 1.58 / 1.69 & 0.41 / 0.87 & ---\\
Five-basket, geometric ($\alpha=0.10$, $N=31$) & 3.43 & 0.32 / 0.39 & 0.26 / 0.42 & 0.44 / 0.48 & ---\\
\bottomrule
\end{tabular}
\end{table}

\begin{table}[!tbp]
\centering
\small
\caption{Spacing and domain study for the Gaussian panel, \(\mu=0.5\), \(\alpha=0.05\).  Entries are
the feasible-policy gain in percent; \(h\) is the log-evidence lattice spacing and \(n_{\rm sd}\) the
left-tail width parameter (the asymmetric right edge is defined in the text), with \(h=0.25\),
\(n_{\rm sd}=4\) the setting used in Table~\ref{tab:dpvalidation}.  The checked domain changes do not
affect the displayed gains, whereas spacing does, and the
shortest horizon is where it matters most.  A dash marks a setting not run: the cost of one solve
grows like the number of grid points to the power \(m\), so the finest spacings were affordable only
at \(m=2\).}
\label{tab:dpconvergence}
\setlength{\tabcolsep}{4pt}
\begin{tabular}{@{}L{3.7cm}rrrrrr@{}}
\toprule
 & \multicolumn{3}{c}{\(n_{\rm sd}=4\)} & \multicolumn{2}{c}{\(n_{\rm sd}=6\)} & \\
\cmidrule(lr){2-4}\cmidrule(lr){5-6}
cell & \(h=0.25\) & \(h=0.125\) & \(h=0.0625\) & \(h=0.25\) & \(h=0.125\) & grid, \(h=0.25\)\\
\midrule
\(m=2\), \(N=15\), geometric      & 12.51 & 12.70 & 12.63 & 12.51 & 12.70 & 72\\
\(m=2\), \(N=15\), level-uniform  & 7.92  & 6.99  & 6.99  & 7.92  & 6.99  & 72\\
\(m=2\), \(N=50\), geometric      & 0.81  & 0.85  & 0.81  & 0.81  & 0.85  & 115\\
\(m=3\), \(N=20\), geometric      & 8.85  & 8.74  & ---   & ---   & ---   & 79\\
\bottomrule
\end{tabular}
\end{table}

\smallskip\noindent(This appendix is cited in \S\mainnum{sec:dpvalidation} of the main article.)

\section{Gaussian simulation details}\label{app:simdetails}

The four Gaussian experiments of the main article illustrate its theoretical predictions at the stated
Monte Carlo sizes; this appendix gives the boundaries they use, the complete tables, and the
conventions behind each reported number.  All use \(f_0=\mathcal N(0,1)\),
\(f_A=\mathcal N(0.5,1)\), \(\alpha=0.05\), \(D=0.125\), and \(t_{\mathrm{cross}}=24.0\).  The
level-uniform prior \(\beta^{(m)}_l=1/m\), or \(w_l^{(m)}=\{m\binom ml\}^{-1}\), is the \(a=b=1\)
member of the beta-binomial family \(w_l^{(m)}\propto B(l+a,m-l+b)\), with \(B\) the beta function,
induced by a \(\mathrm{Beta}(a,b)\) prior on the per-hypothesis alternative probability, conditional
on at least one alternative.  The choice \(a=b=1\) gives a uniform prior on that probability before
conditioning; it is a design choice rather than a unique representation of ignorance.

\paragraph{Calibrated boundaries.}  Table~\ref{tab:boundary} lists the boundaries
\(\widehat b(\alpha,2{,}500,m)\) used by the \(K=6\) closed test, computed by the Monte Carlo route
from \(50{,}000\) null paths at confidence \(\gamma=0.999\).  Three quantities must be kept apart
here: the population minimum \(b^\star\), which is not computed; the fraction of calibration paths
whose maximum reaches \(\widehat b\), which is an empirical proportion; and the Clopper--Pearson
upper confidence bound on the population crossing probability, which is what the acceptance rule
tests and what the validity statement rests on.  Every size accepts the same hit count,
\(2{,}350\) of \(50{,}000\), so every row has empirical fraction \(0.04700\) and the same
\(0.999\)-confidence bound \(0.049996\).  Equal accepted hit counts across sizes, or across the
priors of Table~\ref{tab:priorfamily}, do not mean equal true operating levels: they mean only that
each boundary was advanced until it fell into the same acceptance class.  The calibration
certificate is pointwise in size: each displayed boundary has a \(0.999\)-confidence upper bound
of \(0.049996\) on its null crossing probability.  A union bound gives simultaneous confidence
at least \(0.994\) for the six rows.  The in-sample fraction is \(94\%\) of the nominal budget;
it is not an independent estimate of the selected boundary's population spend.  The flat threshold
is \(1/\alpha=20\) at every size, so calibration lowers the
threshold by \(20\)--\(33\%\) in this table, and the flat threshold's estimated spend is between
\(63\%\) and \(76\%\), the shortfall broadly growing with \(m\) because a mixture over a larger
family overshoots \(1/\alpha\) further.  Those percentages describe this table, not every merge:
at \(K=10\) the calibrated arithmetic boundary is \(16.638\) and the product boundary \(8.601\),
reductions of \(16.8\%\) and \(57.0\%\), while the matched-level reference and \(e\)-Bonferroni keep
flat thresholds by construction.

\begin{table}[!tbp]
\centering
\small
\caption{The Monte Carlo boundary $\widehat b(\alpha,N,m)$ for the $K=6$ Gaussian design ($\mu=0.5$,
$\alpha=0.05$, horizon $N=2{,}500$, level-uniform prior), from $50{,}000$ null paths at confidence
$\gamma=0.999$.  ``Hits'' is the number of calibration paths whose maximum reaches $\widehat b$ and
``fraction'' is that count over $50{,}000$; neither is the population crossing probability.  ``CP
bound'' is the one-sided $0.999$ Clopper--Pearson upper bound on that probability, the quantity the
acceptance rule tests.  The last column is the Monte Carlo point estimate of the crossing
probability at the flat threshold $1/\alpha=20$.}
\label{tab:boundary}
\begin{tabular}{rrrrrrr}
\toprule
$m$ & $\widehat b$ & $\log\widehat b$ & hits & fraction & CP bound & est.\ level at $1/\alpha$\\
\midrule
1 & 15.99 & 2.7722 & 2{,}350 & 0.04700 & 0.049996 & 0.03784\\
2 & 15.33 & 2.7300 & 2{,}350 & 0.04700 & 0.049996 & 0.03574\\
3 & 14.83 & 2.6968 & 2{,}350 & 0.04700 & 0.049996 & 0.03400\\
4 & 13.92 & 2.6330 & 2{,}350 & 0.04700 & 0.049996 & 0.03278\\
5 & 13.73 & 2.6195 & 2{,}350 & 0.04700 & 0.049996 & 0.03174\\
6 & 13.45 & 2.5991 & 2{,}350 & 0.04700 & 0.049996 & 0.03220\\
\bottomrule
\end{tabular}
\end{table}

\paragraph{Dense scaling.}  Under dense evidence the exact closed delay is nearly flat in \(K\),
whereas the \(e\)-Bonferroni delay grows with \(\log K\); Table~\ref{tab:kscale} gives the complete
study.  The study runs closed testing at the product corner, the Bayes-optimal member for the
all-alternatives configuration, and the exact closed delay is the time until every intersection
containing the selected hypothesis has crossed.  From \(K=4\) upward the first-order benchmark \(1+\log K/\log(1/\alpha)\)
exceeds the observed ratio.  This drift benchmark omits discrete boundary overshoot, calibration,
and the closure's residual cost; at \(K=2\), the observed ratio instead exceeds the benchmark,
\(1.28\) against \(1.23\).
\(e\)-Bonferroni uses the fixed threshold \(K/\alpha\) and is unchanged by calibration.

\begin{table}[!tbp]\centering
\caption{Mean rejection delay of a hypothesis vs.\ $K$, all-alternatives configuration, exact closure
at the product corner (all $2^{K-1}$ intersections containing the hypothesis; $6{,}000$ Monte Carlo
replications, horizon $8{,}000$, no censoring observed, standard errors $\le0.4$). The estimated
delay gain over $e$-Bonferroni increases overall across the examined range of $K$; the final column is
the drift benchmark. The ratio is $e$-Bonferroni delay divided by closed-testing delay,
computed from unrounded means.}
\label{tab:kscale}
\begin{tabular}{rrrrr}
\toprule
$K$ & closed delay & $e$-Bonf delay & ratio & benchmark $1+\tfrac{\log K}{\log(1/\alpha)}$\\
\midrule
2 & 25.1 & 32.1 & 1.28 & 1.23\\
4 & 26.4 & 37.8 & 1.43 & 1.46\\
6 & 26.9 & 40.6 & 1.51 & 1.60\\
8 & 27.8 & 43.4 & 1.56 & 1.69\\
12 & 29.5 & 46.0 & 1.56 & 1.83\\
\bottomrule
\end{tabular}
\end{table}
\FloatBarrier

\paragraph{Sparsity sweep at \(K=10\).}  The level-uniform PEM statistic crosses on every simulated path at every displayed \(l\), whereas product crosses on a minority of paths at \(l\le3\).  Table~\ref{tab:competitors} gives the full sparsity-axis comparison; the supporting drift calculations are in Appendix~\ref{app:hetero}.  At each displayed \(l\), the
level-uniform PEM statistic crosses by the horizon on all 3,000 paths.  Product crosses on only 245,
485, and 830 paths at \(l=1,2,3\), or \(8.2\%\), \(16.2\%\), and \(27.7\%\), matching its negative
log drift below half-density, and on 1,608 at \(l=4\), or \(53.6\%\): one active stream short of
half-density, the calibrated threshold carries a majority of paths where the flat threshold left
the product stalling.  At \(l=7\), the product median of 4 is one round earlier than the PEM median of 5, and
under calibration the exact distribution-free 95\% order-statistic intervals separate, \([4,4]\)
against \([5,5]\), where under the flat threshold both were \([5,6]\) and the ordering was unresolved.  The arithmetic mean is faster at \(l=1\), where the full mixture
pays both the which-subset and level-prior costs.  These are full-intersection outcomes, not elementary
closed-testing power.

\begin{table}[!tbp]\centering
\caption{Censoring-aware median full-intersection delay by configuration $l$ at $K=10$ ($3{,}000$
replications, horizon $2{,}500$; noncrossers coded as the horizon plus one). ``Stall'' means fewer than half the
paths cross by the horizon; the binomial standard errors of the corresponding crossing
fractions are at most $0.0092$. The product
merge stalls for $l\le3$; the arithmetic mean
is robust but slow under dense alternatives; the PEM statistic tracks the better of the two corners closely for
$l\ge2$ (not always strictly fastest, e.g.\ $l=7$) and does not stall in the configurations shown.
For the highlighted $l=7$ comparison, exact distribution-free $95\%$ order-statistic intervals are
$[5,5]$ for the PEM median and $[4,4]$ for the product median.
The final column is the level-$l$ mixture, not the configuration-aware likelihood-ratio oracle.}
\label{tab:competitors}
\begin{tabular}{rrrrr}
\toprule
$l$ & PEM (ours) & arith.\ mean & product & matched-level $\tilde e_l$\\
\midrule
1  & 46 & 37 & stall & 38\\
2  & 24 & 26 & stall & 22\\
3  & 15 & 22 & stall & 14\\
4  & 10 & 19 & 40 & 10\\
5  & \phantom{0}7 & 17 & \phantom{0}8 & \phantom{0}8\\
7  & \phantom{0}5 & 15 & \phantom{0}4 & \phantom{0}5\\
10 & \phantom{0}3 & 13 & \phantom{0}2 & \phantom{0}3\\
\bottomrule
\end{tabular}
\end{table}

\FloatBarrier

\paragraph{Full closed test at \(K=6\).}  Table~\ref{tab:fwer} reports the literal closed test at
\(K=6\), including all \(2^K-1=63\) intersections, across configuration classes, monitored at the
calibrated boundaries of Table~\ref{tab:boundary}.  Through horizon \(2{,}500\), the maximum FWER
estimate is \(0.0461\); after accounting for selection over the six null-containing classes, the
largest Bonferroni-adjusted one-sided \(95\%\) Clopper--Pearson upper bound is \(0.049822<0.05\).
That margin is thin, since four further false rejections among the maximal row's \(20{,}000\)
replications would carry it past \(\alpha\), and it bounds Monte Carlo precision under a six-way
selection adjustment, not the validity of PEM closure.  Rerunning the identical experiment at the
flat threshold \(1/\alpha\) returns a maximum of \(0.0369\), about four fifths of the calibrated
figure: the lower calibrated boundaries use more of the nominal error budget.  Every false null is rejected by the horizon on every simulated path; the corresponding
exact lower bound is \(0.999850\) within each nonnull class and \(0.999761\) simultaneously over
the six classes.  The pooled median elementary rejection delay falls from \(43\) at \(l=1\), where
null-padded intersections of five sizes must clear, to \(22\) at \(l=6\), where every nonsingleton
intersection can pool multiple active streams and the calibrated boundaries sit below \(1/\alpha\).

\begin{table}[!tbp]
\centering
\small
\caption{Full closed testing at $K=6$ under the calibrated boundary of Table~\ref{tab:boundary}:
finite-horizon Monte Carlo estimates for all configuration classes ($20{,}000$ replications,
horizon $N=2{,}500$; pointwise plug-in binomial standard errors $\le0.002$).  ``Alternative
completion'' is the probability that every alternative hypothesis is rejected by $N$.  ``Pooled
median alt delay'' is the median of individual alternative-hypothesis rejection times over paths
and alternative coordinates, with non-rejections coded as $N+1$.  No alternative non-rejections by
$N$ were observed; the simultaneous bounds in the text use all six nonnull configuration classes
displayed below.}
\label{tab:fwer}
\begin{tabular}{lrrr}
\toprule
configuration & FWER by $N$ & alternative completion by $N$ & pooled median alt delay\\
\midrule
all null ($l{=}0$) & 0.014 & N/A & N/A\\
$l{=}1$ & 0.018 & 1.00 & 43\\
$l{=}2$ & 0.022 & 1.00 & 40\\
$l{=}3$ & 0.031 & 1.00 & 37\\
$l{=}4$ & 0.040 & 1.00 & 32\\
$l{=}5$ & 0.046 & 1.00 & 27\\
all alt ($l{=}6$) & N/A & 1.00 & 22\\
\bottomrule
\end{tabular}
\end{table}

\paragraph{Design-prior sweep at \(K=10\).}  The bound-minimax geometric weights equalize
\(\log(1/w_l)/l\).  This worst-case log-regret criterion does not determine finite-level
crossing delays.  Table~\ref{tab:priorfamily} sweeps a
prespecified prior grid; the beta-binomial family is indexed by \(\mathrm{Beta}(a,b)\) on the
per-hypothesis alternative probability, and the binomial \(q\) member is the point mass on that
probability used in the basket sensitivity of Table~\ref{tab:basketncisensitivity}.  The
\(\max_l\log(1/w_l)/l\) column is infinite for the arithmetic mean because it puts zero weight on
cardinalities above one.  The full-support priors have finite bounds at every cardinality.
Their per-configuration weights need not decrease across \(l\): for
\(w_l\propto B(l+a,m-l+b)\) at \(a=b=1\) and \(m=10\) the weights are \(0.01\) at \(l=1\), fall to
\(3.97\times10^{-4}\) at \(l=5\), and rise again to \(0.1\) at \(l=10\), because the level mass
\(\beta_l=1/m\) is constant while the number \(\binom ml\) of configurations sharing it is not.
Here the level mass is constant; other beta-binomial choices need not have monotone level masses.
The \(\beta_{10}\) column below reports the level mass on the all-alternatives configuration, which
for that member is \(0.1\).  Every row except the arithmetic mean falls by a
factor of at least four from \(l=1\) to \(l=10\), whereas the mean falls by less than three,
illustrating the rate pattern of \mainref{Corollary}{cor:rate} at this finite level.

Prior choices with similar worst-case bounds can still give different dense-configuration delays.
Here differences are median delays relative to geometric, with \(95\%\) paired
percentile-bootstrap intervals from 2,000 resamples of the evaluation paths, holding calibrated
boundaries fixed. Within the family, the \(\mathrm{Beta}(1/2,8)\) and binomial \(q=0.10\)
members agree to within \(0.04\) on \(\max_l\log(1/w_l)/l\), yet their all-alternatives mass
differs by five orders of magnitude.  At \(l=1\) neither member separates from geometric: the point
estimates are \(+1.0\) and \(+0.5\) rounds, with \(95\%\) confidence intervals \([+0.0,+1.0]\) and
\([-1.0,+1.0]\), so the sparse price the two pay is not resolvable from zero at this Monte Carlo
size.  The dense gain is where they part, and there it is sharp: the beta-binomial member returns
\(-2.0\) rounds at \(l=K\), \([-2.0,-2.0]\), exactly twice the binomial member's \(-1.0\),
\([-1.0,-1.0]\), both intervals excluding zero.  Under the flat threshold the binomial member's
dense interval covered zero, so that separation rested on the point estimates alone; calibrating
each prior's threshold resolves it away from zero.  Nor is the ordering confined to the cell in
Table~\ref{tab:priorfamily}: the beta-binomial member's dense gain is strictly the larger of the two
in all ten cells of the sweep, the primary cell and the nine robustness cells at
\(K\in\{6,10,15\}\) and \(\mu\in\{0.35,0.5,0.7\}\).  The robustness cells run at a smaller Monte
Carlo size and shorter horizon.  The rows with the largest dense gains have longer median delays
at the sparse configuration in this experiment.  The binomial member nonetheless matches or beats
geometric at every \(l\) in Table~\ref{tab:priorfamily} to the table's integer resolution; its
\(+0.5\)-round point estimate at \(l=1\) sits inside an interval that covers zero, and the binomial
member and geometric separate only at \(K=15\), where the sparse configuration is hardest.
The geometric prior remains the primary choice because it minimizes the worst-case bound: max regret \(2.63\) against
\(2.82\), a design-time property of the weights rather than a measured delay.  The accompanying analysis
code produces the full sweeps.

\begin{table}[!tbp]
\centering
\footnotesize
\caption{Design-prior sweep at \(K=10\), \(\mu=0.5\), \(\alpha=0.05\), on the Monte Carlo size and
horizon of the sparsity sweep (3,000 paths, horizon 2,500) with common random numbers across
priors.  Every row is monitored at its own \(\widehat b\), each accepted at the same calibration
hit count \(2{,}350/50{,}000\) and hence the same \(0.999\) Clopper--Pearson bound \(0.049996\).
Thus every row uses the same calibration protocol; an equal accepted hit count is not an equal true
operating level. Both calibration and evaluation runs use seeds distinct from those of the
sparsity sweep, so the arithmetic and level-uniform rows do not reproduce the medians of
Table~\ref{tab:competitors} exactly.
Entries are median first crossings of the full-intersection statistic.  ``max reg'' is
\(\max_l\log(1/w_l)/l\), minimized by geometric; \(\beta_{10}\) is the design-prior mass on the
all-alternatives configuration.  Level-uniform is the \(a=b=1\) beta-binomial member.}
\label{tab:priorfamily}
\begin{tabular}{@{}lrrrrrr@{}}
\toprule
Design prior & \(l=1\) & \(l=2\) & \(l=5\) & \(l=10\) & max reg & \(\beta_{10}\)\\
\midrule
Arithmetic (GROW)                  & $36$ & $25$ & $17$ & $13$ & $\infty$ & $0$\\
Geometric (primary)                & $36$ & $23$ & $13$ & $8$  & $2.63$ & $3.6\times10^{-12}$\\
Level-uniform $=\mathrm{Beta}(1,1)$ & $46$ & $23$ & $8$  & $3$  & $4.61$ & $1.0\times10^{-1}$\\
$\mathrm{Beta}(1/2,1/2)$           & $45$ & $23$ & $8$  & $3$  & $4.49$ & $2.1\times10^{-1}$\\
$\mathrm{Beta}(1/2,2)$             & $39$ & $21$ & $8$  & $4$  & $3.55$ & $1.0\times10^{-2}$\\
$\mathrm{Beta}(1/2,8)$             & $37$ & $22$ & $11$ & $6$  & $2.86$ & $1.8\times10^{-5}$\\
Binomial $q=0.10$                  & $36$ & $22$ & $12$ & $7$  & $2.82$ & $1.5\times10^{-10}$\\
\bottomrule
\end{tabular}
\end{table}

\smallskip\noindent(This appendix is cited in \S\S\mainnum{sec:ebonf} and \mainnum{sec:sims} of the main article.)
\FloatBarrier
\section{Heterogeneity-robustness growth table}\label{app:hetero}

For $c_l=\{1,\dots,l\}\subseteq[K]$, let $Q_l:=Q_{c_l}$ denote the one-step product law in the
Gaussian model used in the main article's simulations: the coordinates in $c_l$ have law $\mathcal N(0.5,1)$ and the remaining
coordinates have law $\mathcal N(0,1)$. By exchangeability, the expectation below is the same for every
configuration of size $l$. For each single-level merge $\tilde e_q=e_q/\binom{K}{q}$,
Table~\ref{tab:hetero} reports $\E_{Q_l}[\log\tilde e_q(r)]$ ($\times100$) at $K=10$.
These are drift diagnostics for the merges themselves. The accumulated statistics in the main article's
closed-testing procedure grow at rate $lD$ for any full-support prior, as shown in
\mainref{Corollary}{cor:rate}. Only the product row is additive in log evidence and therefore has the same per-step and
accumulated growth rate. In the equal-variance Gaussian model, that product rate is $(2l-K)D$.

The product ($\tilde e_{10}$) stalls (negative
growth) for $l<K/2$; $\tilde e_1$ is positive for every $l$; the matched level peaks near the
diagonal: no fixed merge is uniformly best, motivating mixtures across levels. The per-step GROW
prior is exactly the uniform singleton distribution, so its row duplicates $\tilde e_1$: the sparsest
configuration binds and no cross-level borrowing improves the worst case, as
\mainref{Proposition}{prop:growmean} proves.

\begin{table}[!tbp]\centering
\caption{Expected one-step log-growth ($\times100$) of the per-step merges at $K=10$ by configuration
$l$ (number of alternatives; $300{,}000$ Monte Carlo draws per column). The nonproduct entries are
Monte Carlo estimates with standard errors at most $0.15$ on the displayed $\times100$ scale; the
product row is exact. Negative drift means that eventual crossing is not almost sure, and zero is the boundary case.
The product has negative drift for $l<K/2$, while $\tilde e_1$ has positive drift in every displayed
nonnull configuration; the per-step GROW row equals $\tilde e_1$ and is omitted.}
\label{tab:hetero}
\begin{tabular}{lrrrrrr}
\toprule
merge & $l{=}1$ & $l{=}2$ & $l{=}3$ & $l{=}5$ & $l{=}8$ & $l{=}10$\\
\midrule
$\tilde e_1$ (average) & 1.4 & 4.1 & 6.8 & 11.8 & 19.1 & 23.6\\
$\tilde e_2$ & 0.0 & 5.4 & 10.7 & 20.8 & 35.4 & 44.6\\
$\tilde e_5$ & $-19.1$ & $-6.1$ & 6.9 & 32.0 & 69.1 & 92.9\\
$\tilde e_{10}$ (product) & $-100$ & $-75$ & $-50$ & 0.0 & 75.0 & 125\\
\bottomrule
\end{tabular}
\end{table}
\smallskip\noindent(This appendix is cited from Appendix~\ref{app:simdetails} and supports \S\mainnum{sec:sims} of the main article.)

\section{Random-order finite-population evidence}\label{app:finitepop}

The likelihood-ratio model in the main text describes prospective i.i.d.\ streams.  A fixed benchmark
requires a different sampling statement: its questions are a finite population, not an i.i.d.\ sample
from an unspecified task distribution.  The following construction supplies exact randomization
validity for the benchmark applications without extending the Bayes or GROW optimality claims beyond
the product model.

\begin{proposition}[Random-order validity]\label{prop:finitepop}
For stream $k$, fix values $z_{k,1},\ldots,z_{k,N_k}\in[-1,1]$. Let $\Pi_k$ be a uniform random
permutation of $[N_k]$, independently across $k$, and define $Z_{k,t}=z_{k,\Pi_k(t)}$. The
finite-population null is
\[
H_k^{\rm fp}:\sum_{i=1}^{N_k}z_{k,i}\le0.
\]
For $1\le t\le N_k$, immediately before draw $t$, let
\[
s_{k,t-1}=\sum_{j<t}Z_{k,j},\qquad n_{k,t}=N_k-t+1,
\qquad b_{k,t}=\min\{1,-s_{k,t-1}/n_{k,t}\}.
\]
Write $s_{k,t}=s_{k,t-1}+Z_{k,t}$ after draw $t$.
Fix $0<\eta<1$.  On every history compatible with $H_k^{\rm fp}$ and having $b_{k,t}>-1$, first
compute the finite linear factor
\[
\widetilde R_{k,t}=1+\eta\frac{Z_{k,t}-b_{k,t}}{1+b_{k,t}}.
\]
Define $R_{k,t}$ in priority order: set $R_{k,t}=+\infty$ if the current observation makes
$s_{k,t}>N_k-t$; otherwise set $R_{k,t}=\widetilde R_{k,t}$ when $b_{k,t}>-1$, and set
$R_{k,t}=1$ in the remaining null-compatible boundary case $b_{k,t}=-1$ and $Z_{k,t}=-1$.
After a null-impossible history, keep the extended process absorbing at $+\infty$ (equivalently, any
nonnegative continuation can be chosen off the null support).  For $t>N_k$, set $R_{k,t}=1$ and keep
$L_{k,t}^{\rm fp}=L_{k,N_k}^{\rm fp}$.
Then, under every fixed population in $H_k^{\rm fp}$,
$L_{k,t}^{\rm fp}=\prod_{j\le t}R_{k,j}$ agrees, outside a set of probability zero under that
population's own randomization law, with a nonnegative supermartingale, and is hence an e-process.
The $+\infty$ convention applies only on histories that are impossible under every population in
$H_k^{\rm fp}$.  For any intersection $S$, every fixed nonnegative mixture
\[
E_{S,t}^{w,\rm fp}=\sum_{\substack{\varnothing\ne c\subseteq S\\w_c>0}}w_c
                    \prod_{k\in c}L_{k,t}^{\rm fp},
\qquad \sum_cw_c=1,
\]
is therefore an intersection e-process.  Restricting the sum to positive weights avoids the
undefined product $0\cdot\infty$ on null-impossible histories; the implementation likewise skips
zero-weight terms.  Closed testing with these e-processes controls strong
anytime-valid FWER for the finite-population nulls.
\end{proposition}

\begin{proof}
Under $H_k^{\rm fp}$, the sum of the $n_{k,t}$ remaining values is at most
$-s_{k,t-1}$.  Uniform random order therefore gives
\[
\E[Z_{k,t}\mid\F_{t-1}]
\le \min\{1,-s_{k,t-1}/n_{k,t}\}=b_{k,t}.
\]
For $b_{k,t}>-1$, the smallest possible factor occurs at $Z_{k,t}=-1$ and equals
$1-\eta>0$, while its conditional expectation is at most one.  If $b_{k,t}=-1$, the remaining
mean is at least $-1$ and at most $-1$, so every remaining value equals $-1$ on a
null-compatible history.  A value $b_{k,t}<-1$ can occur only after a history incompatible with the
null.  Likewise, conditional on a previously compatible history, an observation producing
$s_{k,t}>N_k-t$ has probability zero under every population in $H_k^{\rm fp}$.  The immediate
$+\infty$ convention is thus an almost surely equivalent extended-value reporting convention under
each null and records that the null has become logically impossible, even on the last draw. Unit factors after exhaustion
preserve this property. This proves the marginal e-process assertion. Independent permutations make the
next-draw factors conditionally independent across streams; hence the conditional expectation of
each product multiplier is the product of quantities no larger than one.  Nonnegative fixed
mixtures preserve the supermartingale property, and \mainref{Proposition}{prop:fwer}
completes the closed-testing argument.
\end{proof}

The construction tests the realized finite benchmark contrast.  A claim about performance on a
future task population instead needs a prospective sampling or exchangeability assumption.  In the
language-model application, $Z_{k,t}$ is candidate correctness minus incumbent correctness in activity stratum $k$.
The exact guarantee uses mathematically independent uniform permutations. The recorded seeded streams
are their reproducible pseudorandom implementation, not by themselves a proof of realized
independence or exact uniformity.
The implementation handles the $b=-1$ boundary and histories that make the null impossible
explicitly, operates in log space, and uses factor one for prespecified idle updates after a stream
is exhausted.  Restricting the primary display to a common horizon no larger than
$\min_kN_k$ avoids any reliance on those idle updates.

\smallskip\noindent(This appendix is cited in \S\mainnum{sec:apps} of the main article.)

\section{Frozen small-GPT application protocol}\label{app:llmprotocol}

\subsection{Population, randomization, and task definition}

The HellaSwag source is \texttt{Rowan/hellaswag} at revision
\hashid{218ec52e09a7}.  The labeled validation population has 10,042
records (ordered-record SHA-256 prefix \hashid{9fb1890335fd}); the disjoint train split
used only for engineering smoke tests has 39,905 records (prefix
\hashid{f2b6149f71c8}).  Full revisions and digests are retained in the frozen design manifests. The dataset card and
upstream repository record an MIT license.  Applying the outcome-blind rule
``validation activity-label count at least 160'' gives the complete finite populations in
Table~\ref{tab:llmstratumresults}.  Ten seeds sampled independently using a cryptographically secure pseudorandom number generator
initialize the stratum
permutations, with no seed derivation across strata; their full values are retained in the design
manifest rather than repeated here.

The query/choice transformation is pinned to the HellaSwag task in the Language Model Evaluation
Harness~\citep{eleutherharness2026}.  The audit records identify the harness revision by prefix \hashid{e90e5b656bcc}, the task
utility and YAML by \hashid{07d23b60b28f} and \hashid{e9ef8ac3fed0}, and the ordered confirmatory
design by \hashid{2b871a1ed6b7}.  The corresponding manifests retain each full identity.

\subsection{Models, tokenizers, and scoring}

The frozen contrasts are \texttt{EleutherAI/\allowbreak gpt-neo-1.3B} versus
\texttt{EleutherAI/\allowbreak gpt-neo-125M}, at revision prefixes \hashid{dbe59a7f4a88} and
\hashid{21def0189f57}, and \texttt{gpt2-large} versus \texttt{gpt2-medium}, at prefixes
\hashid{32b71b12589c} and \hashid{6dcaa7a952f7}.  The manifests retain the full revisions.
The pinned cards identify all four as MIT-licensed.  The size names 1.3B, 125M, large, and medium are
model labels rather than audited counts; every inference manifest reports the sum of \texttt{numel}
over the unique tensors returned by \texttt{named\_parameters} before any numerical result is used.

The four pinned tokenizers have 50,257 entries and the same complete token-to-ID map, whose SHA-256
prefix is \hashid{79ff2372d75e}.  Hashing that map together with all special-token IDs gives
prefix \hashid{f33016e4a397}.  Both full identities are retained in the score manifests, recomputed
for each role, and checked for equality again during scoring.  Exactly as in the pinned task,
the query is formed from the activity label, \texttt{ctx\_a}, and capitalized \texttt{ctx\_b}; the
cleanup removes WikiHow bracket artifacts from the query and choices.  Query and choice are joined
by one ASCII space and retokenized as a whole; the separately tokenized query must be an exact
prefix.  Summed continuation-token log probability is divided by the preprocessed choice's Python
character count, in accordance with the task's \texttt{acc\_norm} convention.  The source option order and lowest-index
tie rule are fixed.  The reference source files supplying this transformation are pinned by the
hashes above.  All 1,610 selected items passed the fail-closed prefix and context-limit checks, so
every scored query--choice sequence remained within the common 1,024-token limit.

Write $C$ for the candidate model of a contrast, the larger checkpoint whose improvement is being
tested, and $I$ for the incumbent, the smaller one; $r\in\{C,I\}$ indexes the role.  For item
$(k,i)$, stratum $k$ and position $i$ within it, HellaSwag supplies four source endings indexed by
$j\in\{0,1,2,3\}$, and $a_{ki}$ denotes the index of the correct one as the pinned dataset records
it.  Let $u_{ki}$ be the preprocessed query, $e_{kij}$ the preprocessed choice, $d_{kij}$ its Python
character count, and $v_{kij,1:J_{kij}}$ the continuation tokens produced by retokenizing their
concatenation. Write $p_r$ for model $r$'s next-token conditional probabilities.
The executed score and paired outcome are
\[
q^{(r)}_{kij}=d_{kij}^{-1}\sum_{s=1}^{J_{kij}}
\log p_r(v_{kij,s}\mid u_{ki},v_{kij,1:s-1}),\qquad
\widehat a^{(r)}_{ki}=\argmax_j q^{(r)}_{kij},
\]
\[
Z_{k,i}=\mathbf 1\{\widehat a^{(C)}_{ki}=a_{ki}\}
-\mathbf 1\{\widehat a^{(I)}_{ki}=a_{ki}\},
\]
with the lowest source index breaking a tie.  This makes the finite-population estimand and every
normalization choice explicit.

\paragraph{Numerical implementation and audit.} Real inference uses local pinned snapshots,
half-precision model forward passes, float32 log-softmax,
and Python-binary64 accumulation of the resulting float32 scalar log probabilities, followed by
character-count normalization.  The item-level metadata's shorthand ``score accumulation in
float32'' describes the source scalars rather than this executed accumulation path; the stored scores
reflect the latter.  Evaluation mode, disabled TensorFloat-32 (TF32) arithmetic, deterministic PyTorch algorithms, and batch
size one are enforced.  Software, numerical flags, GPU identity,
checkpoint revision, tokenizer identity, exact parameter count, and Git commit are carried into the
score manifest.  That manifest also content-binds the complete design file and its manifest, whose
semantic design digest commits to every selected item's prompt hashes.  Candidate and incumbent
shards cannot be combined across different software identities, and all shards within a role must
have the same recorded hardware identity.  The checked-in configuration files for the two
contrasts are the authoritative machine-readable protocol.

\subsection{Audited small-GPT outcomes}\label{app:llmresults}

Both model pairs completed without prefix, context-overflow, tokenizer-alignment, or score-join
failures. Independent reconstruction from the scored item records reproduced every reported global
crossing, exact closed-testing decision, $\eta$ sensitivity, and conditional-order replay summary.

The audited runs used float16 on Tesla V100-PCIE-16GB devices and were generated from
\auditdetail{Git commit \hashid{3d3db54661f3} (full identity in the score manifests)}
{a software revision frozen before analysis}.
For GPT-Neo, the exact parameter counts were $1{,}315{,}575{,}808$ for the candidate and
$125{,}198{,}592$ for the incumbent; for GPT-2, they were $774{,}030{,}080$ and $354{,}823{,}168$,
respectively. Every role reproduced the common tokenizer identity in the protocol above.

Table~\ref{tab:llmmethods} compares canonical full-intersection crossings, exact-closure rejections, and
conditional-order replay costs for every primary merger and anytime baseline.

\begin{table}[!tbp]
\centering
\footnotesize
\setlength{\tabcolsep}{3pt}
\caption{Primary $\eta=0.25$ small-GPT results.  The four ESP entries use the full-intersection
crossing and literal exact closure; $\dagger$ marks instead the first $e$-Bonferroni marginal crossing,
which is not an ESP global-intersection process.  ``Rejected'' and ``last'' concern strong-FWER
elementary decisions on the canonical order.  Replay summaries are the mean capped crossing round and
10th--90th percentiles over 2,000 common conditional within-stratum permutations; they hold the
selected 161-item sets fixed and do not average over which full-stratum items enter those sets.  All
displayed crossing types crossed in every replay.  For $e$-Bonferroni the replay cost is again its first
elementary crossing.}
\label{tab:llmmethods}
\begin{tabular}{@{}llrrrr@{}}
\toprule
Contrast & Procedure & Canonical crossing & Rejected & Last rejection &
Replay mean $[q_{.10},q_{.90}]$\\
\midrule
GPT-Neo & Geometric PEM     & 33 & $10/10$ & 73  & $28.85\ [25,33]$\\
         & Level-uniform PEM & 10 & $10/10$ & 73  & $8.81\ [7,11]$\\
         & Arithmetic mean   & 52 & $10/10$ & 73  & $49.40\ [45,54]$\\
         & Product           & 8  & $10/10$ & 73  & $6.61\ [5,9]$\\
         & $e$-Bonferroni    & $73^\dagger$ & $10/10$ & 113 & $64.89\ [56,73]$\\
\addlinespace
GPT-2   & Geometric PEM      & 73  & $9/10$ & 141 & $66.90\ [58,76]$\\
        & Level-uniform PEM  & 19  & $9/10$ & 141 & $21.65\ [16,28]$\\
        & Arithmetic mean    & 108 & $8/10$ & 153 & $104.64\ [95,113]$\\
        & Product            & 14  & $9/10$ & 141 & $16.03\ [11,22]$\\
        & $e$-Bonferroni     & $129^\dagger$ & $4/10$ & 153 & $119.32\ [108,129]$\\
\bottomrule
\end{tabular}
\end{table}

For GPT-Neo, the canonical geometric full-intersection accounting summed across candidate and incumbent
roles is 2,640 conditional option-continuation forward passes and 255,026 full input-sequence token
positions. For GPT-2, the corresponding totals are 5,840 passes and 565,152 token positions. In the
audited batch-one V100 runs, the associated inference-call wall times were 74.27 and 144.91 seconds.
These implementation- and hardware-specific quantities are not scheduler allocation times, monetary
costs, or computation through the last elementary rejection.

For the geometric method, the conservative least-favorable shortcut reproduced the literal
exact-closure elementary rejection time in all 20 stratum--contrast cells.  The audit was run
under the current-time rule; since the latched shortcut is no later than that rule and no earlier
than literal closure, the same equality holds a fortiori for the latched form deployed here.  This
pathwise equality is an empirical audit result for these two contrasts, not a general equality claim
for the shortcut.

Under the separately assumed i.i.d.\ paired-score sign model, terminal Holm rejects all ten strata in
both contrasts.  The exact-sign Bonferroni grid with eight interim analyses uses rounds
20, 40, 60, 80, 100, 120, 140, and 161, with local level $0.05/(10\cdot8)=0.000625$; its first rejection occurs
at round 40 for GPT-Neo and round 100 for GPT-2, and it has rejected 10/10 and 8/10 strata,
respectively, by round 161.  These sign-test comparators do not test the full-stratum
finite-population null used by the primary procedure and cannot be monitored between the eight listed analysis times.

Table~\ref{tab:llmstratumresults} gives every balanced-prefix effect and geometric rejection;
the prespecified betting-fraction sensitivity appears in Table~\ref{tab:llmeta}. The $\eta=0.50$
rows remain sensitivities and do not replace the primary $\eta=0.25$ decisions.

\begin{table}[!tbp]
\centering
\footnotesize
\setlength{\tabcolsep}{3pt}
\caption{Per-stratum balanced-prefix outcomes.  C/I gives candidate/incumbent correct counts out of
161, $\widehat\Delta_k$ is their paired accuracy difference in percentage points, and
$\tau^{\rm cl}_{k,\rm geo}$ is the primary geometric PEM closure rejection time.  A dash is a
genuine noncrossing through round 161, not the censor code 162.  The $N_k$ column is the complete
finite-population size used by the random-order factors.}
\label{tab:llmstratumresults}
\begin{tabular}{@{}l r rrr rrr@{}}
\toprule
& & \multicolumn{3}{c}{GPT-Neo 1.3B versus 125M}
& \multicolumn{3}{c}{GPT-2 large versus medium}\\
\cmidrule(lr){3-5}\cmidrule(lr){6-8}
Stratum & $N_k$ & C/I correct & $\widehat\Delta_k$ & $\tau^{\rm cl}_{k,\rm geo}$
& C/I correct & $\widehat\Delta_k$ & $\tau^{\rm cl}_{k,\rm geo}$\\
\midrule
Personal care and style & 2627 & $77/42$ & 21.74 & 73  & $70/57$ & 8.07 & 121\\
Family life             & 980  & $80/44$ & 22.36 & 64  & $76/65$ & 6.83 & --\\
Food and entertaining   & 500  & $77/45$ & 19.88 & 69  & $69/55$ & 8.70 & 139\\
Computers and electronics & 454 & $87/44$ & 26.71 & 42 & $87/71$ & 9.94 & 98\\
Health                   & 427  & $80/47$ & 20.50 & 53  & $72/57$ & 9.32 & 115\\
Home and garden          & 390  & $75/35$ & 24.84 & 44  & $68/55$ & 8.07 & 141\\
Finance and business     & 265  & $84/41$ & 26.71 & 54  & $81/65$ & 9.94 & 101\\
Pets and animals         & 255  & $76/44$ & 19.88 & 72  & $73/57$ & 9.94 & 107\\
Education and communications & 201 & $91/49$ & 26.09 & 54 & $79/63$ & 9.94 & 106\\
Youth                    & 161  & $73/42$ & 19.25 & 56  & $68/53$ & 9.32 & 97\\
\bottomrule
\end{tabular}
\end{table}

\begin{table}[!tbp]
\centering
\small
\caption{Frozen betting-fraction sensitivity.  Each entry is global crossing round / exact-closure
rejected count at round 161.  Only $\eta=0.25$ is primary; the stronger $\eta=0.50$ result is a
prespecified sensitivity and does not replace the primary decision.}
\label{tab:llmeta}
\begin{tabular}{@{}lccccc@{}}
\toprule
Contrast & $\eta$ & Geometric & Level-uniform & Arithmetic & Product\\
\midrule
GPT-Neo & 0.10 & $68/10$  & $22/10$ & $104/10$ & $14/10$\\
GPT-Neo & 0.25 & $33/10$  & $10/10$ & $52/10$  & $8/10$\\
GPT-Neo & 0.50 & $19/10$  & $5/10$  & $33/10$  & $3/10$\\
\addlinespace
GPT-2   & 0.10 & $129/2$  & $48/2$  & $147/1$  & $33/2$\\
GPT-2   & 0.25 & $73/9$   & $19/9$  & $108/8$  & $14/9$\\
GPT-2   & 0.50 & $33/10$  & $9/10$  & $57/10$  & $6/10$\\
\bottomrule
\end{tabular}
\end{table}

The common finite-population implementation diagnostic used 50,000 paths in each frozen
$K=6$, horizon-60 configuration.  The maximum empirical closure FWER across its audited null
configurations and four closure methods was $0.00406$ (binomial standard error $0.000284$), below $0.05$ but not
interpreted as a proof or as a second empirical replication: the two checkpoint contrasts carry the
same diagnostic payload by construction.

\smallskip\noindent(This appendix is cited in \S\mainnum{sec:apps} of the main article.)

\section{Basket-trial protocol and confirmatory operating characteristics}\label{app:basketresults}

\subsection{Design, metrics, and computational audit}

For independent binary cohort responses $X_{k,t}$ and the composite null $H_k:p_k\le p_0$, a
prespecified working alternative $p_{\mathrm{alt}}>p_0$ gives the per-patient factor
\begin{equation}\label{eq:basketlr}
R_{k,t}=
\left(\frac{p_{\mathrm{alt}}}{p_0}\right)^{X_{k,t}}
\left(\frac{1-p_{\mathrm{alt}}}{1-p_0}\right)^{1-X_{k,t}},
\end{equation}
whose expectation is increasing in $p_k$ and equal to one at $p_0$, so the accumulated factor is a
supermartingale throughout the null.  Independence of disjoint cohorts then validates the ESP
products; monitoring at reduced boundaries additionally requires the schedule-specific calibration
in Appendix~\ref{app:bstar}.  Shared patients, common drift, or response-adaptive
dependence would require a different joint e-process.  The single-cohort version of this problem is
treated by~\citet{baas2026curtailment}, \citet{sokolova2026}, and \citet{brannathfischer2026}; what
the basket setting adds is the multiplicity layer across cohorts and the horizon at which it must be
controlled.  The NCI grid uses 200,000 paths per configuration and each five-basket grid 50,000,
for six million audited paths in total.

The primary NCI-MATCH-style experiment has $K=10$, $p_0=0.05$, a working
alternative $p_{\mathrm{alt}}=0.25$, a horizon of 31, and nominal $\alpha=0.05$.  Its primary
family varies the number $l=0,\ldots,10$ of cohorts with response probability
0.25; additional families use homogeneous rates 0.15 and 0.40, near-null rates
0.08 and 0.10, and two heterogeneous configurations.  Except for the explicitly
labeled unequal-accrual configuration, every cohort contributes one response
per outer round until its cap.  The five-of-31 rule is included only as
unadjusted protocol context.  The exact Bonferroni comparator with prespecified interim analyses
analyzes after 10, 20, and 31 observations per cohort, with $\alpha$ allocated across cohorts and
analysis times.

The two five-basket experiments have $p_0=0.15$, a working alternative
$p_{\mathrm{alt}}=0.45$, and nominal $\alpha=0.10$.  The planned design has 13 observations
per cohort and interim analyses after 5, 9, and 13 observations; the realized design has
caps $(20,10,8,18,7)$ and common outer interim analyses at 5, 10, 15, and 20 observations, with factor-one
updates after a cohort exhausts its cap.  The Maulik--Zhou terminal Bayes rules $L(k_1,k_2)$ minimize posterior expected loss. Each false
positive is weighted by $c$ times the number of true nulls raised to $k_1$, and each false negative
by $(1-c)$ times the number of alternatives raised to $k_2$. The rules use
independent $\operatorname{Beta}(1/2,1/2)$ priors and the published values of $c$:
$(0.9831,0.9946)$ for $L(1,1),L(1,2)$ in the planned design and
$(0.9804,0.9943)$ in the realized design. Those values target terminal
weak FWER under the global null, not strong FWER over partial-null
configurations.

For a method $M$ and configuration with active set $A$ and null set $S_0$, we
define FWER when $S_0\ne\varnothing$ and the two power functionals when
$A\ne\varnothing$ by
\[
\begin{split}
\operatorname{FWER}(M)&=\Pr\{\exists k\in S_0:\ k\text{ is rejected}\},\qquad S_0\ne\varnothing,\\
\operatorname{Power}_{\rm disj}(M)&=\Pr\{\exists k\in A:\ k\text{ is rejected}\},\qquad A\ne\varnothing,\\
\operatorname{Power}_{\rm conj}(M)&=\Pr\{\forall k\in A:\ k\text{ is rejected}\},\qquad A\ne\varnothing,
\end{split}
\]
and marginal power, the mean rejection indicator over $k\in A$ for
$A\ne\varnothing$.  We display these metrics as undefined when the relevant
set is empty rather than adopt vacuous-event values.  A closure's
full-intersection time is the first crossing of the full-intersection process; it need not
itself reject an elementary hypothesis.  Write $T$ for the design's maximum number of outer rounds,
the largest cohort cap: $T=31$ for the NCI design, $13$ for the planned five-basket design and $20$
for the realized one.  It is a round count, distinct from the per-cohort caps, which differ from one
another in the realized design, and distinct from the outcome totals used by
$\bar C_{\rm active}$ below.  The displayed capped mean crossing round charges a noncrossing path
$T+1$.  For the five-basket tables,
$\bar C_{\rm active}$ is the mean total number of outcomes observed by the time of the first
rejection of a false null, charging a path with no such rejection the full
deterministic design total (65 planned and 63 realized).  The same metric is
used for the deterministic-accrual NCI configurations, with a nonrejecting path charged
the full 310 outcomes.  The active set is known only for simulation evaluation; PEM closure itself does not use it.
Thus $\bar C_{\rm active}$ is a capped operating characteristic for first
correct active-cohort rejection, not a claim of realized historical patient
savings.  Table~\ref{tab:basketsavings} reports it for every deterministic-accrual
configuration beside the prespecified interim-look Bonferroni rule.

Every FWER, disjunctive-power, and conjunctive-power indicator is one Bernoulli
observation per independent simulation path.  Where reported in the
machine-readable aggregates, their exact intervals and Monte Carlo standard
errors therefore use the path as the sampling unit.
When a marginal-power estimate pools multiple active-cohort decisions within a
path, those decisions are dependent because of closure.  The frozen
accumulators do not retain the replicate-level active-cohort fractions needed
for a clustered interval, so the aggregates omit pairwise-binomial intervals
in those rows; a single-active-cohort row remains one Bernoulli observation per
path and retains its valid path-level interval.
The deterministic-accrual time histograms identify the distribution of the capped outcome count
whose mean is $\bar C_{\rm active}$, so the mean's Monte Carlo standard error is also reported.
Likewise, common random numbers make the method comparisons paired, but the
frozen aggregates do not retain paired-difference variances; differences below
are descriptive, not significance tests.  The largest
possible Bernoulli Monte Carlo standard error is $0.00112$ with 200,000 paths
and $0.00224$ with 50,000 paths.

The exact boundary calibration and schedule checks in Appendix~\ref{app:bstar} cover every true-null
subset for full-accrual NCI and both five-basket designs.  They also validate the particular size-three true-null profile
of the NCI thinned-accrual sensitivity.  That sensitivity reports operating characteristics for its
specified configuration; the retained calibration is not certified over every possible true-null
subset under that thinning schedule.  Unverified schedule extensions require separate calibration
or the Ville threshold.  Table~\ref{tab:basketguarantees} states the monitoring and multiplicity
guarantees under these profile restrictions.

\begin{table}[!tbp]
\centering
\footnotesize
\setlength{\tabcolsep}{3pt}
\caption{The per-configuration cost comparison behind the headline $31\%$.  $\bar C_{\rm geo}$ and
$\bar C_{\rm int}$ are the mean total outcomes observed up to the first rejection of a false null,
for geometric PEM closure and for the prespecified interim-look Bonferroni rule, a path with no such
rejection being charged the full deterministic design total ($310$ for NCI, $65$ planned, $63$
realized).  ``Red.'' is $100(1-\bar C_{\rm geo}/\bar C_{\rm int})$.  The $41$ rows are every
deterministic-accrual configuration in which the cost is identifiable for both rules; the
unequal-accrual NCI configuration is excluded because its outcome total is random.  NCI labels encode the active count and response probability: for example, \texttt{l10\_p0p08}
means ten active cohorts at $p=0.08$, with any remaining cohorts null at $0.05$.
The \texttt{vep} and \texttt{ver} prefixes denote planned and realized five-basket designs,
with scenarios defined in Table~\ref{tab:basketvescenarios}; Table~\ref{tab:basketnciextra}
defines the sparse heterogeneous row. The median
reduction is $30.8\%$, the quartiles $19.7\%$ and $41.2\%$, and the range $2.5\%$ to $65.4\%$.
Common random numbers make the comparison paired, but the frozen aggregates do not retain
paired-difference variances, so these are descriptive point estimates, not significance tests.}
\label{tab:basketsavings}
\begin{tabular}{@{}lrrr@{\qquad}lrrr@{}}
\toprule
Configuration & $\bar C_{\rm geo}$ & $\bar C_{\rm int}$ & Red.\ (\%)
& Configuration & $\bar C_{\rm geo}$ & $\bar C_{\rm int}$ & Red.\ (\%)\\
\midrule
heterogeneous\_sparse & 86.3 & 135.0 & 36.1 & l8\_p0p25 & 59.3 & 114.6 & 48.3\\
l10\_p0p08 & 268.0 & 295.5 & 9.3 & l9\_p0p25 & 55.3 & 111.0 & 50.2\\
l10\_p0p1 & 223.3 & 277.9 & 19.7 & vep\_s1 & 48.4 & 54.5 & 11.2\\
l10\_p0p15 & 114.2 & 202.4 & 43.6 & vep\_s2 & 37.3 & 47.6 & 21.8\\
l10\_p0p25 & 52.0 & 108.4 & 52.0 & vep\_s3 & 29.5 & 42.7 & 30.9\\
l10\_p0p4 & 34.6 & 100.0 & 65.4 & vep\_s4 & 24.4 & 39.2 & 37.8\\
l1\_p0p15 & 276.8 & 294.9 & 6.2 & vep\_s5 & 20.6 & 36.5 & 43.5\\
l1\_p0p25 & 195.3 & 241.6 & 19.2 & vep\_s6 & 56.5 & 60.3 & 6.3\\
l1\_p0p4 & 104.1 & 150.6 & 30.8 & vep\_s7 & 42.0 & 52.9 & 20.5\\
l2\_p0p25 & 140.1 & 196.9 & 28.9 & vep\_s8 & 36.8 & 48.8 & 24.6\\
l3\_p0p08 & 298.2 & 305.7 & 2.5 & vep\_s9 & 29.0 & 43.6 & 33.6\\
l3\_p0p1 & 283.6 & 299.7 & 5.4 & ver\_s1 & 45.7 & 52.9 & 13.6\\
l3\_p0p15 & 222.2 & 268.3 & 17.2 & ver\_s2 & 36.7 & 47.7 & 22.9\\
l3\_p0p25 & 110.3 & 168.2 & 34.4 & ver\_s3 & 30.6 & 43.4 & 29.5\\
l3\_p0p4 & 57.2 & 105.7 & 45.9 & ver\_s4 & 25.5 & 39.8 & 35.8\\
l4\_p0p25 & 92.2 & 148.7 & 38.0 & ver\_s5 & 22.1 & 37.6 & 41.2\\
l5\_p0p15 & 180.1 & 245.7 & 26.7 & ver\_s6 & 54.1 & 58.7 & 7.8\\
l5\_p0p25 & 79.6 & 135.0 & 41.0 & ver\_s7 & 42.7 & 53.4 & 20.0\\
l5\_p0p4 & 44.7 & 100.8 & 55.7 & ver\_s8 & 36.9 & 48.7 & 24.1\\
l6\_p0p25 & 71.0 & 125.8 & 43.6 & ver\_s9 & 29.7 & 44.0 & 32.5\\
l7\_p0p25 & 64.4 & 119.2 & 46.0 &  & & & \\
\bottomrule
\end{tabular}
\end{table}

\begin{table}[!tbp]
\centering
\footnotesize
\caption{Guarantee taxonomy for the basket experiments.  Reduced-boundary guarantees apply through
the design horizon $T$ under the verified null-profile class described above.  Procedures on
different rows answer different monitoring and multiplicity questions; their raw power is therefore
not an interchangeable ranking.}
\label{tab:basketguarantees}
\begin{tabular}{@{}L{2.8cm}L{2.1cm}L{5.0cm}L{3.0cm}@{}}
\toprule
Method family & Information times & Multiplicity guarantee & Role\\
\midrule
Geometric and level/prior PEM closures, product closure
& Every outer round through $T$ & Strong FWER under independent streams and the verified null-profile class
& Full-support primary/sensitivities; product is the dense boundary corner\\
Arithmetic e-closure
& Every outer round through $T$ & Strong FWER under the verified independent null profiles.
The Ville version at $1/\alpha$ needs only full-filtration marginal validity, without cross-stream independence
& Dependence-robust merge at $1/\alpha$; calibrated comparator as displayed\\
$e$-Bonferroni
& Every outer round & Continuous strong FWER at the fixed threshold $K/\alpha$; no cross-stream
independence needed, and no calibration is applied to it
& Marginal anytime baseline\\
Exact Bonferroni at prespecified interim analyses
& Three/four listed analysis times & Strong FWER only at those times
& Group-sequential comparator\\
Exact Holm/Bonferroni
& Terminal & Fixed-horizon strong FWER
& Terminal comparator\\
Maulik--Zhou $L(1,1),L(1,2)$
& Terminal & Published thresholds target weak FWER under the global null
& 2026 decision-theoretic comparator\\
NCI five-of-31
& Terminal & No cross-cohort FWER guarantee
& Protocol context only\\
Matched-level/configuration oracle
& Every round & Not implementable; knows active count/subset
& Labeled reference only\\
\bottomrule
\end{tabular}
\end{table}

All raw shards were labeled confirmatory, had atomic hash-bound completion
sentinels, and came from
\auditdetail{clean Git commit
\hashid{fb53ed9c4177} (full identity in the aggregate manifests)}
{a clean frozen software revision}.  Every identity in this subsection refers to the current
calibrated run, the one behind every number reported here; the earlier flat-threshold run, cited
below only where the Ville comparison is being documented, carries different identities.
The NCI run comprises $25\times40$ shards of 5,000 paths, and each five-basket
run comprises $10\times20$ shards of 2,500 paths: 1,400 logical shards and six
million paths in total, with no missing or duplicate scenario--shard keys.
Python 3.11.9, NumPy 2.4.6, and SciPy 1.17.1 were common across shards.
The NCI, planned, and realized configuration SHA-256 prefixes are, respectively,
\hashid{4aa51f6efbde}, \hashid{9db23ee6067c}, and
\hashid{ce8e39e3781e}; the aggregate JSON files record the full digests.
Local root audits and hardened aggregation then revalidated every
sentinel, seed partition, configuration binding, accumulator conservation
identity, and environment identity.

The three reviewed aggregate JSON files have SHA-256 prefixes
\hashid{1e70b3\allowbreak{}8a6bd3}, \hashid{3e01de\allowbreak{}4308c2}, and
\hashid{a466bf\allowbreak{}526de2}.  Separate checksum manifests kept with the aggregates and
with the figures record the full aggregate and figure digests, respectively; the figure-PDF
metadata also records its generating inputs.
Because reporting was hardened after the frozen simulations, each aggregate
separately records the exact post-run validator and reporter hashes
with prefixes \hashid{14874911c017} for the aggregation module and \hashid{ffbbd6176841} for the
core simulation module; the shard Git identity remains the clean generating commit above.

One frozen NCI metadata string writes the level-uniform weight as
\texttt{1/[m choose(m,l)]}, using juxtaposition inside the brackets for
multiplication.  The executed code uses the equivalent explicit form
$w_l^{(m)}=1/\{m\binom{m}{l}\}$ throughout, as verified directly from the weight-computing
routine and its start-at-one tests.  The frozen configuration
and its hash are preserved; the aggregate records this nonfunctional notation
clarification explicitly.

\subsection{NCI-MATCH-style operating characteristics}

The homogeneous primary grid confirms the intended compromise: geometric PEM closure improves the capped
full-intersection round estimate over arithmetic in every nonnull row but the first, where
at a single active cohort the two are level to within $0.012$ rounds, and avoids the product rule's
sparse failure, while its exact-closure FWER estimate remains below the nominal level.
Table~\ref{tab:basketnciprimary} reports the complete homogeneous $p=0.25$ primary grid.
Table~\ref{tab:basketncisensitivity} gives the frozen prior and comparator sensitivities, while
Table~\ref{tab:basketnciextra} completes the nonprimary scenario grid.
Figure~\ref{fig:basketnci} displays the primary-family crossing and conjunctive-power patterns.

\begin{table}[!tbp]
\centering
\footnotesize
\setlength{\tabcolsep}{3pt}
\caption{Complete primary $p=0.25$ NCI grid, 200,000 paths per row.
Geometric PEM closure is primary.  FWER is undefined at $l=10$ because
there is no true null; power is undefined at $l=0$.  The capped crossing round is
the geometric full-intersection restricted mean with noncrossers coded 32.
Marginal-power entries are point estimates without pseudo-independent cohort
intervals.}
\label{tab:basketnciprimary}
\begin{tabular}{@{}rrrrrrr@{}}
\toprule
$l$ & Geo. FWER & Marginal power & Disjunctive & Conjunctive &
Capped crossing round & Planned-analysis power\\
\midrule
0  & $.04125$  & --        & --        & --        & $31.047$  & --\\
1  & $.04033$  & $.75270$  & $.75270$  & $.75270$  & $18.837$  & $.70505$\\
2  & $.04105$  & $.75767$  & $.94038$  & $.57496$  & $12.586$  & $.70489$\\
3  & $.04093$  & $.76663$  & $.98644$  & $.45226$  & $9.361$   & $.70487$\\
4  & $.04088$  & $.77527$  & $.99725$  & $.36425$  & $7.507$   & $.70533$\\
5  & $.04042$  & $.78979$  & $.99947$  & $.31416$  & $6.309$   & $.70593$\\
6  & $.03981$  & $.80618$  & $.99996$  & $.28468$  & $5.490$   & $.70575$\\
7  & $.03947$  & $.82156$  & $.99999$  & $.26390$  & $4.884$   & $.70519$\\
8  & $.03786$  & $.84853$  & $1.00000$ & $.30092$  & $4.425$   & $.70549$\\
9  & $.03526$  & $.86791$  & $1.00000$ & $.30910$  & $4.054$   & $.70508$\\
10 & --        & $.89537$  & $1.00000$ & $.38826$  & $3.756$   & $.70500$\\
\bottomrule
\end{tabular}
\end{table}

\begin{table}[!tbp]
\centering
\footnotesize
\setlength{\tabcolsep}{1.8pt}
\caption{Prespecified configuration-prior and comparator sensitivity in
selected rows of the primary NCI grid.  For closure methods, capped crossing
round is the full-intersection statistic with noncrossers coded 32.
Here $\dagger$ marks the first elementary $K/\alpha$
crossing, not a full-intersection statistic.  The terminal methods have no
continuously monitorable statistic.  The unadjusted five-of-31 rule's high power
comes with global-null FWER $0.16555$, not strong family-wise control.}
\label{tab:basketncisensitivity}
\begin{tabular}{@{}lrrrr rrr@{}}
\toprule
& \multicolumn{4}{c}{Marginal power} & \multicolumn{3}{c}{Capped crossing round}\\
\cmidrule(lr){2-5}\cmidrule(lr){6-8}
Method & $l=1$ & $l=3$ & $l=5$ & $l=10$ & $l=3$ & $l=5$ & $l=10$\\
\midrule
Geometric PEM (primary) & $.7527$  & $.7666$  & $.7898$  & $.8954$  & $9.361$   & $6.309$   & $3.756$\\
Level-uniform PEM       & $.6015$  & $.6439$  & $.7021$  & $.8954$  & $8.638$   & $4.545$   & $1.940$\\
Prior $q=.10$           & $.7436$  & $.7616$  & $.7907$  & $.8925$  & $9.055$   & $5.982$   & $3.447$\\
Prior $q=7/27$          & $.6554$  & $.7215$  & $.7571$  & $.8956$  & $8.311$   & $4.929$   & $2.456$\\
Prior $q=.50$           & $.4526$  & $.5570$  & $.6653$  & $.8958$  & $8.568$   & $4.381$   & $1.920$\\
Arithmetic e-closure    & $.7698$  & $.7849$  & $.8061$  & $.8874$  & $10.509$  & $7.711$   & $5.260$\\
Product closure         & $.0021$  & $.0114$  & $.0553$  & $.8829$  & $18.340$  & $6.594$   & $1.904$\\
$e$-Bonferroni          & $.7221$  & $.7227$  & $.7233$  & $.7224$  & $13.749^\dagger$ & $10.942^\dagger$ & $8.330^\dagger$\\
Bonferroni (3 interim analyses)   & $.7051$  & $.7049$  & $.7059$  & $.7050$  & --        & --        & --\\
Fixed exact Holm        & $.8237$  & $.8232$  & $.8241$  & $.8712$  & --        & --        & --\\
Unadjusted five-of-31   & $.9177$  & $.9165$  & $.9176$  & $.9172$  & --        & --        & --\\
\bottomrule
\end{tabular}
\end{table}

\begin{table}[!tbp]
\centering
\footnotesize
\setlength{\tabcolsep}{2pt}
\caption{All nonprimary NCI scenarios for the primary geometric PEM closure.
Unlisted coordinates are null at $0.05$.  The sparse heterogeneous row has
three active rates $(0.15,0.25,0.40)$.  The dense row has seven rates
$(0.15,0.20,0.25,0.35,0.40,0.25,0.20)$ and prespecified independent accrual
probabilities $(1,.95,.90,.85,.80,.90,.85)$ for them; the last three null
cohorts have accrual probabilities $(1,.95,.90)$.  The crossing metric counts
outer rounds and therefore remains defined under thinning.}
\label{tab:basketnciextra}
\begin{tabular}{@{}llrrrrr@{}}
\toprule
Family & Configuration & Geo. FWER & Marginal power & Disjunctive & Conjunctive & Capped crossing\\
\midrule
Weak $.15$      & $l=1$                  & $.03810$  & $.25266$  & $.25266$  & $.25266$  & $27.320$\\
                & $l=3$                  & $.03569$  & $.26284$  & $.59474$  & $.01953$  & $20.524$\\
                & $l=5$                  & $.03094$  & $.27936$  & $.79387$  & $.00245$  & $15.408$\\
                & $l=10$                 & --        & $.35362$  & $.97419$  & $.00175$  & $8.584$\\
\addlinespace
Strong $.40$    & $l=1$                  & $.04089$  & $.99102$  & $.99102$  & $.99102$  & $9.990$\\
                & $l=3$                  & $.04222$  & $.99203$  & $1.00000$ & $.97626$  & $4.829$\\
                & $l=5$                  & $.04372$  & $.99337$  & $1.00000$ & $.96727$  & $3.498$\\
                & $l=10$                 & --        & $.99847$  & $1.00000$ & $.98487$  & $2.336$\\
\addlinespace
Near-null $.08$ & $l=3$                  & $.02966$  & $.02798$  & $.08153$  & $.00004$  & $29.714$\\
                & $l=10$                 & --        & $.03216$  & $.27243$  & $.00000$  & $25.874$\\
Near-null $.10$ & $l=3$                  & $.03117$  & $.06719$  & $.18679$  & $.00041$  & $27.831$\\
                & $l=10$                 & --        & $.08110$  & $.54913$  & $.00001$  & $19.835$\\
\addlinespace
Heterogeneous   & sparse $(.15,.25,.40)$ & $.03911$  & $.67594$  & $.99860$  & $.20700$  & $7.581$\\
Heterogeneous   & dense, thinned accrual & $.03450$  & $.70911$  & $.99999$  & $.07235$  & $5.148$\\
\bottomrule
\end{tabular}
\end{table}

\begin{figure}[!tbp]
\centering
\includegraphics[width=\linewidth]{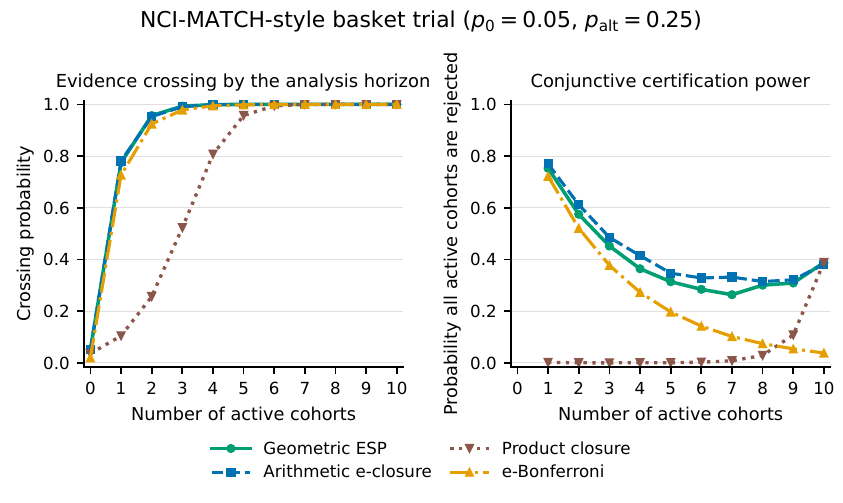}
\caption{Geometric PEM closure interpolates across active counts in the homogeneous
$p=0.25$ NCI family while retaining low global-null crossing probability. For the three
closure methods, the left panel gives the full-intersection crossing probability
by round 31; for $e$-Bonferroni it is the first elementary $K/\alpha$ crossing.
The right panel is conjunctive elementary power and is undefined at the global
null.  Each point uses 200,000 independent paths.  Marginal-power intervals
are not involved in this figure.}
\label{fig:basketnci}
\end{figure}

\subsection{Five-basket planned and realized designs}

The planned- and realized-size grids retain the same robustness pattern across their common
configurations: geometric PEM closure avoids the sparse product failure while retaining higher-order
pooling on dense paths.
Table~\ref{tab:basketvescenarios} defines the shared S0--S9 configurations.
Tables~\ref{tab:basketveplanned} and~\ref{tab:basketverealized} report the planned- and realized-size
results.

\begin{table}[!tbp]
\centering
\footnotesize
\setlength{\tabcolsep}{3pt}
\caption{Response-probability scenarios shared by both five-basket designs.
Cohort caps are $(13,13,13,13,13)$ in the planned design and
$(20,10,8,18,7)$ in the realized design.
Because S1/S6, S3/S7/S8, and S4/S9 are distinct configurations despite
having equal active counts, the figures use scenario labels.}
\label{tab:basketvescenarios}
\begin{tabular}{@{}cll@{}}
\toprule
Scenario & Response probabilities $(p_1,\ldots,p_5)$ & Active cohorts\\
\midrule
S0 & $(.15,.15,.15,.15,.15)$ & none\\
S1 & $(.45,.15,.15,.15,.15)$ & 1\\
S2 & $(.45,.45,.15,.15,.15)$ & 1--2\\
S3 & $(.45,.45,.45,.15,.15)$ & 1--3\\
S4 & $(.45,.45,.45,.45,.15)$ & 1--4\\
S5 & $(.45,.45,.45,.45,.45)$ & 1--5\\
S6 & $(.35,.15,.15,.15,.15)$ & 1\\
S7 & $(.35,.35,.35,.15,.15)$ & 1--3\\
S8 & $(.45,.35,.35,.15,.15)$ & 1--3\\
S9 & $(.45,.45,.35,.35,.15)$ & 1--4\\
\bottomrule
\end{tabular}
\end{table}

\begin{table}[!tbp]
\centering
\footnotesize
\setlength{\tabcolsep}{3.2pt}
\caption{Planned-size five-basket results, 50,000 paths per scenario.
Max partial FWER ranges over S1--S4 and S6--S9; S5 has no true null.
The Maulik--Zhou global-null values reproduce their intended weak calibration,
while their partial-null maxima illustrate that these procedures do not provide strong FWER
control.  $\bar C_{\rm active}$ charges a path with no correct active-cohort
rejection the full 65 outcomes; its Monte Carlo standard error is at most 0.091 outcomes
in the displayed cells.  Terminal methods necessarily equal 65.}
\label{tab:basketveplanned}
\begin{tabular}{@{}lrrrrrr@{}}
\toprule
& Global-null & Max partial & \multicolumn{2}{c}{S3} & \multicolumn{2}{c}{S5}\\
\cmidrule(lr){4-5}\cmidrule(lr){6-7}
Method & FWER & FWER & Marg.\ power & $\bar C_{\rm active}$ & Marg.\ power & $\bar C_{\rm active}$\\
\midrule
Geometric PEM (primary) & $.08860$  & $.08840$  & $.66983$  & $29.523$  & $.79465$  & $20.635$\\
Level-uniform PEM       & $.04158$  & $.07244$  & $.58502$  & $33.147$  & $.79336$  & $21.129$\\
Arithmetic e-closure    & $.09964$  & $.09448$  & $.67890$  & $29.976$  & $.77264$  & $22.474$\\
Product closure         & $.01224$  & $.07118$  & $.29584$  & $46.583$  & $.77706$  & $22.970$\\
$e$-Bonferroni          & $.03630$  & $.02842$  & $.49469$  & $40.636$  & $.49681$  & $33.835$\\
Bonferroni (3 interim analyses)   & $.03624$  & $.02820$  & $.45697$  & $42.745$  & $.45844$  & $36.516$\\
Fixed exact Holm        & $.03774$  & $.03024$  & $.57355$  & $65.000$  & $.66068$  & $65.000$\\
Maulik--Zhou $L(1,1)$   & $.09864$  & $.11410$  & $.77070$  & $65.000$  & $.86336$  & $65.000$\\
Maulik--Zhou $L(1,2)$   & $.09874$  & $.13828$  & $.78673$  & $65.000$  & $.90214$  & $65.000$\\
\bottomrule
\end{tabular}
\end{table}

\begin{table}[!tbp]
\centering
\footnotesize
\setlength{\tabcolsep}{3.2pt}
\caption{Realized-size five-basket results, 50,000 paths per scenario, with
the same definitions as Table~\ref{tab:basketveplanned}.  Unequal caps make
power depend on which cohorts are active, so rows S1--S9 cannot be reduced to
active-count curves.  $\bar C_{\rm active}$ charges a path with no correct
active-cohort rejection the full 63 outcomes; its Monte Carlo standard error is at most
0.085 outcomes in the displayed cells.}
\label{tab:basketverealized}
\begin{tabular}{@{}lrrrrrr@{}}
\toprule
& Global-null & Max partial & \multicolumn{2}{c}{S3} & \multicolumn{2}{c}{S5}\\
\cmidrule(lr){4-5}\cmidrule(lr){6-7}
Method & FWER & FWER & Marg.\ power & $\bar C_{\rm active}$ & Marg.\ power & $\bar C_{\rm active}$\\
\midrule
Geometric PEM (primary) & $.06422$  & $.06440$  & $.60081$  & $30.621$  & $.69896$  & $22.116$\\
Level-uniform PEM       & $.03142$  & $.05900$  & $.53077$  & $34.715$  & $.70748$  & $22.495$\\
Arithmetic e-closure    & $.06964$  & $.06584$  & $.59723$  & $31.474$  & $.68024$  & $24.103$\\
Product closure         & $.01112$  & $.05032$  & $.31283$  & $45.389$  & $.69232$  & $23.522$\\
$e$-Bonferroni          & $.03018$  & $.02178$  & $.44454$  & $39.375$  & $.43747$  & $33.386$\\
Bonferroni (4 interim analyses)    & $.02300$  & $.01796$  & $.41088$  & $43.440$  & $.41334$  & $37.639$\\
Fixed exact Holm        & $.04186$  & $.04954$  & $.60566$  & $63.000$  & $.66237$  & $63.000$\\
Maulik--Zhou $L(1,1)$   & $.09938$  & $.11934$  & $.72742$  & $63.000$  & $.81085$  & $63.000$\\
Maulik--Zhou $L(1,2)$   & $.09966$  & $.14224$  & $.74668$  & $63.000$  & $.83980$  & $63.000$\\
\bottomrule
\end{tabular}
\end{table}

Taken together, the grids show why no single empirical ranking is appropriate.
Arithmetic and $e$-Bonferroni are strong sparse comparators, product is fast in
the all-active corner but can be ineffective under sparse heterogeneity, and
the weak-FWER terminal rules often have higher raw power under their different
guarantee.  The positive evidence for the prespecified geometric PEM closure is
therefore its full-support compromise: it retains FWER control through the design horizon
under the verified independent null profiles, has a higher marginal-power point
estimate than the planned interim analyses in every one of the 42 nonnull configurations, and avoids the product
corner's sparse failure without being selected after the scenario was observed.

Calibration is what carries the marginal-power comparison from $40/42$ to $42/42$.  Under the flat threshold it stood at $40/42$, the two exceptions being the
one-active-cohort NCI configurations at response probabilities $0.25$ and $0.40$, where the geometric
estimates trailed the planned-analysis ones by $0.0018$ and $0.0008$.  Calibration reverses both, and the
reversal is entirely on the geometric side: the planned-analysis rule takes no boundary, so its column is
identical before and after, while the geometric estimates rise from $0.703255$ and $0.987570$ to
$0.752700$ and $0.991015$.

At the elementary-decision level, geometric PEM closure has a higher marginal active-cohort power point estimate than product closure in
all 42 nonnull configurations and a lower capped first-correct-rejection outcome-count point
estimate in all 41 configurations for which that cost is identifiable.  In the saturated NCI
$p=0.40$, $l=10$ configuration, the two marginal-power point estimates differ by only $0.0000060$
($0.9984730$ versus $0.9984670$).  The frozen aggregates do not retain the paired-difference
variance needed to assess this small difference; that row contributes to the descriptive count of
$42$ but does not establish superiority on its own.  Product nevertheless has the lower capped full-intersection crossing round in 17 of the
42 configurations, illustrating that full-intersection speed and exact elementary-closure performance need
not have the same ordering. All of these counts are descriptive point-estimate summaries under the
uncertainty convention above. Figures~\ref{fig:basketveplanned} and~\ref{fig:basketverealized}
display planned- and realized-size power/FWER panels for all S0--S9 scenarios.

\begin{figure}[!tbp]
\centering
\includegraphics[width=\linewidth]{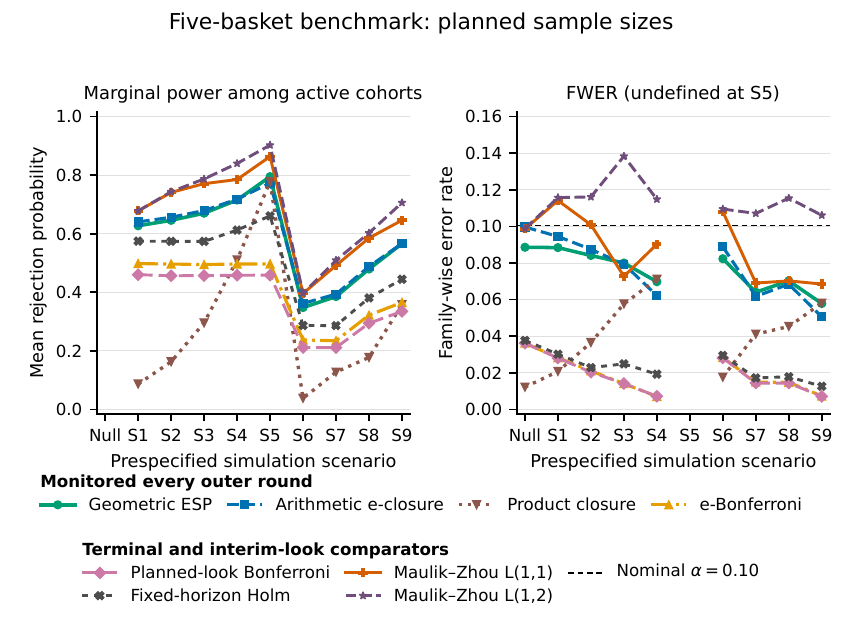}
\caption{Across the planned-size S0--S9 scenarios, geometric PEM closure avoids the product rule's sparse
power loss while its empirical FWER remains below the nominal level.  Product power is lowest
at the sparse scenarios S1 and S6 and recovers at the dense S5.
Marginal power is undefined at S0 and FWER at S5.  The dashed
horizontal line is nominal $\alpha=0.10$.  Maulik--Zhou rules are terminal and
weak-FWER calibrated; the other displayed guarantees are classified in
Table~\ref{tab:basketguarantees}.  Each point uses 50,000 paths.  Marginal
power is shown without a pseudo-independent cohort interval.}
\label{fig:basketveplanned}
\end{figure}

\begin{figure}[!tbp]
\centering
\includegraphics[width=\linewidth]{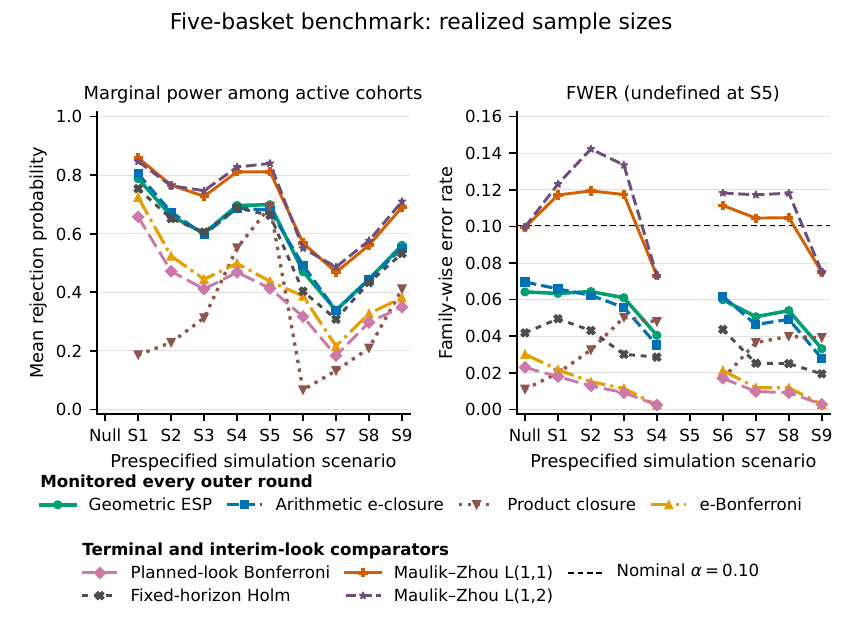}
\caption{The realized-size S0--S9 scenarios retain the same robustness pattern: geometric PEM closure
avoids the product rule's sparse power loss while its empirical FWER remains below nominal
$\alpha=0.10$.  Product power is lowest at the sparse scenarios and recovers at the dense S5.
Marginal power is undefined at S0 and FWER at S5;
the dashed horizontal line is the
nominal level. The categorical axis preserves differences among scenarios with the same active-cohort
count. Maulik--Zhou rules are terminal and weak-FWER calibrated; the other guarantees are classified
in Table~\ref{tab:basketguarantees}. Each point uses 50,000 paths, and marginal power is displayed
without a pseudo-independent cohort interval.}
\label{fig:basketverealized}
\end{figure}

\smallskip\noindent(This appendix is cited in \S\mainnum{sec:apps} of the main article.)
\clearpage

\section{Criteo uplift audit and endpoint results}\label{app:criteoresults}

\subsection{Frozen design and estimand}

Treated and control users are independently ordered within each stratum and paired without
replacement.  For paired binary outcomes $Z_{k,t}=Y^T_{k,t}-Y^C_{k,t}$, with $Y^T_{k,t}$ and
$Y^C_{k,t}$ the outcomes of the treated and control members of the $t$-th pair in stratum $k$, the
prospective model of i.i.d.\ arm samples and the null $\E Z_{k,t}\le0$ make the process with factor
$R_{k,t}=1+\eta Z_{k,t}$, $0<\eta<1$, a marginal supermartingale; the same conclusion holds under
the conditional null $\E[Z_{k,t}\mid\F_{t-1}]\le0$.  Conditional on the frozen design and arm
counts, disjoint strata and independent pairs justify the ESP products.  The ten visit hypotheses
are primary; conversion is a separate secondary family, so no joint 20-hypothesis FWER claim is
made.  The outcome-blind pseudorandom order is reproducible; the prospective interpretation assumes
i.i.d.\ arm samples conditional on the frozen design and arm counts.

The authoritative run is
\auditdetail{Phoenix job 11844386}{a frozen confirmatory cluster run} from
\auditdetail{clean Git commit
\hashid{059c69a2d6e6} (full identity in the artifact manifest)}
{a clean frozen software revision}.
Its twelve-file scientific artifact manifest has SHA-256
prefix \hashid{a32ecb072020}; the retained manifest records the full digest.
After transfer, the committed finalizer rehashed that inventory and semantically recomputed both
endpoint analyses from the paired-score records, including all marginal paths and all 1,023 exact
intersection decisions.  The \auditdetail{repository}{reproducibility archive} retains the compact
reviewed summaries and manifests in a dedicated Criteo audit directory;
the two roughly 100-MB paired-score files
remain outside version control.

No confirmatory outcome enters ranker fitting or cutpoint selection.  The corrected archive has SHA-256 prefix \hashid{2716e1bf0fd1}
(the full digest is retained in the data checksum file)
and contains $13{,}979{,}592$ rows.  Both endpoint configurations use the frozen seed
\hashid{0xcbf85be3\allowbreak{}9e9014b2}, sampled from operating-system entropy before a complete real-data
output was inspected; its decimal representation is recorded in the machine-readable
configurations.  Each endpoint has its own checked-in configuration file, which is the
authoritative machine-readable protocol.  The deterministic split produced
$2{,}795{,}900$ training, $1{,}398{,}060$ calibration, and $9{,}785{,}632$ confirmatory rows.  For
each endpoint, the transformed-outcome ridge ranker uses a ridge constant of 10 and a fixed design
propensity of $0.85$.  It is fitted on a deterministic $500{,}000$-row reservoir from the training
split, and the nine cutpoints defining ten strata are frozen from an analogous $250{,}000$-row
calibration reservoir.  Both reservoirs select the smallest fixed SplitMix64 priority keys based only on the seed, source-row index, and data role.  The observed treatment fraction in the training reservoir is $0.850502$, within the prespecified $0.01$ tolerance.  Neither reservoir is the entire corresponding split.

\subsection{Audit and endpoint results}

The largest observed contrasts are concentrated in the highest-score strata. These hold
$8.942\%$ and $9.233\%$ of the visit and conversion pairs but $74.005\%$ and $78.316\%$ of their net
paired-score totals; the overall conversion difference is $0.09709$ percentage points, and
geometric's last conversion rejection is within $0.7\%$ of the fastest exact closure.

Each endpoint yielded $1{,}467{,}725$ outcome-blind matched pairs; the larger treatment arm left
$6{,}850{,}182$ confirmatory users unmatched.  Endpoint-specific rankers mean that the same numeric
score-stratum label need not contain the same users for visit and conversion.  Independent
recomputation from the paired records reproduced the score algebra, marginal summaries, all four
global paths and all $1{,}023$ intersection decisions for every endpoint, mode and merger exactly.

Table~\ref{tab:criteostrata} reports all endpoint-specific stratum effects and primary geometric
rejection times.

\begin{table}[!tbp]
\centering
\small
\setlength{\tabcolsep}{4pt}
\caption{Audited endpoint-specific paired contrasts in the confirmatory Criteo holdout.  The quantity
$n_k$ is the number of matched pairs and $100\bar Z_k$ their treated-minus-control difference in
percentage points. The time $\tau^{\rm cl}_{k,\rm geo}$ is the geometric elementary rejection round.
Stratum 9 is the highest within-endpoint learned-score stratum.  Rejection times use exact
geometric PEM closure under the i.i.d.\ paired-score interpretation at $\alpha=0.05$ and $\eta=0.25$; a dash
denotes no rejection by the endpoint-specific terminal horizon.  The final row is descriptive,
not another elementary hypothesis.}
\label{tab:criteostrata}
\begin{tabular}{@{}lrrr@{\qquad}rrr@{}}
\toprule
& \multicolumn{3}{c}{Visit (primary)}
& \multicolumn{3}{c}{Conversion (secondary)}\\
\cmidrule(lr){2-4}\cmidrule(lr){5-7}
Score stratum & $n_k$ & $100\bar Z_k$ & $\tau^{\rm cl}_{k,\rm geo}$
              & $n_k$ & $100\bar Z_k$ & $\tau^{\rm cl}_{k,\rm geo}$\\
\midrule
0 & $149{,}350$ & $+0.10177$ & -- & $149{,}969$ & $+0.00533$ & --\\
1 & $148{,}824$ & $+0.01680$ & -- & $150{,}120$ & $+0.01599$ & --\\
2 & $150{,}589$ & $+0.02324$ & -- & $148{,}594$ & $+0.00942$ & --\\
3 & $149{,}810$ & $+0.03738$ & -- & $149{,}898$ & $+0.00133$ & --\\
4 & $147{,}293$ & $+0.06789$ & -- & $147{,}380$ & $+0.00475$ & --\\
5 & $146{,}972$ & $+0.13336$ & -- & $147{,}575$ & $+0.02101$ & $129{,}487$\\
6 & $149{,}421$ & $+0.20278$ & -- & $148{,}020$ & $+0.01486$ & --\\
7 & $148{,}536$ & $+0.42414$ & -- & $148{,}843$ & $+0.02956$ & $125{,}842$\\
8 & $145{,}683$ & $+0.75232$ & -- & $141{,}810$ & $+0.11071$ & $53{,}917$\\
9 & $131{,}247$ & $+5.62451$ & $3{,}441$ & $135{,}516$ & $+0.82352$ & $6{,}774$\\
\midrule
All matched pairs & $1{,}467{,}725$ & $+0.67962$ & --
                  & $1{,}467{,}725$ & $+0.09709$ & --\\
\bottomrule
\end{tabular}
\end{table}

The visit aggregate contains 66,023 treated and 56,048 control events (4.49832\% and 3.81870\%);
conversion contains 4,263 and 2,838 (0.29045\% and 0.19336\%).  These are matched-sample summaries
of the released, nonuniformly subsampled benchmark.  Their causal or original-advertiser-population
interpretation requires the stated randomization, i.i.d.\ sampling, and transport assumptions.

Table~\ref{tab:criteomethods} separates full-intersection crossings from exact elementary rejections
for every merger and the marginal anytime baseline.

\begin{table}[!tbp]
\centering
\footnotesize
\setlength{\tabcolsep}{3pt}
\caption{Canonical Criteo results under the i.i.d.\ paired-score interpretation.  The four mergers use
the ten-stratum global intersection and literal exact closure.  The $\dagger$ entries instead give
the first $e$-Bonferroni marginal crossing, not an ESP full-intersection process; the rejected count and last time
in those rows likewise refer to marginal $K/\alpha$ rejections rather than closed-testing decisions.
One round processes one matched pair from every nonexhausted stratum; a dash denotes noncrossing on
the complete replay.}
\label{tab:criteomethods}
\begin{tabular}{@{}llrrr@{}}
\toprule
Endpoint & Procedure & Global/first crossing & Rejected & Last rejection\\
\midrule
Visit (primary) & Geometric PEM & $3{,}436$ & $1/10$ & $3{,}441$\\
                & Level-uniform PEM & $745$ & $1/10$ & $10{,}990$\\
                & Arithmetic mean & $3{,}441$ & $1/10$ & $3{,}441$\\
                & Product & -- & $0/10$ & --\\
                & $e$-Bonferroni & $3{,}441^\dagger$ & $1/10$ & $3{,}441$\\
\addlinespace
Conversion (secondary) & Geometric PEM & $6{,}754$ & $4/10$ & $129{,}487$\\
                & Level-uniform PEM & $1{,}925$ & $4/10$ & $128{,}619$\\
                & Arithmetic mean & $6{,}776$ & $4/10$ & $145{,}065$\\
                & Product & $1{,}849$ & $4/10$ & $128{,}619$\\
                & $e$-Bonferroni & $6{,}892^\dagger$ & $2/10$ & $62{,}984$\\
\bottomrule
\end{tabular}
\end{table}

At the terminal horizon, the separately labeled one-sided paired-sign tests with Holm
correction~\citep{holm1979}
reject 7/10 visit strata and 4/10 conversion strata under their i.i.d.\ paired-sign sampling model.
This fixed-horizon comparator is not anytime valid, would not test the empirical-total null under a
finite-population interpretation, and is not used for the primary claims.  That it rejects seven
visit strata, while the anytime procedures reject one, is a further reason not to read the replay as
identifying a single active stratum.

\paragraph{Empirical product growth.} The product's behavior can be diagnosed directly
from its fixed betting fraction $\eta=0.25$. If $\widehat p_k^+$ and $\widehat p_k^-$ denote the
observed fractions of positive and negative discordant pairs, its average log increment in
stratum $k$ is
\[
\widehat p_k^+\log(1+\eta)+\widehat p_k^-\log(1-\eta).
\]
It is positive only when $\widehat p_k^+/\widehat p_k^-$ exceeds
$|\log(1-\eta)|/\log(1+\eta)\approx1.2892$. Every stratum has more treated than control wins,
yet eight fall below this log-growth threshold. The ten empirical average log increments sum to
$-0.004700$, diagnosing the product's poor growth on the observed data. These empirical
frequencies are not the unknown population probabilities, and the statistic is a betting process,
so the sparse-configuration likelihood-ratio theorem does not explain this replay directly.

\paragraph{The observed singleton gap.} The highest-score visit stratum's own marginal process
crosses $1/\alpha=20$ at round $2{,}943$.  Geometric closure, arithmetic closure and $e$-Bonferroni
all reject the stratum at round $3{,}441$, a gap of $498$ rounds from this marginal benchmark.
Each closure must wait until all $512$ intersections containing the stratum have been marked;
$e$-Bonferroni instead waits for the marginal threshold $K/\alpha=200$.  Their equality is an
observed property of this replay, not a consequence of \mainref{Lemma}{lem:ceiling}.  That lemma is
a distributional envelope under the likelihood-ratio product model; it does not make this
particular betting process's observed crossing a pathwise lower bound for arbitrary procedures.

Figure~\ref{fig:certification} collects the decision-level results from the two exact-closure
audits.  It deliberately separates full-intersection crossings from elementary rejections: an early global
crossing does not by itself support any individual claim.

\begin{figure}[!tbp]
\centering
\includegraphics[width=\linewidth]{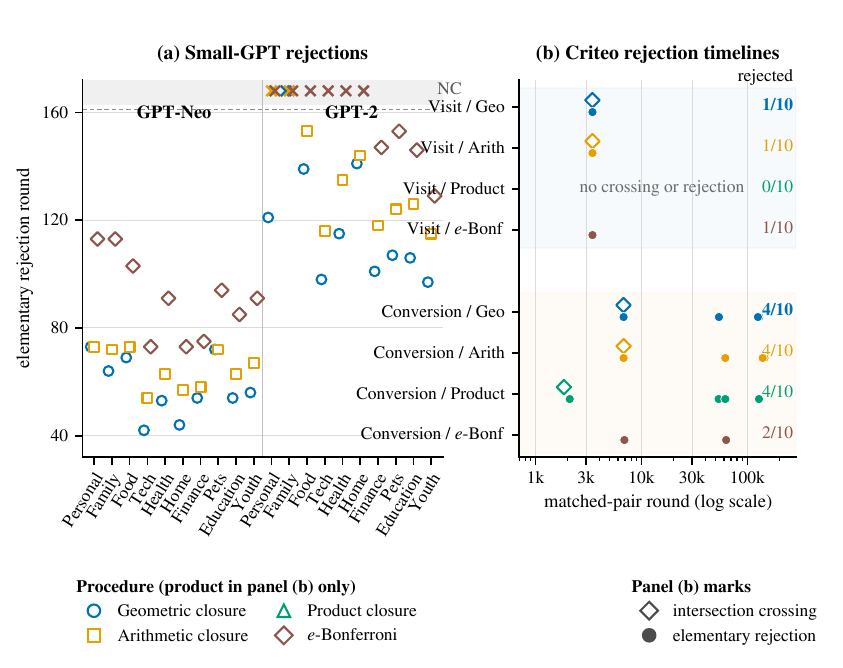}
\caption{Exact elementary rejections in the small-GPT and Criteo audits.  Panel (a) gives every
stratum-level rejection round for geometric PEM closure, arithmetic-mean closure, and $e$-Bonferroni;
crosses in the shaded strip denote nonrejection by the common round-161 horizon.  Geometric is no
later in all 20 cells, strictly earlier than arithmetic in 17 of them, and strictly earlier than
$e$-Bonferroni in 19.  Panel (b) reports the Criteo matched-pair timelines on a log scale.  Open
diamonds are full-intersection crossings and filled circles are elementary rejections; right-margin
labels give the rejected count out of ten, and geometric's last conversion rejection comes $10.7\%$
earlier in outer rounds than arithmetic's.  Conversion is a separate secondary family.}
\label{fig:certification}
\end{figure}

\smallskip\noindent(This appendix is cited in \S\mainnum{sec:apps} of the main article.)

\ifdefined\AVBLIND
\else
\FloatBarrier
\section*{Administrative statements}

\paragraph{Competing interests.} The authors declare no competing interests.

\paragraph{Use of generative AI and AI-assisted technologies.} We used generative AI tools
(ChatGPT, Google Gemini, and Claude) for language editing and code formatting support only.  All
data, results and mathematical derivations are the authors' own work.

\paragraph{Data and code availability.} The benchmark data, advertising data, and model checkpoints
used in the applications are available from the public sources cited in the article. Model and
benchmark revisions and applicable licenses, together with cryptographic identities for frozen
assets, are recorded in this supplement. Analysis code, frozen configurations, and compact audit
records are maintained in a version-controlled development repository. A snapshot archive will
accompany submission for review, and the archival release location will be recorded here once it
is chosen. A versioned reproducibility
snapshot will be deposited in a DOI-bearing repository no later than acceptance, and the article
record will be updated with that DOI\@. Large upstream datasets and model weights are not
redistributed.
\fi
\fi 

\ifavbibdone\else
\bibliographystyle{imsart-nameyear}
\IfFileExists{references.bib}{\bibliography{references}}{\bibliography{manuscript/references}}
\fi

\end{document}